\documentclass[runningheads]{llncs}
\usepackage{mycommands}

\newif\ifextended
\extendedtrue  

\newcommand{\extrefn}[2]{%
  \ifextended
    \Cref{#1}\ignorespaces
  \else
    #2\cite{this2026Extended}\ignorespaces
  \fi
}
\newcommand{\extref}[1]{\extrefn{#1}{}}

\newcommand{\AbbrExtRefn}[2]{%
  \ifextended
    \AbbrCref{#1}\ignorespaces
  \else
    #2\cite{this2026Extended}\ignorespaces
  \fi
}
\newcommand{\AbbrExtRef}[1]{\AbbrExtRefn{#1}{}}

\let\oldappendix\appendix
\renewcommand{\appendix}{\oldappendix\inappendixtrue}

\newif\ifinappendix
\inappendixfalse

\newcommand{\annotationsize}{\fontsize{8}{10}\selectfont}
\newcommand{\refproof}[1]{%
  \unless\ifinappendix
    \ifextended
      \hfill\mbox{\emph{\annotationsize{[Proof in~\AbbrCref{#1}]}}}%
    \else
      \hfill\mbox{\emph{\annotationsize{(Proof in~\cite{this2026Extended})}}}%
    \fi
  \fi
}

\newcommand{\reforig}[1]{%
  \ifinappendix
    \hfill\mbox{\emph{\annotationsize{[$\uparrow$ \hyperref[#1]{main}]}}}%
  \fi
}

\usepackage[T1]{fontenc}
\usepackage{graphicx}
\usepackage{color}
\usepackage{url}

\usepackage[utf8]{inputenc}

\newcommand{\startpara}[1]{{%
\vskip6pt\noindent
{\bf #1.}}}

\begin{document}
\title{Robust Verification of \\ Concurrent Stochastic Games}
%
\author{Angel Y. He~\orcidID{0009-0008-6481-0754} \and
David Parker~\orcidID{0000-0003-4137-8862}}
\authorrunning{A. Y. He \& D. Parker}
%
\institute{Department of Computer Science, University of Oxford, Oxford OX1 2JD, UK 
\email{angel.he@balliol.ox.ac.uk, david.parker@cs.ox.ac.uk}
}
\maketitle              
\begin{abstract}
Autonomous systems often operate in multi-agent settings and need to make concurrent, strategic decisions, typically in uncertain environments.
Verification and control problems for these systems can be tackled with concurrent stochastic games (CSGs),
but this model requires transition probabilities to be precisely specified --- an unrealistic~requirement in many real-world settings.
%
We introduce \emph{robust CSGs} and their subclass \emph{interval CSGs} (ICSGs),
which capture epistemic uncertainty about transition probabilities in CSGs.
We propose a novel framework for \emph{robust} verification of these models under worst-case assumptions about transition uncertainty. 
Specifically, we develop the underlying theoretical foundations and efficient algorithms, for finite- and infinite-horizon objectives in both zero-sum and nonzero-sum settings, the latter based on (social-welfare optimal) Nash equilibria.
We build an implementation in the PRISM-games model checker and demonstrate the feasibility of robust verification of ICSGs across a selection of large benchmarks.



\end{abstract}

\section{Introduction}
\label{sec:intro}
Autonomous and intelligent systems are increasingly deployed in environments that are \emph{nondeterministic}, \emph{stochastic} and \emph{concurrent}, such as autonomous vehicle coordination, robotic exploration and networked market interactions. In these settings, decision-making involves simultaneous strategic interactions between multiple agents, often within uncertain and dynamic environments.

\emph{Concurrent stochastic games} (CSGs)~\cite{shapley1953SGs}, also known as \emph{Markov games}, provide a powerful framework for modelling such multi-agent systems. Unlike the simpler model of \emph{turn-based} stochastic games (TSGs) \cite{condon1992complexity}, CSGs allow players to select their actions simultaneously, without knowledge of each other’s choices. The outcomes depend probabilistically on the players' joint actions.

Formal verification techniques for CSGs provide a means to establish quantitative guarantees on the behaviour of these stochastic multi-agent systems, e.g., ensuring that ``a drone can safely reach its target with at least 95\% probability, regardless of the actions of other aircraft''.
They can also be used to automatically synthesise controllers or strategies that achieve these guarantees.
Early work on these models focused on the zero-sum setting (e.g., \cite{de2007concurrent,de2001quantitative,chatterjee2013strategy}),
whilst more recent work has added support for various temporal logics and the use of nonzero-sum game-theory solution concepts such as Nash equilibria (NE) and their variants~\cite{kwiatkowska2021automatic,kwiatkowska2022correlated}, along with widely used tool support~\cite{KNPS20}.

Despite their modelling effectiveness, a limitation of CSGs is that they assume transition probabilities are \textit{precisely known}. In reality, system dynamics are often only partially known due to abstraction, modelling inaccuracies, noise, or limited data in learned statistical models.
This is particularly evident in data-driven contexts like (model-based) reinforcement learning (RL), 
where transition probabilities are estimated from data.
These issues limit the reliability of guarantees from verification
and can make synthesised strategies sub-optimal.


In recent years, there has been growing interest in principled approaches to reasoning about \emph{epistemic uncertainty} in probabilistic models for verification~\cite{badings2023decision}.
For the simpler, single-agent setting, where decision making is performed using Markov decision processes (MDPs), a well studied approach is \emph{robust MDPs} (RMDPs) \cite{iyengar2005robust,nilim2005robust,wiesemann2013robust}, which capture model uncertainty via a set of possible transition probability functions. A common subclass is \emph{interval MDPs} (IMDPs)~\cite{givan2000bounded}, where transition probabilities are bounded within intervals.
Verification techniques then provide guarantees or synthesise optimal controllers in a \emph{robust} manner, i.e., making \emph{worst-case} assumptions about model uncertainty.

However, analogous frameworks for stochastic multi-agent systems remain underdeveloped.
In this paper, we address that gap and propose \emph{robust concurrent stochastic games} (RCSGs), a novel verification framework that augments CSGs with transition uncertainty and robust solution concepts.
In fact, RMDPs already have a link to stochastic games:
they can be interpreted as a \emph{zero-sum TSG} between the agent and an adversarial \emph{nature} player that resolves uncertain transition probabilities;
this view underlies many algorithms for solving RMDPs \cite{nilim2005robust,iyengar2005robust,chatterjee2024solving,meggendorfer2025solving}.
However, extending this framework to \textit{multi-agent}, and especially \textit{concurrent} settings, is non-trivial, as the interplay between transition uncertainty and simultaneous player actions significantly complicates both the reasoning and the very definition of robustness.



\startpara{Contributions and challenges}
In this work, we develop a framework for robust verification of CSGs under adversarial transition uncertainty, covering both zero-sum and nonzero-sum settings with finite- and infinite-horizon objectives. We focus primarily on the subclass of interval CSGs (ICSGs) characterised by transition probability intervals.
Extending robustness from MDPs to concurrent multi-agent games introduces fundamental challenges: optimality requires mixed strategies; equilibria definitions must incorporate uncertainty resolutions; and, in the nonzero-sum case, the adversarial role of nature differs from standard best-response reasoning. 
To address these, we: introduce robust equilibrium notions; establish theoretical results, e.g., on value preservation under player/nature action ordering; and present novel reductions from ICSGs to (non-robust) CSGs by adding an adversarial nature player. The latter yields a 2-player game in the zero-sum case, and a more subtle 3-player construction in the nonzero-sum case where nature minimises social welfare.
Building on these results, we derive tractable algorithms for solving ICSGs and implement them in PRISM-games~\cite{KNPS20}. We show their practicality via empirical evaluations on a set of large benchmarks: verification for zero-sum ICSGs performs comparably to CSGs, while nonzero-sum methods scale effectively but also provide insights into the intrinsic challenges of robust multi-agent reasoning.
\ifextended
This paper is an extended version of~\cite{HP26}.
\else
\fi

\startpara{Related work}
In the \emph{single-agent} setting, RMDPs
are solvable via robust dynamic programming (RDP) \cite{iyengar2005robust,nilim2005robust,wiesemann2013robust}. Recent work develops generic algorithms for polytopic \cite{chatterjee2024solving,wang2023robust} and more general RMDPs with constant support (e.g., \cite{meggendorfer2025solving}) via reduction to TSGs.
Robust methods for multi-agent settings are more limited, restricted to \textit{turn-based} polytopal stochastic games \cite{castro2025polytopal}, \textit{qualitative} verification \cite{berthon2025robustmulti}; or \textit{sampling-based}, learning-driven algorithms (e.g., \cite{shi2024sample,farhat2025online,roch2025drmg}) for the similar problem of distributionally \emph{robust Markov games} \cite{littman1994markovgames,shapley1953SGs,filar2012competitive} in RL, without model-checking capabilities or verification guarantees. 

\vspace{-0.15cm}
\section{Preliminaries}
\vspace{-0.1cm}
\label{sec:prelim}
Let $\distr(X)$ denote the set of discrete probability distributions over a finite set $X$, and 
let $\ind[A]$ be the indicator function that equals $1$ if $A$ holds and $0$ otherwise.
\vspace{-0.1cm}

\subsection{Robust Markov Decision Processes}
\label{sec:prelim-RMDP}
A core model for verification and control tasks in the context of uncertainty is \emph{Markov decision processes (MDPs)} \cite{bellman1966dynamic,howard1960dynamic}.
\vspace*{0.3em}
\begin{definition}[MDP]
\label{def:MDP}
A \emph{Markov decision process} (MDP) is a tuple $M = (S, \bar{s}, A, P)$, where 
$S$ is a finite set of states with initial state $\bar{s} \in S$;
$A$ is a finite set of actions; and 
$P: S\times A \rightharpoonup \distr(S)$ is a probabilistic transition function.
\end{definition}
\vspace*{-0.7em}
$A(s)$ denotes the enabled actions in state $s$. We write $P_{sa}=P(s,a)$ for the next-state distribution at $s$ under action $a$ and $P_{sas'}=P_{sa}(s')$ for the corresponding transition probability to $s'$. 
A \emph{path} is a finite or infinite sequence $\smash{\pi = s_0 \xrightarrow{a_0} s_1 \xrightarrow{a_1} \ldots}$ such that $s_0=\sbar$, $a_i \in A(s_i)$ and $P_{s_i a_i s_{i+1}}>0$ for all $i$. 
We write $\pi(i)=s_i$, $\pi[i]=a_i$, and let $FPaths_M$ and $IPaths_M$ be the sets of finite and infinite paths in $M$, respectively. 

A \emph{strategy} (or policy) of $M$ is a function $\sigma: FPaths_M \rightarrow \distr(A)$
that resolves the choices of action in each state.
Typically, we aim to find an optimal strategy for an MDP,
e.g., one that maximises the probability of a target state set being reached or the expected value of some reward function.

In order to reason about MDPs \emph{robustly} in the context of (epistemic) uncertainty about the model itself,
we can use \emph{robust MDPs}.
\vspace*{0.3em}
\begin{definition}[RMDP]
\label{def:RMDP}
    A \emph{robust MDP} (RMDP) is a tuple $M_R =(S, \bar{s}, A, \Punc)$, where $S$, $\bar{s}$ and $A$ are as for MDPs (\Cref{def:MDP}), and 
    $\Punc: S\times A \rightharpoonup 2^{\distr(S)}$ is an uncertain probabilistic transition function. 
\end{definition}
\vspace*{-0.7em}
Intuitively, an RMDP captures unknown transition dynamics:
for each state $s$ and action $a\in A(s)$, the \emph{uncertainty set} $\Punc_{sa} = \Punc(s,a)$
represents the \emph{set} of possible next-state distributions.
Selecting a single $P_{sa}\in\Punc_{sa}$ for each $(s,a)$ yields a probabilistic transition function $P:S\times A\rightarrow\distr(S)$, giving a specific MDP.
Abusing notation slightly, we also treat $\Punc$ as a set and write $P\in\Punc$,
referring to each $P$ as an \emph{uncertainty resolution}.
Typically, we aim to find a \emph{robust~optimal} strategy,
i.e., one that is optimal against the worst-case uncertainty resolution.

An RMDP $M_R$ can be viewed as a zero-sum TSG, i.e., a game which alternates between an agent choosing an $a\in A(s)$ in each state $s$ and then an adversarial player \emph{nature} resolving the choices $P_{sa}\in\Punc_{sa}$. 
Assumptions or restrictions on the strategies for nature dictate the kind of uncertainty considered:
\begin{enumerate*}[label=\arabic*)]
    \item \emph{rectangularity} \cite{nilim2005robust,iyengar2005robust}, i.e., whether transition uncertainty is resolved independently across different states (\emph{$s$-rectangular}) or state-action pairs (\emph{$(s,a)$-rectangular});
    \item \emph{static} (\emph{stationary}) vs. \emph{dynamic} (\emph{time-varying}) semantics \cite{iyengar2005robust,nilim2005robust}, i.e., whether nature follows a \emph{memoryless} strategy that has to make consistent choices at each $(s,a)$ over time. 
\end{enumerate*}
We can also restrict the nature of the uncertainty sets $\Punc_{sa}$, notably whether they are polytopic,
i.e., next-state distributions in $\Punc_{sa}$ form a polytope.
In this work, we focus on $(s,a)$-rectangular, polytopic uncertainty. This setting includes the well-studied class of IMDPs \cite{nilim2005robust,givan2000bounded}, where each $\Punc_{sa}$ is defined by independent intervals over transition probabilities.

\vspace{-0.1cm}
\subsection{Concurrent Stochastic Games}
\label{sec:prelim-CSG}

CSGs \cite{shapley1953SGs} provide the semantic basis for the class of games we introduce. 

\begin{definition}[CSG]
\label{def:csg}
An ($n$-player) \emph{concurrent stochastic game} (CSG) is a tuple $\csg=(N,S,\sbar, A, \Delta, P)$
where $N = \{1, \ldots, n\}$ is a finite set of players; $S$, $\sbar\in S$ and $P:S\times A\rightarrow\distr(S)$ are as defined for an MDP (\Cref{def:MDP});
$A= \times_{i\in N}{(A_i \union \idleset)}$ where $A_i$ is the set of actions for player $i$ and $\idle$ is an idle action disjoint from $\union_{i\in N}{A_i}$; 
$\Delta: S \rightarrow 2^{\union_{i\in N}{A_i}}$ is an action assignment function.
\end{definition}
\vspace*{-0.9em}
A CSG $\csg$ begins in the initial state $\sbar$. At each state $s \in S$, each player ${i \in N}$ simultaneously selects an action $a_i\in A_i(s)$, where $A_i(s) = \Delta(s) \intersect A_i$ if ${\Delta(s) \intersect A_i} \\{\neq \emptyset}$ and $A_i(s) =\idleset$ otherwise.
The game then transitions to state $s'$ following the distribution $P_{sa}$, where $a=(a_1,\dots,a_n)\in A(s) := \bigtimes_{i\in N}{A_i(s)}$. 

To allow quantitative analysis of $\csg$, we augment CSGs with \emph{reward structures}.
\vspace{-1em}
\begin{definition}[Reward structure]
\label{def:reward-struct}
    A \emph{reward structure} for a CSG $\csg$ is a tuple $r=(r_A,r_S)$ where $r_A: S\times A \rightarrow \reals$ is the \emph{action reward} function,
    and $r_S: S \rightarrow \reals$ is the \emph{state reward} function.
    We denote the total reward associated with a state-action pair $(s,a)$ as $r_{sa}=r(s,a) := r_A(s,a)+r_S(s)$.
\end{definition} 
\vspace*{-0.8em}

A \emph{strategy} for player $i$ is a function $\sigma_i: FPaths_{\csg}\to\distr(A_i)$ mapping finite histories to distributions over actions. 
A \emph{strategy profile} (or just \emph{profile}) is a tuple of strategies for each player, denoted $\sigma=(\sigma_1, \ldots, \sigma_n) \in \Sigma := \bigtimes_{i\in N}{\Sigma_i}$.
An \emph{objective} (or utility function) of player $i$ is a random variable $X_i: \ipathg \rightarrow \reals$. 
In a \emph{zero-sum} game, which is 2-player by definition, players have directly opposing objectives, i.e., $X_1 = -X_2$. In this case, we will represent their objectives using a single variable $X := X_1$, so that $X_2=-X$. In the \emph{nonzero-sum} (or \emph{general-sum}) case, we write $X = (X_1,\dots,X_n)$ for the tuple of all player objectives.

In this paper we focus on the four common objectives below, two of which are finite-horizon and two infinite-horizon. We assume a set of target states $T\subseteq S$ and, for the finite-horizon case, a time horizon $k\in\nats$.
\vspace{-.3em}
\begin{itemize}
\label{lst:objs}
    \item \emph{Bounded probabilistic reachability}: $X(\pi) = \ind\left[\exists{j \leq k}. \ {\pi(j)\in T}\right]$;
    \item \emph{Bounded cumulative reward}:
    $X(\pi) = \sum_{i=0}^{k-1}{r(\pi(i),\pi[i])}$;
    \item \emph{Probabilistic reachability}: 
    $X(\pi) = \ind\left[ \exists{j \in \nats}. \ {\pi(j)\in T} \right]$; and 
    \item \emph{Reachability reward}:
    $X(\pi) = \sum_{i=0}^{k_{\min}-1}{r(\pi(i),\pi[i])}$ if $\exists{j \in \nats}. \ {\pi(j)\in T}$ and $X(\pi) =\infty$ otherwise, 
    where $k_{\min} = \min{\{ j\in \nats \mid \pi(j)\in T\}}$.
\end{itemize}
We denote the \emph{expected utility} of player $i$ from state $s$ under profile $\sigma$ in $\csg$ as $u_i(\sigma \mid s,X) :=  \valg^i(s \mid \sigma, X) := \stratgs{\ev}[X_i]$,
with the index $i$ omitted in the zero-sum case. 
In zero-sum games, the \emph{value} of $\csg$ with respect to $X$ exists if the game is \emph{determined}, i.e., the maximum payoff that player~1 can guarantee equals the minimum payoff player~2 can enforce; the corresponding strategies are said to be \emph{optimal}.
For nonzero-sum games where players may cooperate or compete, we use the concept of a \emph{Nash equilibrium} (NE): a profile in which no player can improve their utility by unilaterally deviating. A \emph{social-welfare optimal NE} (SWNE) \cite{kwiatkowska2021automatic} refers to an NE that also maximises the players' total utility. 

A special, degenerate ``one-shot'' case of a CSG is a \emph{normal form game} (NFG), which consists of a single state and a single decision round. Thus, an NFG can be represented as a simplified tuple $\nfg=(N,A,u)$, where $N$ and $A$ are as defined for a CSG, and $u=(u_1,\ldots,u_n)$ with $u_i: A\to \reals$ defining player $i$'s utility for each joint action. 
A 2-player NFG can be represented as a \emph{bimatrix game}, defined by two matrices $\Zi_1, \Zi_2 \in \reals^{l \times m}$ with entries $z^1_{ij} = u_1(a_i, b_j)$ and $z^2_{ij} = u_2(a_i, b_j)$, where $A_1 = \{a_1, \ldots, a_l\}$ and $A_2 = \{b_1, \ldots, b_m\}$. 
The game is called \emph{zero-sum} if $\forall{a \in A}. \ {u_1(a) + u_2(a) = 0}$, in which case we can represent it as a single \emph{matrix game} $\Zi \in \reals^{l \times m}$ with $z_{ij} = u_1(a_i, b_j) = -u_2(a_i, b_j)$, i.e., $\Zi = \Zi_1 = -\Zi_2$. 

\vspace{-0.15cm}
\section{Robust Concurrent Stochastic Games}
\vspace{-0.1cm}
\label{sec:RCSG}
We now propose the model of \emph{robust CSGs} (RCSGs), which unifies the notions of \emph{robustness} from RMDPs and \emph{concurrent} decision-making from CSGs. 

\begin{definition}[RCSG]
\label{def:rcsg}
    A \emph{robust CSG} (RCSG) is a tuple $\rcsg = (N,S,\sbar,A,\Delta,\Punc)$ 
    where $\Punc$ is an \emph{uncertain} transition function defined as for RMDPs in \Cref{def:RMDP}, and all other components are as defined for CSGs in \Cref{def:csg}. 
\end{definition}
\vspace{-0.35cm}
Similar to the way that fixing the transition function in an RMDP induces an MDP, fixing the transition function in an RCSG to a particular $P\in\Punc$ induces a CSG $\rcsg_P= (N, S, \sbar, A, \Delta, P)$. 
We parametrise the corresponding notation with $P$. Notably, for a state $s$ of $\rcsg$ and a strategy profile $\sigma$ (defined as for CSGs), we write $u_i(\sigma,P \mid s,X) :=  \gs{\ev}^{\sigma,P}[X]$\footnote{We use these interchangeably and omit parameters that are clear from the context.} for the expected value of an objective $X$ under $\sigma$ applied to $\rcsg_P$.

By contrast to the uncertainty semantics in RMDPs (see \Cref{sec:prelim-RMDP}), multi-player RCSGs introduce an additional dimension: whether uncertainty is resolved \emph{adversarially} or is \emph{controlled} by players \cite{castro2025polytopal}.
In the adversarial case, nature resolves uncertainty against the players: in zero-sum games, where players have directly opposing objectives, nature aligns with one player to minimise the other's payoff; in nonzero-sum games, it acts against both by minimising a joint objective such as social welfare or cost. 

The controlled case assumes that one or more players resolve uncertainty to optimise their own objectives. This corresponds to optimistic reasoning in single-agent or zero-sum settings and can be seen as the dual of the adversarial case.
However, in nonzero-sum games, assigning control to a single player can undermine fairness by attributing uncertainty to that player’s decisions, while shared control would require principled coordination among players.
We therefore focus on the adversarial resolution in both zero- and nonzero-sum games.

Next, we define the \emph{robust} analogue of several game-theoretic concepts in the context of RCSGs. 
In general, we enforce that their defining properties hold under every $P\in \Punc$, or equivalently under the worst case $P^* := \arg\min_{P\in\Punc}{\gs{\ev}^{\sigma,P}[X]}$. 

\startpara{Zero-sum RCSGs}
\label{sec:zs-RCSG}
We first adapt classical minimax concepts for (2-player) zero-sum games, assuming that player 1 maximises an objective $X$.
\vspace*{0.2em}
\begin{definition}[Robust determinacy and optimality]
\label{def:adv-robust-determinacy-optimality-value}
    A zero-sum RCSG $\csg$ is \emph{robustly determined} with respect to an objective $X$, if for any state $s \in S$:
    \vspace{-0.1cm}
    \[
    \sup_{\sigma_1\in \Sigma_1} \inf_{\sigma_2\in \Sigma_2} \inf_{P\in \Punc}{\gs{\ev}^{(\sigma_1, \sigma_2), P}[X]}
    = 
    \inf_{\sigma_2\in \Sigma_2} \sup_{\sigma_1\in \Sigma_1} \inf_{P\in \Punc}{\gs{\ev}^{(\sigma_1, \sigma_2), P}[X]}
    =: \valg(s,X)
    \vspace{-0.1cm}
    \]
    where we call $\valg(s,X)$ the \emph{robust value} of $\csg$ in $s$ with respect to $X$.~Also,~$\sigma_1^*\in \Sigma_1$ is~a \emph{robust optimal strategy} of player 1 with respect to $X$ if $\gs{\ev}^{(\sigma_1^*,\sigma_2),P}[X] \geq \valg(s,X)$ for all $s\in S, \sigma_2 \in \Sigma_2, P\in \Punc$; similarly $\sigma_2^* \in \Sigma_2$ is a~\emph{robust optimal strategy} of player 2 
    if ${\gs{\ev}^{(\sigma_1,\sigma_2^*),P}[X] \leq \valg(s,X)}$ for all~${s\in S, \sigma_1 \in \Sigma_1, P\in \Punc}$.
\end{definition}
\vspace{-0.45cm}

\startpara{Nonzero-sum RCSGs}
\label{sec:nz-RCSG}
In the nonzero-sum case, each player $i\in N$ has a distinct objective $X_i$. 
For this setting, we adopt the concept of a \textit{robust Nash equilibrium} (RNE) \cite{aghassi2006robust,klibano1993uncertainty,perchet2020finding}, which refers to a profile $\sigma^*$ that remains a Nash equilibrium under any uncertainty resolution. 
Note that, in the zero-sum case, RNE coincide with the notion of robust optimal strategies. 

As is common for CSGs, we use \emph{subgame-perfect} NE \cite{osborne2004intro},
which require equilibrium behaviour in every state of the game, not just the initial one. 
We call these \emph{subgame-perfect RNE} but, for brevity, often refer to them simply as~RNE.
%
In standard CSGs with infinite-horizon objectives, NE may not exist \cite{bouyer2014mixed}, but $\varepsilon$-NE do exist for any $\varepsilon > 0$ under the objectives we consider. We therefore work with \emph{subgame-perfect $\varepsilon$-RNE} for infinite-horizon properties.
\begin{definition}[Subgame-perfect $\varepsilon$-RNE]
\label{def:subgame-perfect-rne}
    A profile $\sigma^*$ is a \emph{subgame-perfect robust $\varepsilon$-NE} ($\varepsilon$-RNE) iff ${\varepsilon + \inf_{P\in \Punc}{\left[ u_i(\sigma_{-i}^*[\sigma_i^*], P) - u_i(\sigma_{-i}^*[\sigma_i], P) \right]} \geq 0}$ 
    for all $\sigma_i \in \Sigma_i, i\in N$ at \emph{every state} $s\in S$.
    We define $\left\langle \inf_{P\in \Punc}{u_i(\sigma^*, P)} \right\rangle_{i\in N}$ as the corresponding \emph{$\varepsilon$-RNE values}. 
    A \emph{subgame-perfect robust NE} (RNE) is an $\varepsilon$-RNE with $\varepsilon=0$. We write $\epsRNE$ for the set of all $\varepsilon$-RNE and $\RNE$ for the set of all RNE.
\end{definition}
\vspace*{-0.7em}
Even if all induced CSGs $\rcsg_P$ of an RCSG $\rcsg$ have an NE ($\varepsilon$-NE), there may not exist an RNE ($\varepsilon$-RNE).
This is because a profile that is an NE in one induced CSG may not be an NE across all others. See \extrefn{sec:RNE-existence-in-ICSG}{the extended version of this paper}~for an illustration.

Next, we propose the robust counterparts of SWNE \cite{roughgarden2010algorithmic,kwiatkowska2021automatic} as RNE that maximise the robust (worst-case) total utility of the players, denoted $\usum(\sigma,P):=\sum_{i\in N}{u_i(\sigma, P)}$ for a given profile $\sigma\in \Sigma$ and $P\in \Punc$.

\begin{definition}[RSWNE]
\label{def:rsw-rsc}
    An RNE $\sigma^*$ of $\rcsg$ is a \emph{robust social-welfare optimal NE} (RSWNE) if it maximises the \emph{robust social welfare} amongst all RNE, i.e., 
    $\sigma^* \in \arg\max_{\sigma\in \RNE}{\usum(\sigma, P^*_\sigma)}$ where $P^*_\sigma := \arg\inf_{P\in \Punc}{\usum(\sigma,P)}$. 
    We define $\left\langle {u_i(\sigma^*, P^*_{\sigma^*})} \right\rangle_{i\in N}$ as the corresponding \emph{RSWNE values}.
\end{definition}
\vspace*{-0.7em}

Like RMDPs, various uncertainty models are applicable in RCSGs, such as those characterised by $L^p$-balls \cite{strehl2004empirical,ho2018fast} and non-rectangular sets. 
However, value computation is often computationally intractable under these models, even in single-agent settings \cite{wiesemann2013robust}. By contrast, \emph{interval} uncertainty yields convex uncertainty sets, enabling tractable computation while effectively capturing bounded but unstructured estimation errors, e.g., those derived from confidence intervals \cite{strehl2005theoretical}.
Hence, for the remainder of the paper we focus on \emph{interval CSGs}, as a natural and scalable foundation for incorporating robustness into CSGs.
\vspace*{0.3em}
\begin{definition}[ICSG]
\label{def:ICSG}
    An \emph{interval CSG} (ICSG) is a tuple $\rcsg=(N,S,\sbar, A, \Delta,\\\Pcheck, \Phat)$ 
    where 
    $\Pcheck, \Phat: S\times A\times S \rightharpoonup [0,1]$ are partial functions that assign lower and upper bounds, respectively, to transition probabilities, such that $\Pcheck_{sas'} \leq \Phat_{sas'}$.
    All other components are defined as for CSGs (\Cref{def:csg}).
\end{definition}
\vspace*{-0.7em}
An ICSG is an RCSG where $\Punc_{sa} = \{ P_{sa}\in \distr(S)
\mid \forall{s' \in S}. \ {P_{sas'} \in [\Pcheck_{sas'}, \Phat_{sas'}]} \}$.
We also require that $\Pcheck_{sas'}=0 \iff \Phat_{sas'}=0$,
i.e., each transition is either excluded or assigned a non-degenerate interval with a strictly positive lower bound.
This enables the standard \emph{graph preservation} constraint \cite{chatterjee2008model,meggendorfer2025solving}, which requires that all $P \in \Punc$ share the same support. This property is essential for ensuring the tractability of RDP \cite{nilim2005robust,iyengar2005robust} over $(s,a)$-rectangular uncertainty models.

\vspace{-0.15cm}
\section{Zero-sum ICSGs}
\vspace{-0.1cm}
\label{sec:zs-ICSG}
We now establish the theoretical foundations for robust verification of ICSGs, starting with the \emph{zero-sum} case.
We fix an ICSG $\csg = (N, S, \sbar, A, \Delta, \Pcheck, \Phat)$ where~${N = \{1,2\}}$
and in which player 1 maximises an objective $X$. For now, we assume that $X$ is infinite-horizon: either probabilistic/reward reachability.

At a high level, analogously to the stochastic game view of an RMDP, we will reduce ICSG $\csg$ to a CSG $\adv{\csg}$ extended with a third player, \emph{nature}, who resolves transition uncertainty adversarially against player~1. Since player 2 and 3 (nature) share the same objective, they can be merged into a single coalition, making $\adv{\csg}$ a 2-player CSG between coalitions $\{1\}$ and $\{2,3\}$.
We refer to $\adv{\csg}$ as the \emph{adversarial expansion} of $\csg$.
We will establish a one-to-one correspondence between optimal values and strategies in~$\adv{\csg}$ and their robust counterparts in~$\csg$, allowing us to reduce robust verification of zero-sum ICSGs to solution of zero-sum CSGs. The latter can be performed with value iteration~\cite{kwiatkowska2021automatic} although, as for RMDPs, explicit construction of the full CSG $\adv{\csg}$ is not required. 

\startpara{Player-first vs. nature-first semantics}
In the CSG reduction, a natural question to consider is the order in which uncertainty is resolved relative to players' moves.
Under the \emph{player-first} semantics, nature acts \emph{after} both players have chosen their actions. 
While this aligns closely with the adversarial interpretation of robustness (see \Cref{def:adv-robust-determinacy-optimality-value}), it requires nature's minimisation problem (against player 1's objective) to be solved separately for every player profile $\sigma\in\Sigma$, which is computationally demanding.
By contrast, the \emph{nature-first} semantics assumes that nature first commits to a realisation of $\Punc$ \emph{before} any player acts, thereby inducing a fixed CSG upfront.
This formulation allows the use of efficient dynamic programming techniques, such as \emph{robust value iteration} (RVI) \cite{iyengar2005robust,nilim2005robust}, which we adopt for solving these games.

This distinction corresponds to the difference between \emph{agent first} and \emph{nature first} semantics for robust partially observable MDPs in \cite{bovy2024imprecise}. While in general~this assumption can affect the game value, we establish in \Cref{thm:infh-nature-or-player-first-val-equiv} that both semantics yield the same value in our setting (finitely-branching zero-sum ICSGs).

\begin{restatable}[Player/nature-first Value Equivalence]{theorem}{PlayerNatureEq}
\label{thm:infh-nature-or-player-first-val-equiv}
From any $s \in S$, $\valg(s)$ is invariant under the player-first or nature-first semantics: 
\[
\sup_{\sigma_1\in \Sigma_1}{\inf_{\sigma_2\in \Sigma_2}{\inf_{P\in \Punc}{\gs{\ev}^{(\sigma_1, \sigma_2), P}[X]}}}
=
\inf_{P\in \Punc}{\sup_{\sigma_1\in \Sigma_1}{\inf_{\sigma_2\in \Sigma_2}{\gs{\ev}^{(\sigma_1, \sigma_2), P}[X]}}}.
\quad \reforig{thm:infh-nature-or-player-first-val-equiv}
\]
\end{restatable}
\vspace*{-0.7em}
\begin{proof}[Sketch]
We prove the result top-down via construction of the \emph{adversarial expansion} $\adv{\csg}$ (\Cref{def:infh-adv-resolution}) and subsequent determinacy and value preservation results (\Cref{cor:infh-adv-value-preservation}) established in this section. Specifically, determinacy of the finite CSG $\adv{\csg}$ justifies exchanging the order of optimisation between players and nature without changing the game value. Full proof in \extref{prf:infh-nature-or-player-first-val-equiv}.
\end{proof}

This value equivalence justifies using the nature-first semantics in implementations, so that nature's minimisation problem is solved only \emph{once} per step, after which player strategies are derived from the minimising distributions $P^*_{sa}$. 
%
Henceforth, without loss of generality, we focus on the player-first semantics.

We next observe that optimal strategies for ICSGs with infinite-horizon objectives admit a memoryless form, which allows the game to be analysed via fixed-point equations over the state space.
\begin{restatable}[Strategy class sufficient for optimality]{lemma}{MemorylessSufficiency}
\label{lem:infh-memoryless-optimal-strategies}
   Given an \emph{infinite-horizon} objective $X$ for $\csg$, each player has a \emph{memoryless} robust optimal strategy, and nature has a \emph{deterministic memoryless} optimal strategy.
   \reforig{lem:infh-memoryless-optimal-strategies}
\end{restatable}
\vspace*{-0.7em}
\begin{proof}[Sketch]
    Under $(s,a)$-rectangularity, nature's optimal resolution and players' continuation values depend only on the current state, so histories ending in the same state can be collapsed. Further, nature's independent choices across state–action pairs define a single transition function, so nature \emph{deterministically} commits to one such function. 
    Full proof in \extref{prf:infh-memoryless-optimal-strategies}.
\end{proof}
Henceforth in the zero-sum setting, we let $\Sigma_i$ denote the set of \emph{memoryless} strategies for each player $i\in \{1,2\}$, and interpret $\Punc$ as the set of transition functions resulting from \emph{memoryless} nature strategies.

Using also the $(s,a)$-rectangularity of ICSGs, the \emph{robust Bellman equation} for $\csg$ (proved in \extref{sec:zs-ICSG-supp}) is given by:
\vspace{-0.2cm}
\begin{align}
V(s) 
&= \sup_{\sigma_1\in \distr(A_1(s))} \inf_{\sigma_2\in \distr(A_2(s))}
\left\{ r_{s}^{\sigma} +
    \sum_{a\in A(s)}{\sigma_{sa} \inf_{P_{sa}\in \Punc_{sa}}{\sum_{s'\in S}{P_{sas'} \cdot V(s')}}}
\right\}
\label{eq:infh-adv-bellman}
\end{align}
where $\sigma_{sa} = \sigma_1(s,a_1) \sigma_2(s,a_2)$ with $a=(a_1,a_2)$ and $r_{s}^{\sigma} = \sum_{a\in A(s)}{\sigma_{sa} r_{sa}}$.

This characterises the fixed-point that our game solving algorithms will later compute iteratively.
We remark that, unlike the standard Bellman equation for CSGs, \Cref{eq:infh-adv-bellman} includes nature's \emph{inner problem} $\inf_{P_{sa} \in \Punc_{sa}}$, which captures transition uncertainty and is solved using a greedy algorithm adapted from the IMDP setting~\cite{nilim2005robust}; details are provided in \extref{sec:infh-rvi-solve-inner}. 
Furthermore, whereas the Bellman equations for TSGs and RMDPs~\cite{nilim2005robust,iyengar2005robust} optimise over \emph{deterministic} player strategies by selecting pure actions, here we optimise over \emph{randomised} (memoryless) strategies, i.e., distributions over actions. This reflects the added complexity of concurrent multi-agent interaction.

\startpara{The CSG reduction}
\label{sec:infh-adv-resolution}
We now formally define the \emph{adversarial expansion} $\adv{\csg}$ for ICSG $\csg$, which is a CSG containing intermediate states representing state–action pairs of $\csg$. 
In the following, we use the operator $\adv{\cdot}$ to denote a structure associated with $\adv{\csg}$. We write $\verts(K)$ for the vertices of a polytope $K$ and use notation such as ``$\adv{s}=s \in S$'' as shorthand for ``$\exists{s\in S}.\,{\adv{s}=s}$''.

\begin{definition}[Adversarial expansion]
\label{def:infh-adv-resolution}
    We define the \emph{adversarial expansion} of ICSG $\csg$ as a 2-player CSG $\adv{\csg} = (\{1,2\}, \adv{S}, \sbar, \adv{A}, \adv{\Delta}, \adv{P})$ where:
    \begin{itemize}[nosep]
        \item $\adv{S}=S\union S'$, with $S' = \{(s,a) \mid s\in S, a\in A\}$;
        \item $\adv{A} = (\adv{A_1}\union \idleset) \times (\adv{A_2}\union \idleset)$, with $\adv{A_1}=A_1$, $\adv{A_2}=A_2\union \left( \bigcup_{s\in S, a\in A}{\verts[\Punc_{sa}]} \right)$;
        \item $\adv{\Delta}: \adv{S} \rightarrow 2^{\adv{A_1}\union \adv{A_2}}$, such that if $\adv{s}=s\in S$ then $\adv{\Delta}(\adv{s})= \Delta(s)$, else if $\adv{s}=(s,a)\in S'$ then $\adv{\Delta}(\adv{s})= \verts[\Punc_{sa}]$, else $\adv{\Delta}(\adv{s})=\emptyset$;
        \item $\adv{P}: \adv{S} \times \adv{A} \rightarrow \distr(\adv{S})$ such that if $\adv{s} = s\in S \land s' = (s,a) \in S'$ then $\adv{P}(\adv{s},\adv{a},s') = 1$, else if $\adv{s} = (s,a) \in S' \land \adv{a} = (\idle, P_{sa}) \land s' \in S$ then $\adv{P}(\adv{s},\adv{a},s') = P_{sas'}$, and $\adv{P}(\adv{s},\adv{a},s') = 0$ otherwise.
    \end{itemize}
\end{definition}
\vspace{-0.4cm}
As discussed earlier, player 2 in $\adv{\csg}$ acts as a coalition of nature and player 2 in $\csg$, such that the original player 2 acts at the $S$-states and nature moves at the $S'$-states. 
Given a choice $P_{sa} \in \Punc_{sa}$ of nature, a $\csg$-transition $s\xrightarrow{a} s'$ corresponds to the two-step $\adv{\csg}$-transition $s \xrightarrow{a} (s,a) \xrightarrow{(\idle,P_{sa})} s'$, and vice versa. 
In essence, the dynamics at $S$-states are unchanged. At an auxiliary state~${(s,a)\in S'}$, both players receive zero reward; and player 1 stays idle whilst player 2 deterministically selects a next-state distribution $P_{sa}\in \Punc_{sa}$, with $P^*_{sa}$ being an optimal such choice as characterised by \Cref{lem:infh-memoryless-optimal-strategies}. 
The notion of adversarial expansion naturally extends to strategies, paths, rewards, and objectives (see \extref{sec:infh-adv-supp}). 

We highlight that $\adv{\csg}$ is a finite-state, \emph{finite-action} CSG. As formalised in \extref{sec:infh-fin-act-resolution}, since an ICSG is polytopic, we can restrict player 2's actions to the vertices of the polytope $\Punc_{sa}$ at each $S'$-state without loss of optimality. 
The following results formalise the relationship between $\csg$ and $\adv{\csg}$.
\begin{restatable}[Utility-preserving strategy bijection]{lemma}{UtilPreservingStrategy}
\label{lem:infh-util-preserving-strategy-eq}
    For any $\csg$-profile $\sigma=(\sigma_1,\sigma_2)$, there exists a corresponding $\adv{\csg}$-profile $\adv{\sigma} = (\adv{\sigma_1}, \adv{\sigma_2})$ and vice versa, such that 
    $\inf_{P\in \Punc}{u_1(\sigma,P)} = \adv{u_1}(\adv{\sigma})$ 
    and 
    $\sup_{P\in \Punc}{u_2(\sigma,P)} = \adv{u_2}(\adv{\sigma})$, 
    where $\adv{u_i}(\adv{\sigma})$ denotes player $i$'s expected utility in $\adv{\csg}$ under $\adv{\sigma}$. 
\end{restatable}
\begin{proof}
\label{prf:infh-util-preserving-strategy-eq}
As shown in \extref{prop:infh-path-bijection,prop:infh-prob-preserved,prop:infh-obj-preserved}, $\adv{\csg}$ preserves the set of possible paths, their probabilities and objective values. 
It follows directly that: 
\[
u_1(\sigma,P) 
= \sum_{\pi \in \gs{IPaths}}{\prob^{\sigma} (\pi) \cdot X(\pi)} 
= \sum_{\adv{\pi} \in IPaths_{\adv{\csg},s}}{\prob^{\sigma}(\adv{\pi}) \cdot \adv{X}(\adv{\pi})}
= \adv{u_1}(\adv{\sigma})
\]
where $\prob^{\sigma}$ is the path probability function under profile $\sigma$.
Then by the zero-sum structure: $\sup_{P}{u_2(\sigma,P)} = u_2(\sigma,P^*) = -u_1(\sigma,P^*)=-\adv{u_1}(\adv{\sigma}) = \adv{u_2}(\adv{\sigma})$.

\end{proof}

\vspace*{-0.7em}

\begin{restatable}[Determinacy and Value Preservation]{corollary}{DetValPreservation}
\label{cor:infh-adv-value-preservation}
    $\csg$ is determined iff $\adv{\csg}$ is determined. 
    Further, if both games are determined, then the robust value of $\csg$ is equal to the value of $\adv{\csg}$, i.e., $\valg(s,X) = \vag(s,\adv{X})$.
    \refproof{prf:infh-adv-value-preservation}
    \reforig{cor:infh-adv-value-preservation}
\end{restatable}
\vspace*{-0.7em}

\begin{restatable}[RNE$_\csg$ $\Leftrightarrow $ NE$_{\adv{\csg}}$]{theorem}{RNEnNEBijection}
\label{thm:infh-adv-rne-to-ne}
    In a determined zero-sum ICSG $\csg$, for any $\csg$-profile $\sigma\in \Sigma$, 
    $\sigma$ is an RNE in $\csg$ with value $\valg(s,X)$ iff $\adv{\sigma}$ is an NE in $\adv{\csg}$ with value $\vag(s,\adv{X})=\valg(s,X)$.
    \reforig{thm:infh-adv-rne-to-ne}
\end{restatable}
\vspace*{-0.7em}
\begin{proof}[Sketch]
    Both directions follow from \Cref{def:subgame-perfect-rne} of RNE. The forward case additionally uses utility preservation (\Cref{lem:infh-util-preserving-strategy-eq}); the reverse relies on \Cref{def:adv-robust-determinacy-optimality-value} of the game value and value preservation (\Cref{cor:infh-adv-value-preservation}). 
    Full proof in \extref{prf:infh-adv-rne-to-ne}.
\end{proof}

\vspace{-0.3cm}
\startpara{Solving zero-sum ICSGs}
Finally, combining the above results, since $\adv{\csg}$ is finite-state and finitely-branching, it is determined for all the objectives we consider \cite{martin1998determinacy}. By \Cref{cor:infh-adv-value-preservation}, the original game $\csg$ is therefore also robustly determined with the same value. 
Moreover, since $\csg$ is zero-sum, any RNE profile and its value coincides with an optimal profile and the game value.

Hence, by \Cref{thm:infh-adv-rne-to-ne}, we can perform robust verification of an ICSG $\csg$ by solving the 2-player CSG $\adv{\csg}$, e.g, with the value iteration approach from~\cite{kwiatkowska2021automatic}.
In fact, we do not need to explicitly construct $\adv{\csg}$, nor its auxiliary states $S'$ corresponding to the possible $(s,a)$ pairs. Instead, for each state $s$, we first solve an inner optimisation problem over uncertainty set $\Punc_{sa}$ for each joint action $a$, and then solve a linear programming (LP) problem of size $|A|$ using the resulting values. We discuss this further in \Cref{sec:mc-ICSG} and give full details in \extref{sec:zs-mc-supp}.

\startpara{Finite-horizon properties}
\label{sec:zero-sum-finh}
When $X$ is a \emph{bounded} probabilistic reachability or cumulative reward objective, the previous results still hold under two changes: 
\begin{enumerate*}[label=\arabic*)]
    \item Robust optimal strategies are now \emph{time-varying}, i.e., with the extended signature $\sigma: S\times H \rightharpoonup A$ and $P^*: S\times A \times H \rightharpoonup \distr(S)$, where $H$ represents \emph{finite-memory} used to track the time-step. 
    Accordingly, we extend $\adv{\csg}$ with time-augmented states.
    \item Finite-horizon objectives are evaluated over $k < \infty$ steps, thus \emph{exact} game values can be computed via \emph{robust backward induction} (RBI) \cite{nilim2005robust,iyengar2005robust}, noting that $\csg$ is always determined due to the finite game tree.
\end{enumerate*}
We provide the detailed construction of $\adv{\csg}$ in this setting in \extref{sec:finh-supp}.

\vspace{-0.15cm}
\section{Nonzero-sum ICSGs}
\label{sec:nz-ICSG}
Next, we consider \emph{nonzero-sum} 2-player ICSGs, where each player $i$ maximises a distinct objective $X_i$, assuming for now that both objectives are infinite-horizon.%
\footnote{Following the usual approach for nonzero-sum CSGs~\cite{kwiatkowska2021automatic}, we focus on ICSGs that can be seen as a variant of \emph{stopping games} \cite{chen2013complexity}, where each player's target set is reached with probability 1 from all states under all profiles.}
As in the zero-sum~case, we continue to focus on player-first semantics and memoryless strategies, as both \Cref{thm:infh-nature-or-player-first-val-equiv} and \Cref{lem:infh-memoryless-optimal-strategies} extend to the nonzero-sum setting (see proofs in \extref{sec:nz-nature-first-supp}).

We again reduce a 2-player ICSG $\csg$ to its \emph{adversarial expansion} $\adv{\csg}$, which is a 3-player CSG, but in which the \emph{nature} player now acts adversarially against both other players, aiming to minimise their social welfare. The reduction is more complex than the zero-sum case and requires an additional \emph{filtering} step per iteration to identify robust equilibria.
This reduction again allows us to build on standard CSG solution methods \cite{kwiatkowska2021automatic}.

Our goal for nonzero-sum ICSGs is to find subgame-perfect $\varepsilon$-RNE (\Cref{def:subgame-perfect-rne}), and more specifically, $\varepsilon$-RSWNE (\Cref{def:rsw-rsc}),
which consider the \emph{sum} of the utilities for the two players.
We add a subscript $+$ to the relevant game notation (e.g., $r,u,X,V$) to indicate this.
%
The \emph{robust Bellman equation} is:
\vspace{-0.15cm}
\begin{equation}
\Vsum(s) = \sup_{\sigma \in \epsRNE} \inf_{P \in \Punc} \gs{\ev}^{\sigma, P}[\Xsum]
= \sup_{\sigma \in \epsRNE} \left[ 
       \rsum(s,\sigma) + \sum_{a\in A} \inf_{P_{sa}\in \Punc_{sa}} f_{sa}^{\sigma,P}
     \right]
\label{eq:nz-infh-bellman}
\end{equation}
where $f_{sa}^{\sigma,P}:=\sigma_{sa} \sum_{s'\in S} P_{sas'} \Vsum(s')$
and $\epsRNE$ denotes the set of \textit{one-shot} $\varepsilon$-RNE.
The second equality follows from a similar proof to the zero-sum case (see \extref{sec:zs-ICSG-supp}).
In this formulation, nature's inner problem is now to minimise the social welfare $\usum :=u_1 + u_2$, while each player $i\in \{1,2\}$ maximises their individual expected payoff $u_i$. Consequently, we maximise over the set of $\varepsilon$-RNE, capturing equilibrium behaviour under worst-case uncertainty.

The adversarial expansion $\adv{\rcsg}$ of ICSG $\rcsg$ for the nonzero-sum case follows a similar construction to the zero-sum setting (\Cref{def:infh-adv-resolution}), but now models nature as an explicit third player distinct from player~2.
Henceforth, if $a=(a_1,a_2)$, let $*a$ denote the flattened tuple $a_1,a_2$.

\begin{definition}[Adversarial expansion]
\label{def:nz-infh-adv-resolution} 
    We define the \emph{adversarial expansion} of $\csg$ as a \emph{3-player} CSG $\adv{\csg} = (\adv{N}, \adv{S}, \sbar, \adv{A}, \adv{\Delta}, \adv{P})$ where:
    \begin{itemize}[nosep]
        \item $\adv{N} = \{1,2,3\}$, with player 3 representing \emph{nature};
        \item $\adv{S}=S\union S'$, with $S' = \{(s,a) \mid s\in S, a\in A\}$;
        \item $\adv{A} = \times_{l=1}^3{(\adv{A_l} \union \idleset)}$ where $\adv{A_1} = A_1$, $\adv{A_2}=A_2$ and $\adv{A_3}=\union_{s\in S, a\in A}{\verts[\Punc_{sa}]}$;

        \item $\adv{\Delta}: \adv{S} \rightarrow 2^{\union_{i=1}^3{\adv{A_i}}}$, such that if $\adv{s}=s\in S$ then 
        $\adv{\Delta}(\adv{s})=\Delta(s)$, 
        else if $\adv{s}=(s,a)\in S'$ then $\adv{\Delta}(\adv{s})=\verts[\Punc_{sa}]$ and $\emptyset$ otherwise.
        \item 
        $\adv{P}: \adv{S} \times \adv{A} \rightarrow \distr(\adv{S})$, such that if $\adv{s} = s \in S \land \adv{a} = (*a, \idle) \land s' = (s, a) \in S'$ then $\adv{P}(\adv{s}, \adv{a}, s') = 1$, else if $\adv{s}=(s,a) \in S' \land \adv{a} = (\idle, \idle, P_{sa}) \land s' \in S$ then $\adv{P}(\adv{s}, \adv{a}, s') = P_{sas'}$, and $\adv{P}(\adv{s}, \adv{a}, s') =0$ otherwise.
    \end{itemize}
\end{definition}
\vspace*{-1.2em}
Under this definition of $\adv{\csg}$, a $\csg$-transition $s \xrightarrow{(a_1,a_2)} s'$ corresponds to a two-step $\adv{\rcsg}$-transition $s \xrightarrow{(a_1,a_2,\idle)} (s,(a_1,a_2)) \xrightarrow{(\idle,\idle,P_{s(a_1,a_2)})} s'$. 
Adversarial expansions~of paths, strategies, rewards and objectives follow analogously to the zero-sum case, and are formalised in \extref{sec:nz-infh-adv-resolution-supp}.
The following results relate $\rcsg$ and $\adv{\rcsg}$. Their proofs mirror the zero-sum case, with adaptations to the nonzero-sum definition of $\adv{\rcsg}$.
\begin{restatable}[Utility-preserving strategy bijection]{lemma}{UtilPreservingStrategiesNZ}
\label{lem:nz-infh-util-preserving-strategy-eq}
    For any $\csg$-profile $\sigma$ under nature's choice of $P\in \Punc$, there exists a corresponding $\adv{\csg}$-profile $\advP{\sigma}$ and vice versa such that $u_i(\sigma, P) = \adv{u_i}(\advP{\sigma})$ for $i\in \{1,2\}$. 
    Further, for any $\sigma\in \Sigma$, there exists a corresponding $\adv{\sigma}\in \adv{\Sigma}$ and vice versa such that 
    $\inf_{P\in \Punc}{\usum(\sigma,P)} = \adv{\usum}(\adv{\sigma}) = \adv{u_1}(\adv{\sigma}) + \adv{u_2}(\adv{\sigma})$.
\end{restatable}

\vspace*{-0.6em}
\begin{restatable}[$\varepsilon$-RNE$_\csg$ $\Rightarrow$ $\varepsilon$-NE$_{\adv{\csg}}$]{lemma}{EpsRNEtoNEInjection}
\label{lem:nz-infh-rne-to-ne}
    For any $\csg$-profile $\sigma\in \Sigma$, 
    if $\sigma$ is an $\varepsilon$-RNE in $\csg$ then $\adv{\sigma}\in \adv{\Sigma}$ is an $\varepsilon$-NE in $\adv{\csg}$.
    \refproof{prf:nz-infh-rne-to-ne}
    \reforig{lem:nz-infh-rne-to-ne}
\end{restatable}

Unlike the zero-sum setting, for \Cref{lem:nz-infh-rne-to-ne} the converse does \emph{not} necessarily hold: if $\adv{\sigma}$ is an $\varepsilon$-NE in $\adv{\csg}$, $\sigma$ need not be an $\varepsilon$-RNE in $\csg$. This is because $\adv{\sigma_3}$ selects a transition function $P^*$ that minimises the \emph{total} utility of player 1 and 2, rather than each player's utility individually as required by the $\varepsilon$-RNE condition (see \Cref{def:subgame-perfect-rne}). Therefore, any ICSG where nature's minimisation of the sum induces asymmetric incentives suffices as an example (see \extref{eg:nz-infh-epsNE-neq-epsRNE}).

Consequently, before identifying the RSWNE, for each state $s\in S$ we filter the set of $\varepsilon$-NE in $\adv{\csg}$ to retain only those that correspond to an $\varepsilon$-RNE in $\csg$, i.e., $\advepsRNE := \{\adv{\sigma} \in \adv{\Sigma} \mid \sigma \in \epsRNE \}$. Note that filtering is applied separately at each state since we construct subgame-perfect equilibria.


\startpara{Filtering $\advepsNE$ for $\epsRNE$}
\label{sec:nz-infh-filter-ne-for-rne}
Our method of filtering is based on a notion of \emph{deviations} made by players.
In this section, we fix a state $s\in S$ and candidate $\varepsilon$-RNE profile 
$\sigma \in \Sigma$. We designate player $i$ as the \emph{deviator}, whose strategy deviations $\sigma_i' \in \Sigma_i$ from $\sigma$ will be evaluated.
Further, we define the \emph{deviation gain} of player $i$ under its deviation $\sigma_i'$ and nature's choice of $P\in \Punc$ as:
\begin{equation}
\label{eq:rve-ui-delta}
    u_i^\Delta(\sigma_i', P) := u_i(\sigma_{-i}[\sigma_i'], P) - u_i(\sigma, P).
\end{equation}

\vspace{-0.1cm}
\begin{restatable}[$\varepsilon$-RNE condition over pure deviations]{lemma}{PureDevFiltersRNE}
\label{lem:nz-infh-pure-dev-to-filter-rne}
Let $\pureSigmai$ denote the set of (memoryless) deterministic strategies for player $i \in N=\{1,2\}$. 
A profile $\sigma \in \Sigma$ is an $\varepsilon$-RNE
iff the following condition holds:
\begin{equation*}
\Vbar_{i,\sigma}  := \sup_{P\in \Punc}\sup_{\sigma_i' \in \pureSigmai}{u_i^\Delta(\sigma_i',P)} \ \leq \varepsilon
\quad 
\forall{i\in N}.
\quad \reforig{lem:nz-infh-pure-dev-to-filter-rne}
\end{equation*}
\end{restatable}
\vspace*{-0.7em}
\begin{proof}[Sketch]
    The $\varepsilon$-RNE condition for $\sigma$ (\Cref{def:subgame-perfect-rne}) can be rewritten~as $\sup_{P}\sup_{\sigma_i' \in \Sigma_i}{u_i^\Delta(\sigma_i',P)} \leq \varepsilon$. 
    Since the expected utility (from any state $s$) is linear in the deviation $\sigma_i'$, the maximal gain is attained by a deterministic deviation. So it suffices to consider pure deviations. 
    Full proof in \extref{prf:nz-infh-pure-dev-to-filter-rne}.
\end{proof}

%
Observe that $\Vbar_{i,\sigma}$ corresponds exactly to the \emph{optimistic value} of an IMDP $\rve{\csg_{i,\sigma}}$ in which player $i$ acts as the agent. We refer to this IMDP as the \textit{deviation IMDP} and formalise it in \extref{sec:dev-IMDP-supp}. The value correspondence relies on the IMDP's construction whereby, given state space $\adv{S}= S\union S'$ of $\adv{\csg}$, transitions from states in $S$ to $S'$ depend exclusively on the fixed player's strategy, whilst those from $S'$ to $S$ are governed by nature’s choice of $P \in \Punc$. Further, the reward assigned to each state $s \in S$ corresponds to the expected reward gain at $s$ if player $i$ deviates from $\sigma_i$ to $\sigma_i'$.

Therefore, following \Cref{lem:nz-infh-pure-dev-to-filter-rne}, if the deviation IMDP has optimistic value $\Vbar_{i,\sigma} \leq \varepsilon$ for both players $i \in \{1,2\}$, then the candidate profile $\sigma$ constitutes an $\varepsilon$-RNE profile of $\csg$. We can thus characterise our goal in solving $\csg$ as identifying:
\begin{equation}
\label{eq:advRNE-with-Vbar}
\epsRNE = \left\{ \sigma \in \Sigma \ \middle| \ \adv{\sigma} \in \advepsNE 
\ \land \
\forall i\in\{1,2\}.\,{\Vbar_{i,\sigma} \leq \varepsilon}
\right\}.
\end{equation}

\vspace*{-.5em}
\startpara{Solving nonzero-sum ICSGs}
\label{sec:nz-infh-find-rswne}
Altogether, a profile $\sigma \in \Sigma$ is an $\varepsilon$-RSWNE in $\csg$ iff its corresponding $\adv{\sigma} \in \advepsRNE$ maximises $\adv{\usum}$ in $\adv{\csg}$. Hence, computing $\varepsilon$-RSWNE in the 2-player ICSG $\csg$ reduces to finding SWNE in the 3-player CSG $\adv{\csg}$ over $\advepsRNE$.
While this would in principle require a general 3-player CSG solver (e.g., \cite{kwiatkowska2020multi}), such algorithms rely on nonlinear programming and are computationally expensive. 
However, by exploiting the zero-sum coalitional structure of $\adv{\csg}$, the \emph{trimatrix} game at each state $s\in S$ can be reduced to a \emph{bimatrix} game (see \extref{sec:nz-mc-supp}). This can thus be solved more efficiently using the 2-player nonzero-sum CSG solution approach from \cite{kwiatkowska2021automatic}, together with the inner-problem solution algorithm in \extref{sec:infh-rvi-solve-inner}~to ensure robustness. 
As in the zero-sum case, this computation does not require explicit construction of $\adv{\csg}$.

Once we have the set of one-shot $\varepsilon$-NE of $\adv{\csg}$ at each $s\in S$, we filter for the $\varepsilon$-RNE equivalents: for each profile $\adv{\sigma} \in \advepsNE$, we: 
\begin{enumerate*}[label=\arabic*)]
    \item compute $\Vbar_{i,\sigma}$~for each player $i \in \{1,2\}$ on the deviation IMDP $\rve{\csg_{i,\sigma}}$; and
    \item retain $\sigma$ if all $\Vbar_{i,\sigma} \leq \varepsilon$.
\end{enumerate*}
The resulting profiles correspond to the $\varepsilon$-RNE in $\csg$, from which we select the one maximising $\usum$. This gives an $\varepsilon$-RSWNE profile and values in $\csg$ (by \Cref{lem:nz-infh-util-preserving-strategy-eq}).


\startpara{Finite- and mixed-horizon properties}
\label{sec:nz-finh-icsg}
Since players' objectives are distinct, they may differ in time horizon. 
If both $X_1$ and $X_2$ are finite-horizon, then we define the analysis horizon as the maximum of the two, i.e., $k := \max(k_1, k_2)$. In mixed-horizon cases, where one objective is finite-horizon and the other infinite-horizon, we transform the game into an equivalent one with two infinite-horizon objectives on an augmented model, following \cite{kwiatkowska2019equilibria}. Thus, we focus on cases where both objectives are either finite- or infinite-horizon.

The infinite-horizon framework generalises to finite-horizon objectives in a similar way to the zero-sum setting. Additionally, we consider exact RNE and RSWNE (i.e., $\varepsilon=0$), and account for potentially different player horizons by labelling time-augmented states that record when each player’s target is reached within their respective horizon. The full construction is detailed in \extref{sec:nz-finh-supp}.

\vspace{-0.15cm}
\section{Value Computation for Two-player ICSGs}
\label{sec:mc-ICSG}
\Cref{sec:zs-ICSG,sec:nz-ICSG} have presented the theoretical foundations for solving zero-sum and non-zero sum ICSGs, respectively, and described how a reduction to an \emph{adversarial expansion} CSG
provides the basis for iterative solution methods.
In this section, we present some additional implementation details
and discuss correctness and complexity.
Full details are provided in \extref{sec:zs-mc-supp,sec:nz-mc-supp}.

Both approaches use elements of robust dynamic programming for RMDPs,
i.e, robust value iteration (RVI) for infinite-horizon properties
and robust backward induction (RBI) for finite-horizon properties,
and of value iteration based methods for CSGs~\cite{kwiatkowska2021automatic}.
For both zero-sum and nonzero-sum objectives,
the procedure performed per iteration for each state~$s\in S$ is:
\vspace{-0.8cm}
\begin{center}
\resizebox{0.88\textwidth}{!}{%
\begin{minipage}{\textwidth}
\begin{algorithm}[H]
\caption{RVI/RBI update for state $s\in S$ in $\csg=(N,S,\sbar, A, \Delta, \Pcheck, \Phat)$}
\label{alg:RVI-RBI-iteration}
\begin{algorithmic}[1]
\ForAll{$a \in A(s)$}
    \State $P^*_{sa} \gets \textsc{SolveInnerProblem}(s, a, V_{prev}, \Pcheck, \Phat)$
    \Comment{\mbox{\annotationsize{see~\AbbrExtRef{sec:infh-rvi-solve-inner}}}}
\EndFor
\State $\Zi \gets \textsc{ConstructNFG}(P^*_s,V_{prev})$
\Comment{\mbox{%
  \ifextended
    \annotationsize{see~\AbbrCref{sec:zs-mc-supp} (zero-sum) and \ref{sec:nz-mc-supp} (nonzero-sum)}\ignorespaces
  \else
    \annotationsize{see~\cite{this2026Extended}}\ignorespaces
  \fi
}}
\State $V_{next}[s] \gets \textsc{SolveNFG}(\Zi)$

\end{algorithmic}
\end{algorithm}
\end{minipage}%
}
\end{center}

The \textsc{SolveInnerProblem} function (line 2) uses an algorithm adapted from a greedy method for IMDPs \cite{nilim2005robust}, which we detail in \extref{sec:infh-rvi-solve-inner}. 
At line 4, we build a normal form game: for zero-sum ICSGs this a \emph{matrix} game, reusing the zero-sum CSG algorithms in \cite{kwiatkowska2021automatic}; 
for nonzero-sum ICSGs, we build a (general-sum) \emph{bimatrix} game using multi-player CSG algorithms in \cite{kwiatkowska2020multi}, which coincide with those for nonzero-sum 2-player CSGs in \cite{kwiatkowska2021automatic}. 
Our derivations appear in 
\ifextended
  \Cref{sec:zs-mc-supp} (zero-sum) and \ref{sec:nz-mc-supp} (nonzero-sum)\ignorespaces
\else
  \cite{this2026Extended}\ignorespaces
\fi, where we demonstrate how to directly compute values of $\csg$ without explicit construction of $\adv{\csg}$. 

At line 5, \textsc{SolveNFG} computes the matrix game value via an LP formulation \cite{v1928theorie,von1944theory} in the zero-sum case. For nonzero-sum ICSGs, this is a multi-step procedure which involves:
\begin{enumerate*}[label=\arabic*)]
\item enumerating NE for bimatrix games, using e.g., the \emph{Lemke-Howson algorithm} \cite{lemke1964equilibrium}; and 
\item filtering these for RNE as outlined in \Cref{sec:nz-infh-filter-ne-for-rne} using deviation IMDPs. The latter can be done using IMDP verification algorithms already supported in PRISM-games \cite{KNPS20}.
\end{enumerate*}
If no profiles remain, our algorithm terminates early; otherwise the value of state $s$ is updated to the RSWNE value of $\Zi$.

\startpara{Correctness}
\label{sec:mc-correctness}
This relies on the reduction of solving an ICSG $\csg$ to solving a standard CSG $\adv{\csg}$, whose correctness is established in \Cref{sec:zs-ICSG} (zero-sum) and \ref{sec:nz-ICSG} (nonzero-sum). Correctness of the underlying CSG algorithm over $\adv{\csg}$ is inherited from \cite{kwiatkowska2021automatic,kwiatkowska2020multi}, which relies on classical results, e.g., correctness of value iteration and backward induction \cite{raghavan1991algorithms} and solution of matrix games \cite{v1928theorie,von1944theory}. 
Robustness is ensured by solving the inner problem, for which the correctness of our adapted algorithm in \extref{sec:infh-rvi-solve-inner}~follows from \cite{nilim2005robust}. 

\startpara{Complexity}
\label{sec:mc-complexity}
Runtime depends on the number of iterations and per-iteration cost, both of which are dependent on game size. Finite-horizon objectives require exactly $k$ iterations, while infinite-horizon ones iterate until convergence, which may be exponential in $|A|$ in the worst case \cite{chatterjee2012survey} (even for MDPs \cite{hansen2011complexity}). 
Since line 2 takes $O(|S|\log{|S|})$ time per execution, the per-iteration cost for zero-sum ICSGs is $O\left(|S|^2 |A| \log{|S|} + |S|L_{|A|} \right)$, where $L_{|A|}$ is the cost of solving an LP problem of size $|A|$. This is polynomial in $|S|$ and $|A|$ with e.g., Karmarkar's algorithm \cite{karmarkar1984new}; or exponential under the simplex algorithm \cite{dantzig1949simplex}, which is PSPACE-complete in the worst case \cite{fearnley2015complexity} but performs well on average \cite{todd2002many}. 
For nonzero-sum ICSGs, each iteration takes worst-case exponential time due to NE enumeration \cite{avis2010enumeration}.

\startpara{Strategy synthesis}
\label{sec:zs-strategy-synthesis}
%
At each state, solving the NFG yields both the value(s) and optimal player strategies, while nature's optimal strategy is given by the returned values of \textsc{SolveInnerProblem} (line 2 of \Cref{alg:RVI-RBI-iteration}). These are memoryless and taken from the final RVI iteration for infinite-horizon objectives; but time-varying and taken from each RBI step for finite-horizon objectives.

\vspace{-0.15cm}
\section{Tool support and Experimentation}
\label{sec:implementation-and-experimentation}
We extended PRISM-games~\cite{KNPS20} to support modelling and solution of 2-player ICSGs, building on its existing functionality for CSGs and IMDPs. The tool and case studies are available at \cite{artifact_2026_18189280}.

\startpara{Experimental setup}
\label{sec:experiment-method}
We evaluate the efficiency and scalability of our techniques,
comparing, as a point of reference, to standard (non-robust) CSG solution from PRISM-games.%
\footnote{Due to improvements in PRISM-games, some statistics differ slightly from 
\cite{kwiatkowska2021automatic}.}
We use the benchmarks from \cite{kwiatkowska2021automatic},
obtaining ICSGs by perturbing all non-0/1 probability CSG transitions with a two-sided uncertainty $\pm \epsilon$. 
This includes a combination of finite-, infinite- and mixed-horizon properties specified in the logic rPATL~\cite{chen2013automatic,kwiatkowska2021automatic}.
All experiments were run on a 3.2 GHz Apple M1 with 16 GB memory.
Further statistics for benchmark models and an extended set of results can be found in \extref{sec:experiment-results-supp}.

\begin{table}[!t]
\renewcommand{\arraystretch}{0.95}
\centering
\caption{Zero-sum verification results, with full statistics in \extref{tab:mc-stats-full}.}
\label{tab:mc-stats}

\begin{adjustbox}{width=\textwidth}
\begin{tabular}{|c|r||c|r|r|r|r|r|r|r|}
\hline
\multicolumn{1}{|c|}{\textbf{Case study: [params], $\epsilon$}} & 
\multicolumn{1}{c||}{\textbf{Param.}} & 
\multicolumn{1}{c|}{\textbf{Avg \#}} & 
\multicolumn{1}{c|}{\textbf{States}} & 
\multicolumn{2}{c|}{\textbf{Val. Iters}} & 
\multicolumn{2}{c|}{\textbf{Verif. time (s)}} & 
\multicolumn{2}{c|}{\textbf{Value}} \\
\cline{5-10}
\multicolumn{1}{|c|}{\textbf{Property}} &
\multicolumn{1}{c||}{\textbf{values}} & 
\multicolumn{1}{c|}{\textbf{actions}} & 
\multicolumn{1}{c|}{} & 
\multicolumn{1}{c|}{\textbf{CSG}} & 
\multicolumn{1}{c|}{\textbf{ICSG}} & 
\multicolumn{1}{c|}{\textbf{CSG}} & 
\multicolumn{1}{c|}{\textbf{ICSG}} & 
\multicolumn{1}{c|}{\textbf{CSG}} & 
\multicolumn{1}{c|}{\textbf{ICSG}} \\
\hline \hline

\multirow{3}{*}{\shortstack{\textbf{Robot coordination:} \\ $[l], \ 0.01$ \\ $\llangle{rbt_1}\rrangle \Rt_{\min=?} [\eventually g_1]$}} 
 & 4 & 2.07,2.07 & 226 & 19 & 18 & 0.15 & 0.27 & 4.55 & 4.39 \\
 & 8 & 2.52,2.52 & 3,970 & 29 & 29 & 2.50 & 3.61 & 8.89 & 8.63 \\
 & 12 & 2.68,2.68 & 20,450 & 39 & 37 & 16.22 & 31.73 & 13.15 & 12.84 \\
\hline

\multirow{2}{*}{\shortstack{\textbf{User centric network:} \\ $[K], \ 0.01$; $\llangle{usr}\rrangle  \Rt_{\min=?} [\eventually f]$}} 
 & 3 & 2.11,1.91 & 32,214 & 60 & 59 & 789.66 & 834.92 & 0.04 & 0.03 \\
 & 4 & 2.31,1.92 & 104,897 & 81 & 81 & 3525.67 & 3729.40 & 4.00 & 4.00 \\
\hline

\multirow{4}{*}{\shortstack{\textbf{Aloha:} \\ $[b_{\max}], \ 1/257$ \\ $\llangle{u_2},{u_3}\rrangle \Rt_{\min=?} [\eventually s_{2,3}]$}} 
 & 2 & 1.00120,1.00274 & 14,230 & 105 & 103 & 5.31 & 5.61 & 4.34 & 4.28 \\
 & 3 & 1.00023,1.00054 & 72,566 & 128 & 125 & 18.50 & 26.39 & 4.54 & 4.46 \\
 & 4 & 1.00004,1.00009 & 413,035 & 195 & 190 & 225.69 & 291.72 & 4.62 & 4.53 \\
 & 5 & 1.00001,1.00002 & 2,237,981 & 343 & 327 & 4669.54 & 4260.59 & 4.65 & 4.54 \\
\hline

\multirow{4}{*}{\shortstack{\textbf{Jamming radio systems:} \\ $[chans,slots], \ 0.01$ \\ $\llangle{u}\rrangle  \Pt_{\max=?} [ \eventually ({sent}\geq {slots}/2) ]$ }} 
 & 4,6 & 2.17,2.17 & 531 & 7 & 7 & 0.32 & 0.46 & 0.84 & 0.80 \\
 & 4,12 & 2.49,2.49 & 1,623 & 13 & 13 & 1.39 & 2.94 & 0.77 & 0.71 \\
 & 6,6 & 2.17,2.17 & 531 & 7 & 7 & 0.25 & 0.45 & 0.84 & 0.80 \\
 & 6,12 & 2.49,2.49 & 1,623 & 13 & 13 & 1.46 & 2.37 & 0.77 & 0.71 \\
\hline
\end{tabular}
\end{adjustbox}
\end{table}

\vspace{-0.2cm}
\begin{table}[!t]
\renewcommand{\arraystretch}{1.1}
\centering
\vspace{-0.35cm}
\caption{Nonzero-sum verification results, with full statistics in \extref{tab:nz-mc-stats-full}.}
\label{tab:nz-mc-stats}

\begin{adjustbox}{width=\textwidth}
\begin{tabular}{|c|r||c|r|r|r|r|r|r|r|}
\hline
\multicolumn{1}{|c|}{\textbf{Case study: [params], $\epsilon$}} & 
\multicolumn{1}{c||}{\textbf{Param.}} & 
\multicolumn{1}{c|}{\textbf{Avg \#}} & 
\multicolumn{1}{c|}{\textbf{States}} & 
\multicolumn{2}{c|}{\textbf{Val. Iters}} & 
\multicolumn{2}{c|}{\textbf{Verif. time (s)}} & 
\multicolumn{2}{c|}{\textbf{Value}} \\
\cline{5-10}
\multicolumn{1}{|c|}{\textbf{Property}} &
\multicolumn{1}{c||}{\textbf{values}} & 
\multicolumn{1}{c|}{\textbf{actions}} & 
\multicolumn{1}{c|}{} & 
\multicolumn{1}{c|}{\textbf{CSG}} & 
\multicolumn{1}{c|}{\textbf{ICSG}} & 
\multicolumn{1}{c|}{\textbf{CSG}} & 
\multicolumn{1}{c|}{\textbf{ICSG}} & 
\multicolumn{1}{c|}{\textbf{CSG}} & 
\multicolumn{1}{c|}{\textbf{ICSG}} \\
\hline \hline

\multirow{3}{*}{\shortstack{\textbf{Robot coordination}: $[l,k], \ 0.01$ \\ $\llangle{r_1:r_2}\rrangle_{\max=?} \left( \Pt [ \lnot{c} \until^{\leq k} g_1 ] + \Pt [ \lnot{c} \until^{\leq k} g_2 ] \right)$}}
& 4,4 & 2.07,2.07 & 226 & 4 & 4 & 0.26 & 0.60 & 1.55 & 1.50 \\
& 8,8 & 2.52,2.52 & 3,970 & 8 & 8 & 1.03 & 54.74 & 0.92 & 0.84 \\
& 12,12 & 2.68,2.68 & 20,450 & 12 & 12 & 8.55 & 2895.60 & 0.49 & 0.40 \\
\hline

\multirow{2}{*}{\shortstack{\textbf{Robot coordination}: $[l,k], \ 0.01$ \\ $\llangle{r_1:r_2}\rrangle_{\max=?} \left( \Pt [ \lnot{c} \until^{\leq k} g_1 ] + \Pt [ \lnot{c} \until\, g_2 ] \right)$}}
& 4,8 & 2.10,2.04 & 226 & 14 & 14 & 1.22 & 17.54 & 2.00 & 2.00 \\
& 4,16 & 2.12,2.05 & 3,970 & 14 & 11 & 2.08 & 75.23 & 2.00 & 2.00 \\
\hline

\multirow{3}{*}{\shortstack{\textbf{Aloha (deadline)}: $[b_{\max},D], \ 1/257$\\ $\llangle u_1:u_2,u_3 \rrangle_{\max=?} \left( \Pt[\eventually s_1] + \Pt[\eventually s_{2,3}] \right)$}} 
 & 1,8 & 1.0048,1.0111 & 14,230 & 23 & 23 & 0.45 & 5.13 & 1.99 & 1.99 \\ 
 & 2,8 & 1.0012,1.0027 & 72,566 & 23 & 23 & 1.17 & 107.63 & 1.98 & 1.97 \\ 
 & 3,8 & 1.0002,1.0005 & 413,035 & 22 & 22 & 3.96 & 3306.24 & 1.97 & 1.97 \\ 
\hline

\multirow{2}{*}{\shortstack{\textbf{Medium access}: $[e_{\max},k_1,k_2], \ 0.01$ \\ $\llangle p_1:p_2,p_3 \rrangle_{\max=?} \left(\Rt[\cumreward{k_1}] + \Rt[\cumreward{k_2}] \right)$}} 
 & 10,20,25 & 1.91,3.63 & 10,591 & 25 & 25 & 577.60 & 614.46 & 26.10 & 25.88 \\
 & 15,20,25 & 1.94,3.75 & 33,886 & 25 & 25 & 1148.02 & 6109.14 & 34.35 & 34.06 \\
\hline

\end{tabular}
\end{adjustbox}
\vspace{-0.5cm}
\end{table}


\paragraph{Q1. How does verification time for ICSGs compare to CSGs?}

Results for solving CSGs and ICSGs under adversarial uncertainty, completed within a 2-hour time limit, are shown in \Cref{tab:mc-stats} for zero-sum games and \Cref{tab:nz-mc-stats} for nonzero-sum games. Overall, the increase in verification time when moving from CSGs to ICSGs is significantly more pronounced in the nonzero-sum setting, highlighting the added complexity introduced by the nonzero-sum formulation.

Across all benchmarks, verification times for zero-sum ICSGs remained within a factor of two of their CSG counterparts (see \Cref{tab:mc-stats}), with all test instances (except for \emph{User-centric network} with $K > 4$) solved within 2 hours. These include models with over 2 million states and 10 million transitions. 
By contrast, for nonzero-sum ICSGs, verification completes within 2 hours for models up to 0.4 million states and 1.3 million transitions. This aligns with the underlying reduction framework: in the zero-sum setting, the problem reduces to a 2-player game, whereas in the nonzero-sum case it is a more complex 3-player game.

While the main overhead arises from solving the inner problem at each iteration, in the zero-sum setting this is often offset by faster RVI convergence. For example, in the \textit{Aloha} model with $b_{\max}=5$, ICSG verification outperformed CSGs in runtime and required noticeably fewer iterations to converge.
Additionally, consistent with \cite{kwiatkowska2021automatic}, for both zero- and nonzero-sum CSGs and ICSGs, verification time appears to depend more on the number of actions per player/coalition than the number of states. For example, the minimally branched \textit{Aloha} instances (averaging close to one action per coalition) are verified relatively efficiently compared to other games with similarly sized state spaces.

\vspace{-0.1cm}
\paragraph{Q2. How does $\epsilon$ affect verification time?} 
Interestingly, as shown in \Cref{fig:rbt-time-vs-eps}, verification times are often the lowest when $\epsilon$ is very small or close to its maximum value, while intermediate $\epsilon$ values tend to be slower. This non-monotonic behaviour again reflects the trade-off due to uncertainty: larger $\epsilon$ increases the per-iteration cost by giving nature more choices, but also accelerates convergence by flattening the value landscape \cite{nilim2005robust}, thus reducing the number of iterations required until convergence. 
Thus $\epsilon$ can be tuned to balance robustness and efficiency. 
However, the net effect of $\epsilon$ on verification time is model-specific: e.g., in the \textit{Robot Coordination} case with $l = 12$ from \Cref{tab:mc-stats}, ICSG verification required fewer iterations but still resulted in a longer overall verification time.

\vspace{-0.1cm}
\paragraph{Q3. How does $\epsilon$ influence the computed value?}
Increasing $\epsilon$ yields more conservative results, as shown in \Cref{fig:IDS-val-vs-rounds}. This is to be expected, as expanding the uncertainty set enables nature to select ``worse'' transitions that reduce the game value. The effect is amplified in larger, more connected models, in which pessimistic transitions propagate over longer paths \cite{wiesemann2013robust}. 
\Cref{fig:IDS-val-vs-rounds} also illustrates a practical use of ICSG verification in computing two-sided, $\epsilon$-parametrised bounds on verification results, which can be interpreted as confidence intervals under transition uncertainty. This is useful in safety-critical and performance-sensitive settings, where robustness must be ensured against worst-case security attacks, while enabling estimation of, e.g., optimistic operational performance.

\begin{figure}[!t]
\centering
\begin{minipage}[t]{0.43\linewidth}
\centering
\vspace*{0pt}

\vspace{0pt}
\centering
\begin{tikzpicture}
  \begin{axis}[
    width=0.9\linewidth,
    height=5.5cm,
    xlabel={$\epsilon$},
    ylabel={Relative time (log scale)},
    ylabel style={font=\footnotesize},
    ticklabel style={font=\scriptsize},
    scaled y ticks=false,
    xmin=0.001,
    xtick={0.01, 0.02, 0.03, 0.04, 0.049},
    xticklabel style={/pgf/number format/fixed, /pgf/number format/precision=3},
    ymode=log,
    log basis y={10},
    ymin=1.5, 
    ymax=1150,
    grid=both,
    grid style={line width=.1pt, draw=gray!10},
    major grid style={line width=.2pt,draw=gray!50},
    legend style={
      at={(0.02,0.14)}, anchor=south west,
      font=\scriptsize,
      cells={align=left}
    },
    cycle list={
      {blue, mark=*, mark options={solid, scale=1.2}},
      {red, mark=square*, mark options={solid, scale=1.1}},
      {green!60!black, mark=triangle*, mark options={solid, scale=1.2}},
      {purple, mark=diamond*, mark options={solid, scale=1.2}},
    },
  ]
    \addplot table[x=eps, y=4, col sep=comma] {robot-time-vs-eps.csv};
    \addlegendentry{$l=4$}

    \addplot table[x=eps, y=8, col sep=comma] {robot-time-vs-eps.csv};
    \addlegendentry{$l=8$}

    \addplot table[x=eps, y=12, col sep=comma] {robot-time-vs-eps.csv};
    \addlegendentry{$l=12$}

    \draw[thick, black, dashed] 
      (axis cs:\pgfkeysvalueof{/pgfplots/xmin},1) -- (axis cs:\pgfkeysvalueof{/pgfplots/xmax},1);

    \addplot[thick, black, dashed] coordinates {(0,1) (0,1)};
  \end{axis}
\end{tikzpicture}
\vspace{-0.15cm}
\caption{Total verification time relative to CSG baseline in the first nonzero-sum \textit{Robot coordination} case study, which requires $\epsilon<0.05$.}
\label{fig:rbt-time-vs-eps}
\end{minipage}
\hfill
\begin{minipage}[t]{0.52\textwidth}
\centering
\vspace*{0pt}
\tikzset{
  mycircle/.style={mark=*, mark options={solid, scale=0.9}, thin, blue},
  mysquare/.style={mark=square*, mark options={solid, scale=0.9}, thin, blue!85!white},
  mytriangle/.style={mark=triangle*, mark options={solid, scale=1.2}, thin, blue!65!white},
  mydiamond/.style={mark=diamond*, mark options={solid, scale=1.2}, thin, blue!45!white}
}

\begin{tikzpicture}
\begin{axis}[
    width=0.97\linewidth,
    height=5.5cm,
    xlabel={rounds},
    ylabel={Value},
    xlabel style={
        font=\footnotesize,
        at={(axis description cs:0.5,-0.07)}, anchor=north
    },
    xticklabel style={font=\scriptsize}, 
    yticklabel style={font=\scriptsize},
    xmin=25, xmax=200,
    xtick={25, 50, 100, 200},
    ymin=0, ymax=185,
    legend style={
        at={(0.02,0.98)}, anchor=north west,
        font=\scriptsize,
        cells={align=left},
        inner xsep=1pt,  
        inner ysep=0.5pt,  
        column sep=2pt,  
        /tikz/every even column/.append style={column sep=2pt},  
    },
    cycle list name=linestyles,
    mark size=1.8pt,
]

\addplot+[mark=x, very thick, gray] table[x=rounds, y="0", col sep=comma] {IDS-value-vs-rounds.csv};
\addlegendentry{CSG}

\addplot+[mycircle, solid] table[x=rounds, y="0.01", col sep=comma] {IDS-value-vs-rounds.csv};
\addlegendentry{$\epsilon=0.01$}

\addplot+[mysquare, solid] table[x=rounds, y="0.05", col sep=comma] {IDS-value-vs-rounds.csv};
\addlegendentry{$\epsilon=0.05$}

\addplot+[mytriangle, solid] table[x=rounds, y="0.10", col sep=comma] {IDS-value-vs-rounds.csv};
\addlegendentry{$\epsilon=0.10$}

\addplot+[mydiamond, solid] table[x=rounds, y="0.19", col sep=comma] {IDS-value-vs-rounds.csv};
\addlegendentry{$\epsilon=0.19$}

\addplot[name path=csg, draw=none] table[x=rounds, y="0", col sep=comma] {IDS-value-vs-rounds.csv};
\addplot[name path=eps001, draw=none] table[x=rounds, y="0.01", col sep=comma] {IDS-value-vs-rounds.csv};
\addplot[name path=eps005, draw=none] table[x=rounds, y="0.05", col sep=comma] {IDS-value-vs-rounds.csv};
\addplot[name path=eps010, draw=none] table[x=rounds, y="0.10", col sep=comma] {IDS-value-vs-rounds.csv};
\addplot[name path=eps019, draw=none] table[x=rounds, y="0.19", col sep=comma] {IDS-value-vs-rounds.csv};

\addplot[fill=blue!65!white] fill between[of=csg and eps001];
\addplot[fill=blue!45!white] fill between[of=eps001 and eps005];
\addplot[fill=blue!25!white] fill between[of=eps005 and eps010];
\addplot[fill=blue!10!white] fill between[of=eps010 and eps019];

\addplot+[mycircle, dashed] table[x=rounds, y="0.01", col sep=comma] {IDS-opt-value-vs-rounds.csv};
\addplot+[mysquare, dashed] table[x=rounds, y="0.05", col sep=comma] {IDS-opt-value-vs-rounds.csv};
\addplot+[mytriangle, dashed] table[x=rounds, y="0.10", col sep=comma] {IDS-opt-value-vs-rounds.csv};
\addplot+[mydiamond, dashed] table[x=rounds, y="0.19", col sep=comma] {IDS-opt-value-vs-rounds.csv};

\addplot[name path=opt001, draw=none] table[x=rounds, y="0.01", col sep=comma] {IDS-opt-value-vs-rounds.csv};
\addplot[name path=opt005, draw=none] table[x=rounds, y="0.05", col sep=comma] {IDS-opt-value-vs-rounds.csv};
\addplot[name path=opt010, draw=none] table[x=rounds, y="0.10", col sep=comma] {IDS-opt-value-vs-rounds.csv};
\addplot[name path=opt019, draw=none] table[x=rounds, y="0.19", col sep=comma] {IDS-opt-value-vs-rounds.csv};

\addplot[fill=blue!65!white] fill between[of=csg and opt001];
\addplot[fill=blue!45!white] fill between[of=opt001 and opt005];
\addplot[fill=blue!25!white] fill between[of=opt005 and opt010];
\addplot[fill=blue!10!white] fill between[of=opt010 and opt019];

\end{axis}
\end{tikzpicture}

\caption{ICSG values over game size in the  \textit{Intrusion Detection} case study, under two resolutions of uncertainty: \emph{adversarial} (solid lines) and \emph{controlled} by coalition 1 (dashed).}
\label{fig:IDS-val-vs-rounds}
\end{minipage}
\label{fig:vs-eps}
\vspace*{-0.45cm}
\end{figure}


\vspace{-0.2cm}
\section{Conclusion}
\label{sec:conclusion}
We have introduced robust CSGs, an extension of classical CSGs with transition uncertainty that enables principled analysis of multi-agent, concurrent stochastic systems with imprecise dynamics. 
Focusing on interval uncertainty, i.e., ICSGs, we developed verification algorithms for the 2-player setting, for both finite- and infinite-horizon objectives. Our approach relies on a value-preserving reduction of ICSGs to standard CSGs,
thereby allowing reuse of elements of the solution methods for both CSGs and RMDPs, plus custom adaptations and filtering.
Our implementation in PRISM-games shows that solution in the zero-sum case scales comparably to standard CSGs, with runtime increases below a factor of two. 
In the nonzero-sum case, computational demands are higher but we still scale successfully to large CSGs.
Future work could explore RCSGs with richer uncertainty models and/or alternative solution concepts (e.g., robust correlated equilibria), and consider objectives with more general temporal specifications. 

\unless\ifextended




\startpara{Data Availability Statement}
The models, tools, and scripts to reproduce our experimental evaluation are archived and available at \cite{artifact_2026_18189280}.

\fi

\startpara{Acknowledgments}
Supported by the EPSRC Centre for Doctoral Training no. EP/Y035070/1 and the UKRI AI Hub on Mathematical Foundations of AI.



%
%
\bibliographystyle{splncs04}
\bibliography{references}


\newpage
\appendix

\ifextended
  \setcounter{section}{0}
\setcounter{secnumdepth}{3}
\renewcommand{\thesection}{\Alph{section}}
\renewcommand{\thesubsection}{\thesection.\arabic{subsection}} 
\renewcommand{\thesubsubsection}{\thesubsection.\arabic{subsubsection}} 

\counterwithin{equation}{section}
\counterwithin{figure}{section}
\counterwithin{table}{section}
\counterwithin{theorem}{section}  
\counterwithin{lemma}{section}
\counterwithin{proposition}{section}
\counterwithin{remark}{section}
\counterwithin{definition}{section}
\counterwithin{corollary}{section}
\counterwithin{algorithm}{section}

\renewcommand{\theequation}{\thesection.\arabic{equation}}
\renewcommand{\thefigure}{\thesection.\arabic{figure}}
\renewcommand{\thetable}{\thesection.\arabic{table}}
\renewcommand{\thetheorem}{\thesection.\arabic{theorem}}
\renewcommand{\thelemma}{\thesection.\arabic{lemma}}
\renewcommand{\theproposition}{\thesection.\arabic{proposition}}
\renewcommand{\theremark}{\thesection.\arabic{remark}}
\renewcommand{\thedefinition}{\thesection.\arabic{definition}}
\renewcommand{\thecorollary}{\thesection.\arabic{corollary}}
\renewcommand{\thealgorithm}{\thesection.\arabic{algorithm}}
\renewcommand{\theexample}{\thesection.\arabic{example}}


\section*{Appendix}



\section{Existence of RNE in ICSGs}
\label{sec:RNE-existence-in-ICSG}

\begin{example}[ICSG where RNE exists]
\label{eg:RNE-exists}
Consider the ICSG shown in \Cref{fig:RNE-exists-eg}. From state 0, the best response for player 1 is clearly to play action $a_1$, as any other choice yields zero reward. Similarly, player 2's best response is to play action $b_1$, as playing any other action (i.e., $b_2$) would prevent her from reaching state 2 under any realisation of $p$ selected by nature. Hence, the unique RNE is the deterministic profile in which player 1 plays $a_1$ with probability 1 and player 2 plays $b_1$ at state 0.

\vspace{-0.7cm}
\begin{figure}[H]
\centering
\begin{tikzpicture}[
  ->, 
  >=stealth,
  every node/.style={draw=black, circle, minimum size=8mm, inner sep=1pt},
  labelstyle/.style={fill=none, draw=none}
  ]

  \node (s0) at (0, 0) {0}; 
  \node (s1) at (-3, 0) {1};
  \node (s2) at (3, 0) {2};

  \draw[->, thick] (0, -0.7) -- (s0);

  \coordinate (mid1) at (0, 1.7);

  \draw (s0) -- node[pos=0.4, right, labelstyle] {\footnotesize$(a_1,b_1)$} (mid1);
  \draw (mid1) to[bend right=15] node[left=0.5, labelstyle] {\footnotesize$1-p$} (s1);
  \draw (mid1) to[bend left=15] node[right=0.5, labelstyle] {\footnotesize$p\in [0.7, 0.9]$} (s2);

  \draw (s0) -- node[below=-0.5, labelstyle] {\footnotesize$(a_2,b_1), 1$} (s2);

  \draw (s0) -- node[below=-0.5, labelstyle] {\footnotesize$(\_,b_2), 1$} (s1);

  \draw (s1) edge[loop left] ();
  \draw (s2) edge[loop right] ();

\end{tikzpicture}

\vspace*{-0.3em}
\caption{ICSG with actions $A_1=\{a_1,a_2\}$ and $A_2=\{b_1,b_2\}$. Player 1 aims to maximise her 1-step cumulative reward, receiving $r_1(s_0,(a_1,\_))=1$ and 0 otherwise. Player 2 aims to eventually reach $T_2=\{s_2\}$.}
\label{fig:RNE-exists-eg}
\end{figure}
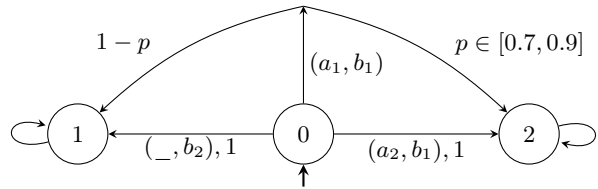
\end{example}

\vspace*{-1cm}
\begin{example}[ICSG where RNE does not exist]
\label{eg:no-RNE}
Consider the ICSG shown in \Cref{fig:no-RNE-eg}. Since the transition probabilities from state 0 depend solely on player 2’s action, we need only consider player 2’s strategy. Suppose player 2 selects action $b_1$ with probability $\alpha \in [0,1]$ and $b_2$ with probability $1 - \alpha$. Then her expected utility is $\alpha p + (1-\alpha) q = \alpha(p-q) + q$. Thus, if nature selects a realisation where $p > q$, player 2 maximises her utility by choosing $\alpha = 1$, i.e., playing $b_1$ deterministically. Conversely, if $p < q$, her best response is $\alpha = 0$, i.e., playing $b_2$ deterministically. When $p = q$, player 2 is indifferent. Thus, for any (mixed) strategy $\alpha \in [0,1]$, there exists a realisation of $p$ and $q$ under which player 2 has an incentive to unilaterally deviate. Therefore, no strategy of player~2 can be a best response across all realisations of $p$ and $q$, and so no RNE exists.

\vspace{-1cm}
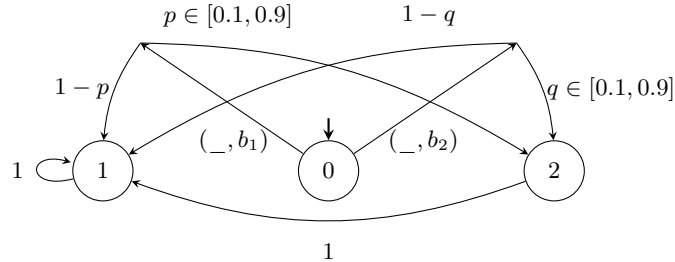
\begin{figure}[H]
\centering
\begin{tikzpicture}[
  ->, 
  >=stealth,
  every node/.style={draw=black, circle, minimum size=8mm, inner sep=1pt},
  labelstyle/.style={fill=none, draw=none}
  ]

  \node (s0) at (0, 0) {0}; 
  \node (s1) at (-3, 0) {1};
  \node (s2) at (3, 0) {2};

  \draw[->, thick] (0, 0.7) -- (s0) node[midway, right, labelstyle] {};

  \coordinate (mid1) at (-2.5, 1.7);
  \coordinate (mid2) at (2.5, 1.7);

  \draw (s0) -- node[pos=0.11, left, xshift=-4pt, labelstyle] {\footnotesize$(\_,b_1)$} (mid1);
  \draw (mid1) to[bend right=15] node[left, labelstyle] {\footnotesize$1-p$} (s1);
  \draw (mid1) to[bend left=15] node[pos=0.2, above=-0.5, labelstyle] {\footnotesize$p\in[0.1,0.9]$} (s2);

  \draw (s0) -- node[pos=0.11, right, xshift=4pt, labelstyle] {\footnotesize$(\_,b_2)$} (mid2);
  \draw (mid2) to[bend right=15] node[pos=0.2, above, labelstyle] {\footnotesize$1-q$} (s1);
  \draw (mid2) to[bend left=15] node[right, labelstyle] {\footnotesize$q\in [0.1,0.9]$} (s2);

  \draw (s1) edge[loop left] node[labelstyle, left, xshift=4pt] {\footnotesize$1$} (s1);

  \draw (s2) to[bend left=20] node[below, labelstyle] {\footnotesize$1$} (s1);

\end{tikzpicture}

\vspace*{-0.7em}
\caption{ICSG with actions $A_1=\{a_1,a_2\}$ and $A_2=\{b_1,b_2\}$. Player 1 and 2's objective is to reach $T_1=\{s_1\}$ and $T_2=\{s_2\}$, respectively, within $k=2$ steps.}
\label{fig:no-RNE-eg}
\end{figure}
\end{example}

\section{Proofs for Infinite-horizon Zero-sum ICSGs}
\label{sec:zs-supp}

\subsection[Sufficiency of Memoryless Strategies for Optimality]{Proof of \Cref{lem:infh-memoryless-optimal-strategies}}
\label{sec:infh-memoryless-optimal-strategies-proof}

\begin{lemma}[Nature's Inner Problem]
\label{lem:infh-nature-inner-problem}
    In an ICSG $\csg$, nature can minimise its value from state $s$ by independently choosing a distribution $P_{sa} \in \Punc_{sa}$ at each state-action pair $(s,a)$ that minimises the expected next-state value, i.e.
    \begin{equation}
    \label{eq:infh-natures-min-problem}
    \inf_{P\in \Punc}{\sum_{a\in A,s'\in S}{\sigma_{sa}{P_{sas'}V(s')}}}
    = \sum_{a\in A}{\sigma_{sa} \inf_{P_{sa}\in \Punc_{sa}}{\sum_{s'\in S}{P_{sas'}V(s')}}}.
    \end{equation}
\end{lemma}
\begin{proof}
\label{prf:infh-nature-inner-problem}
As discussed in \Cref{sec:RCSG}, the ICSG $\csg$ is $(s,a)$-rectangular, meaning that $\Punc$ decomposes as $\Punc=\bigtimes_{(s,a)\in S\times A}{\Punc_{sa}}$. So for any transition function $P\in \Punc$, we have that ${P_{sa}\in \Punc_{sa}}$ for all $(s,a)\in S\times A$.

Fix a state $s\in S$, and let $f(P) = \sum_{a,s'}{C_{sas'}{P_{sas'}}}$ where $C_{sas'} := \sigma_{sa} \cdot V(s')$. 
Let $P^*\in \Punc$ be the globally optimal transition function that minimises the value of $\csg$ from state $s$. 
Suppose for contradiction that $P^*$ does not consist entirely of ``locally optimal'' distributions, i.e., there exists some $a' \in A$ such that $P^*_{sa'} \neq \arg\min_{P_{sa'}\in \Punc_{sa'}}{P_{sa's'} \cdot C_{sa's'}} =: \tilde{P}_{sa'}$.
However, we could then always reduce the value by replacing that $P^*_{sa'}$ with $\tilde{P}_{sa}$. Call the resulting transition function $\tilde{P}$, which is still in $\Punc$ by $(s,a)$-rectangularity. We thus have that:
\begin{align*}
    f(\tilde{P}) 
    &= \sum_{a\in A, s'\in S}{C_{sas'}\tilde{P}_{sas'}}\\
    &= \sum_{a\neq a', s'}{C_{sas'}{P^*_{sas'}}} + \sum_{s'}{C_{sa's'}\tilde{P}_{sa's'}}
    &&\text{\small(by definition of $\tilde{P}$)}\\
    &= \sum_{a\neq a', s'}{C_{sas'}{P^*_{sas'}}} + \min_{P_{sa'} \in \Punc_{sa'}}{\sum_{s'}{C_{sa's'}P_{sa's'}}}
    &&\text{\small(by definition of $\tilde{P}_{sa'}$)}\\
    &< 
    \sum_{a\neq a', s'}{C_{sas'}{P^*_{sas'}}} + \sum_{s'}{C_{sa's'}P^*_{sa's'}}
    &&\text{\small(by definition of $P^*_{sa'}$)}\\
    &= \sum_{a,s'}{C_{sas'}{P^*_{sas'}}}
    = f(P^*)
\end{align*}
which contradicts the global optimality of $P^*$ (for minimisation).
\end{proof}

\MemorylessSufficiency*
\begin{proof}
\label{prf:infh-memoryless-optimal-strategies}
We first prove the sufficiency of \emph{memoryless} strategies for optimality.
Let $P^*: FPaths_{\csg} \rightharpoonup \distr(S)$ be the witnessing optimal (minimising) transition function for nature. 
Further let $\pi_1, \pi_2 \in \fpathg$ be two finite paths which both end in $s\in S$. 
Then for any $\sigma\in \Sigma$:
\begin{align*}
    P^*(\pi_1)
    &= \arg\min_{P\in \Punc}{X(\pi_1) + \valg(s\mid \sigma,P)} 
    &&\text{\small(by additivity of objectives; see \AbbrCref{sec:prelim-CSG})}\\
    &= \arg\min_{P\in \Punc}{\valg(s\mid \sigma,P)} 
    &&\text{\small($X(\cdot)$ is independent of $P$)}\\
    &= \arg\min_{P\in \Punc}{X(\pi_2) + \valg(s\mid \sigma,P)}\\
    &= P^*(\pi_2)
\end{align*}
Therefore $P^*$ is memoryless. 
The same proof of memorylessness applies to player's strategies, except that nature's objective $\min_{P \in \Punc} \valg(s\mid \sigma,P)$ is replaced by that player's objective ($\max_{\sigma_1 \in \Sigma_1}$ for player 1 and $\min_{\sigma_2 \in \Sigma_2}$ for player 2).

Next, we prove that nature has a \emph{deterministic} memoryless optimal strategy. 
Consider a memoryless optimal strategy $P^*\in \Punc$ of nature, given by:
\begin{align*}
    P^* := \arg\min_{P\in\Punc}{\gs{\ev}^{\sigma,P}[X]}
    = \left\{ \arg\min_{P_{sa}\in \Punc_{sa}}{\sum_{s'}{P_{sas'}} V(s')} \right\}_{(s,a)\in S \times A} 
    \quad \text{\small(by \AbbrCref{lem:infh-nature-inner-problem})}.
\end{align*} 

Following the first part of the proof, we restrict ourselves to the set of memoryless policies of nature, denoted $\Sigma_3$.
Let $\eta^*\in \Sigma_3$ denote a deterministic strategy where nature deterministically and independently selects the distribution $P_{sa}^*$ at each $(s,a)\in S\times A$, i.e., $\eta^*(s,a)(P_{sa}) = \ind[P_{sa} = P_{sa}^*]$. Note that this is possible due to the $(s,a)$-rectangularity of the ICSG $\csg$.
Then at each state $s$ under $\eta^*$ we have that for any (possibly mixed) strategy $\sigma_3 \in \Sigma_3$ of nature:
\begin{align*}
    V(s\mid \eta^*) 
    &= \sum_{P_{sa} \in \Punc_{sa}}{\left( \eta^*((s,a),P_{sa}){\sum_{s'\in S}{P_{sas'} V(s')}} \right)}\\
    &= \sum_{P_{sa} \in \Punc_{sa}}{\left( \ind[P_{sa}=P_{sa}^*] \cdot \sum_{s'\in S}{P_{sas'} V(s')} \right)}\\
    &= \sum_{s'\in S}{P_{sas'}^* V(s')}
    = \min_{P_{sa}\in \Punc_{sa}}{\sum_{s'\in S}{P_{sas'} V(s')}} 
    \qquad \text{\small(by definition of $P^*$)}\\
    &\leq \sum_{P_{sa}\in \Punc_{sa}}{ \left( \sigma_3((s,a),P_{sa}) \cdot \sum_{s'\in S}{P_{sas'} V(s')} \right)}  
    = V(s\mid \sigma_3).
\end{align*}
The final inequality follows since the right-hand side expression is a convex combination of $\sum_{s'\in S}{P_{sas'} V(s')}$. 
Therefore $\eta^*$ is optimal.

\end{proof}

\subsection{Derivation of the Robust Bellman Equation}
\label{sec:zs-ICSG-supp}

In the following we write $\sigma=(\sigma_1,\sigma_2)$. 
By \Cref{def:adv-robust-determinacy-optimality-value} of the robust value:
\begingroup
\allowdisplaybreaks
\begin{align*}
    V(s) &= \sup_{\sigma_1 \in \Sigma_1} \inf_{\sigma_2\in \Sigma_2} \inf_{P\in \Punc}{\ev_s^{\sigma,P}[X]} \\
    &= \sup_{\sigma_1 \in \distr(A_1(s))} \inf_{\sigma_2\in \distr(A_2(s))} \inf_{P\in \Punc}{\ev_s^{\sigma,P}[X]} 
    &&\text{\small(by \AbbrCref{lem:infh-memoryless-optimal-strategies})}\\
    &= \sup_{\sigma_1} \inf_{\sigma_2}  \inf_{P}{ 
        \sum_{a\in A}{
        \sigma_{sa} \left( r_{sa} + \sum_{s'\in S}{P_{sas'} \cdot \ev_{s'}^{\sigma,P}[X]} \right)
     }}
     &&\text{\small($X$ is additive)}\\
     &= \sup_{\sigma_1} \inf_{\sigma_2}  \inf_{P} { \left\{
        \sum_{a}{ \sigma_{sa} r_{sa}} +
        {\sum_{a}{ \sigma_{sa} \sum_{s'}{P_{sas'} \cdot \ev_{s'}^{\sigma,P}[X]}}}
        \right\}
     }\\
     &= \sup_{\sigma_1} \inf_{\sigma_2} { 
        \left\{
        \sum_{a}{
        \sigma_{sa} r_{sa} + 
        \inf_{P\in \Punc}{\sum_{a,s'}{\sigma_{sa} {P_{sas'} \cdot \ev_{s'}^{\sigma,P}[X]}}}
     } \right\}} \\
    &= \sup_{\sigma_1} \inf_{\sigma_2} { 
        \left\{
        r_{s}^{\sigma} + 
        \sum_{a}{\sigma_{sa} \inf_{P_{sa}\in \Punc_{sa}} \sum_{s'}{ { P_{sas'} \cdot V(s')}}}
      \right\}}
      &&\text{\small(by \AbbrCref{lem:infh-nature-inner-problem})}
\end{align*}
\endgroup
which matches \Cref{eq:infh-adv-bellman}.

\subsection{Extended Player-first Adversarial Expansion}
\label{sec:infh-adv-supp}

For generality we first work with an infinite-action version of $\adv{\csg}$ in the following proofs, i.e., keeping all else the same in \Cref{def:infh-adv-resolution} except that $\adv{A_2}=A_2\union \bigcup_{s\in S,a\in A}{\Punc_{sa}}$ and $\Delta_2(\adv{s})=\Punc_{sa}$ if $\adv{s}=(s,a)\in S'$. We later reduce it to a \emph{finite-action} CSG in \Cref{sec:infh-fin-act-resolution}.

To distinguish nature's states $S'$ from players' states $S$ in an $\adv{\csg}$-path $\advP{\pi}$, we index $S'$-states with primed indices $\nats' := \{0', 1', \ldots\}$, such that index $j':= 2j+1$ for $j\in \nats$. This ensures $\advP{\pi}$ preserves the state and transition indices in $\pi$ (\Cref{prop:infh-path-index-preserved}).

\begin{definition}[Adversarial Expansion cont.]
\label{def:infh-adv-components}
Consider a $\csg$-\emph{profile} $\sigma=(\sigma_1,\sigma_2)$ under nature's selection $P\in \Punc$. Define $\advP{\sigma} := (\adv{\sigma_1}, \advP{\sigma_2})$ such that for any $\adv{s} \in S$: 
if $\adv{s}=s \in S$, then $\adv{\sigma_1}(\adv{s},a_1)=\sigma_1(s,a_1)$ and $\advP{\sigma_2}(\adv{s},a_2)=\sigma_2(s,a_2)$; 
else if $\adv{s}=(s,a) \in S'$ then $\adv{\sigma_1}(\adv{s},a_1)=\ind[a_1=\idle]$ and $\advP{\sigma_2}(\adv{s},a_2)= \ind \left[ a_2 = P_{sa} \right]$. 

We denote by $\adv{\Sigma}$ the set of admissible memoryless strategy profiles in $\adv{\csg}$.\\

Consider an infinite $\csg$-\emph{path}\footnote{For ICSGs, we use subscripts to denote player indices and superscripts to indicate time-steps in path notation.} $\pi = s_0 \xrightarrow{a^0} s_1 \xrightarrow{a^1} s_2 \rightarrow \ldots$ under nature's choice of $P \in \Punc$, and let $P_j = P_{s_j a^j}$. Define:
\begin{align}
\label{eq:infh-advexpand-pi}
\advP{\pi} 
:= s_0 \xrightarrow{a^0} (s_0,a^0) \xrightarrow{(\idle, P_0)} s_1 \xrightarrow{a^1} (s_1,a^1) \xrightarrow{(\idle,P_1)} s_2 \rightarrow \ldots
\end{align}
The converse mapping also holds. 
The same mapping applies to finite paths, with both $\pi$ and $\advP{\pi}$ ending in $last(\pi)\in S$.\\

Consider the \emph{reward structure} $r=(r_A,r_S)$ for $\csg$. Define $\adv{r} := (\adv{r_A}, \adv{r_S})$ for $\adv{\csg}$, where $\adv{r_A}(s,a) = r_A(s,a) \cdot \ind[s\in S]$ and $\adv{r_S}(s) = r_S(s) \cdot \ind[s\in S]$.
Correspondingly, we represent the reward associated with the $j$th step in $\adv{\pi}$ as 
\[
\adv{r}(\adv{\pi_j}) := \adv{r_S}(\adv{\pi}(j)) + \adv{r_A}(\adv{\pi}(j),\pi[j]) + \adv{r_S}(\adv{\pi}(j')) + \adv{r_A}(\adv{\pi}(j'),\adv{\pi}[j']).
\]

Finally, consider a player's \emph{objective} $X$. Define $\adv{X}$ as:
\begin{itemize}[nosep]
    \item Probabilistic reachability: 
    $\adv{X}(\adv{\pi}) 
    = \ind \left[\exists{j \in \nats}.{\left( \adv{\pi}(j) \in T \right)} \right]$;
    \item Expected reward reachability: 
    $\adv{X}(\adv{\pi}) = \sum_{i=0}^{k_{\min}-1}{\adv{r}(\adv{\pi_i})}$ if $\exists{j \in \nats}.\,{\adv{\pi}(j)\in T}$ and $\infty$ otherwise, where the value of $k_{\min}$ and $T$ are kept the same as in $X$.
\end{itemize}
\end{definition}

Since we are primarily concerned with the worst-case scenario, henceforth we write $\adv{\sigma} := \advPstar{\sigma}$ and $\adv{\pi} := \advPstar{\pi}$ unless otherwise stated.
In the following, let $IPaths_{\csg,s}^P$ and $FPaths_{\csg,s}^P$ denote the sets of infinite and finite paths, respectively, in the ICSG $\csg$ from state $s$ under uncertainty resolution $P\in \Punc$. We write $Paths_{\csg,s}^P := IPaths_{\csg,s}^P \union FPaths_{\csg,s}^P$ for the set of all such paths.

Using this definition, along with our use of primed indices for $S'$-states, leads to the following results:

\begin{proposition}[Path Bijection]
\label{prop:infh-path-bijection}
$Paths_{\adv{\csg},s}^{P} = \{ \advP{\pi} \mid \pi \in \gs{Paths}^{P} \}$ and 
$\gs{Paths}^{P} = \{ \pi \mid \advP{\pi} \in Paths_{\adv{\csg},s}^{P} \}$
\reforig{lem:infh-util-preserving-strategy-eq}
\end{proposition}

\begin{proposition}[Path index preservation]
\label{prop:infh-path-index-preserved}
    For any $\csg$-path $\pi\in Paths_{\adv{\csg},s}^{P}$, 
    $
    \forall{j\in\nats}. \ {\advP{\pi}(j) = \pi(j) \land \advP{\pi}[j] = (\pi[j], \idle)}.
    $
\end{proposition}

\begin{proposition}[Path Probability Preservation]
\label{prop:infh-prob-preserved}
    Given a $\csg$-profile $\sigma$ and nature's selection $P\in \Punc$, for every finite path $\pi \in \stratgs{FPaths}$, we have that 
    $\prob^\sigma(\pi) = \prob^{\advP{\sigma}}(\advP{\pi})$. 
    \reforig{lem:infh-util-preserving-strategy-eq}
\end{proposition}
\begin{proof}
\label{prf:infh-prob-preserved}
First note that in a standard CSG, a profile $\sigma$ induces a probability function $\prob^\sigma: \stratgs{FPaths} \rightarrow [0,1]$ as follows:
\begin{equation}
\label{eq:csg-path-prob}
\prob^\sigma(\pi) := \prod_{j=0}^{|\pi|-1}{\left[ \left( \prod_{i=1}^n{\sigma_i (\pi[0\ldots j])(\pi[j]_i)} \right)
\cdot P(\pi(j),\pi[j],\pi(j+1)) 
\right]}
\end{equation}
where $|\pi|$ is the number of transitions and $\pi[j]_i$ is player $i$'s action in the $j$th transition in $\pi$.

Consider a path $\pi= s_0 \xrightarrow{a^0} (s_0,a^0) \xrightarrow{(\idle,P_0)} s_1 \xrightarrow{a^1} (s_1,a^1) \xrightarrow{(\idle,P_1)} \ldots \allowbreak (s_{m-1},a^{m-1}) \xrightarrow{(\idle,P_{m-1})} s_m$ where $a^j = (a_1^j,a_2^j)$. 
The probability of $\advP{\pi}$ under profile $\advP{\sigma}$ is given by:
\begin{align*}
    \prob^{\advP{\sigma}}(\advP{\pi})
    &= \prod_{j=0}^{m-1}{ \left\{
            \begin{aligned}
                &\adv{\sigma_1} (s_j,a_1^j) 
                \advP{\sigma_2} (s_j,a_2^j) \\
                &\cdot \adv{\sigma_1} ((s_j,a^j),\idle)
                \advP{\sigma_2} ((s_j,a^j),P_j)
                \cdot P_j(s_{j+1})
            \end{aligned}
                \right\} } \\
    &= 
    \prod_{j=0}^{m-1}{ 
                \adv{\sigma_1} (s_j,a_1^j) 
                \cdot \advP{\sigma_2} (s_j,a_2^j) 
                \cdot P_{s_j a^j}(s_{j+1}) } \\
    &= 
    \prod_{j=0}^{m-1}{ 
                \sigma_1(s_j,a_1^j) 
                \cdot \sigma_2 (s_j,a_2^j) 
                \cdot P_{s_j a^j s_{j+1}} }
    \qquad \text{\small(by \AbbrCref{def:infh-adv-components} of $\advP{\sigma}$)} \\
    &= \prob^\sigma(\pi) 
\end{align*}
\end{proof}
This matches \Cref{eq:csg-path-prob}, with $\sigma$ being memoryless.

\begin{proposition}[Reward preservation]
\label{prop:infh-rew-preserved}
For any (finite or infinite) path $\pi$ and nature's choice of $P\in\Punc$, the adversarial expansion preserves per-step rewards, i.e., $\adv{r}(\advP{\pi_j}) = r(\pi_j)$ for any step $j$.
\end{proposition}
\begin{proof}
For any step $j$:
\begin{align*}
\adv{r}(\advP{\pi_j}) 
&= \adv{r_S}(\advP{\pi}(j)) + \adv{r_A}(\advP{\pi}(j), \pi[j]) \\
    &\qquad + \adv{r_S}(\advP{\pi}(j')) + \adv{r_A}(\advP{\pi}(j'), \advP{\pi}[j']) \\
&= r_S(s_j) + r_A(s_j, a^j) + r_S(s_j, a^j) + r_A((s_j, a^j), (\idle, P_j)) \\
&= r_S(s_j) + r_A(s_j, a^j)
= r_S(\pi(j)) + r_A(\pi(j), \pi[j]) 
= r(\pi_j)
\end{align*}
\end{proof}


\begin{proposition}[Objective Preservation]
\label{prop:infh-obj-preserved}
    For any (infinite or finite) $\csg$-path $\pi$ and nature's choice of $P\in \Punc$, we have that $X(\pi)=\adv{X}(\advP{\pi})$.
    \reforig{lem:infh-util-preserving-strategy-eq}
\end{proposition}
\begin{proof}
\label{prf:infh-obj-preserved}
Using \Cref{prop:infh-path-index-preserved,prop:infh-rew-preserved} we have:
\begin{align*}
    \text{Prob. reachability:} \ 
    \adv{X}(\advP{\pi}) 
    &= \ind[\exists{j \in \nats}.{\advP{\pi}(j)\in T}] \\
    &= \ind[\exists{j \in \nats}.{\pi(j)\in T}] 
    = X(\pi)\\
    \text{Rew. reachability:} \ 
    \adv{X}(\advP{\pi}) 
    &= \begin{cases}
    \sum_{i=0}^{k_{\min}-1}{\adv{r}(\advP{\pi_i})} &\ift \exists{j \in \nats}.{\advP{\pi}(j)\in T},\\
    \infty & \otherwiset
    \end{cases}\\
    &= \begin{cases}
    \sum_{i=0}^{k_{\min}-1}{
        r(\pi_i)} &\ift \exists{j \in \nats}.{\pi(j)\in T},\\
    \infty & \otherwiset
    \end{cases}
    = X(\pi).
\end{align*}
\end{proof}

\subsection{Value Preservation under Adversarial Expansion}
\label{sec:value-preservation-proofs}



\begin{corollary}
\label{cor:infh-single-side-val-equiv}
    For any starting state $s\in S$:
    \begin{align}
    \label{eq:infh-supinfinf-eq-supinf}
    \sup_{\sigma_1\in \Sigma_1}{\inf_{\sigma_2\in \Sigma_2}{\inf_{P\in \Punc}{\gs{\ev}^{(\sigma_1, \sigma_2), P}}[X]}}
    &= 
    \sup_{\adv{\sigma_1} \in \adv{\Sigma_1}}{\inf_{\adv{\sigma_2}\in \adv{\Sigma_2}}{\ev_{\adv{\csg},s}^{\adv{\sigma_1}, \adv{\sigma_2}}[X]}}\\
    \inf_{\sigma_2\in \Sigma_2}{\sup_{\sigma_1\in \Sigma_1}{\inf_{P\in \Punc}{\gs{\ev}^{(\sigma_1, \sigma_2), P}}[\adv{X}]}}
    &=
    \inf_{\adv{\sigma_2}\in \adv{\Sigma_2}}{\sup_{\adv{\sigma_1}\in \adv{\Sigma_1}}{\ev_{\adv{\csg},s}^{\adv{\sigma_1}, \adv{\sigma_2}}[\adv{X}]}}
    \end{align}
    \label{eq:infh-infsupinf-eq-infsup}
\end{corollary}
\begin{proof}
Without loss of generality we only prove the first equality \eqref{eq:infh-supinfinf-eq-supinf}, as the proof for \eqref{eq:infh-infsupinf-eq-infsup} is quite similar. 
We first show that for any $\sigma_1 \in \Sigma_1$:
\begin{equation}
\label{eq:inner-inf-equals}
\inf_{\sigma_2\in \Sigma_2}{\inf_{P\in \Punc}{\gs{\ev}^{(\sigma_1, \sigma_2), P}}[X]} = 
\inf_{\adv{\sigma_2}\in \adv{\Sigma_2}}{\ev_{\adv{\csg},s}^{\adv{\sigma_1}, \adv{\sigma_2}}[\adv{X}]}
\end{equation}
For brevity let 
\[
f(\sigma_2 \mid \sigma_1) := \inf_{P\in \Punc}{\gs{\ev}^{(\sigma_1, \sigma_2), P}[X]}
\quad \text{and} \quad
\adv{f}(\adv{\sigma_2} \mid \adv{\sigma_1}) := \ev_{\adv{\csg},s}^{\adv{\sigma_1}, \adv{\sigma_2}}[\adv{X}]
\]
and notice that $f(\sigma_2 \mid \sigma_1) = \adv{f}(\adv{\sigma_2} \mid \adv{\sigma_1})$ by \Cref{lem:infh-util-preserving-strategy-eq}.

Fix an arbitrary $\sigma_1 \in \Sigma_1$, and consider the corresponding $\sigma_2^* \in\Sigma_2$ that minimises $f(\sigma_2 \mid \sigma_1)$.
By \Cref*{lem:infh-util-preserving-strategy-eq}, there exists an equivalent $\adv{\csg}$-strategy $(\adv{\sigma_1}, \adv{{\sigma_2^*}})$ that preserves player 1's expected utility, i.e., $f(\sigma_2^*\mid \sigma_1) = \adv{f}(\adv{{\sigma_2^*}} \mid \adv{\sigma_1})$. 
Thus we have:
\[
\inf_{\sigma_2\in \Sigma_2}{f(\sigma_2 \mid \sigma_1)}
= 
f(\sigma_2^*\mid \sigma_1)
= \adv{f}(\adv{{\sigma_2^*}} \mid \adv{\sigma_1})
\geq  \inf_{\adv{\sigma_2}\in \adv{\Sigma_2}}{\adv{f}(\adv{\sigma_2} \mid \adv{\sigma_1})}.
\]

The other direction holds for the same reasoning. Therefore \eqref{eq:inner-inf-equals} holds for any $\sigma_1\in \Sigma_1$.
Further, by the strategy bijection between $\csg$ and $\adv{\csg}$ (\Cref*{lem:infh-util-preserving-strategy-eq}), this equality must hold for all $\sigma_1 \in \Sigma_1$ and $\adv{\sigma_1}\in \adv{\Sigma_1}$.

Next, we add on the outer $\sup$. 
Let 
\[
g(\sigma_1) := \inf_{\sigma_2\in \Sigma_2}{\inf_{P\in \Punc}{\gs{\ev}^{(\sigma_1, \sigma_2), P}}[X]}
\quad \text{and} \quad
\adv{g}(\adv{\sigma_1}) := \inf_{\adv{\sigma_2}\in \adv{\Sigma_2}}{\ev_{\adv{\csg},s}^{\adv{\sigma_1}, \adv{\sigma_2}}[\adv{X}]}.
\]
Consider $\sigma_1^*\in \Sigma_1$ which maximises $g(\sigma_1)$. 
Again by \Cref*{lem:infh-util-preserving-strategy-eq} there exists a $\adv{{\sigma_1^*}} \in \adv{\Sigma_1}$ such that $g(\sigma_1^*) = \adv{g}(\adv{{\sigma_1^*}})$. 
We thus have the following: 
\[
g(\sigma_1^*)
= \sup_{\sigma_1\in \Sigma_1}{g(\sigma_1)}
= \adv{g}(\adv{{\sigma_1^*}})
\leq  \sup_{\adv{\sigma_1} \in \adv{\Sigma_1}}{\adv{g}(\adv{\sigma_1})}.
\]
The other direction holds by the same reasoning.
Hence \Cref{eq:infh-supinfinf-eq-supinf} holds.

\end{proof}

\DetValPreservation*
\begin{proof}
\label{prf:infh-adv-value-preservation}
%
For the forward direction, by \Cref{def:adv-robust-determinacy-optimality-value} of the game value and \Cref{cor:infh-single-side-val-equiv}, for any $s\in S$ we have:
\begin{equation}
\label{eq:sup-inf-value-preservation}
\valg(s,X) 
= \sup_{\sigma_1\in \Sigma_1}{\inf_{\sigma_2\in \Sigma_2}{\inf_{P\in \Punc}{\gs{\ev}^{(\sigma_1, \sigma_2), P}}[X]}}
= \sup_{\adv{\sigma_1} \in \adv{\Sigma_1}}{\inf_{\adv{\sigma_2}\in \adv{\Sigma_2}}{\ev_{\adv{\csg},s}^{\adv{\sigma_1}, \adv{\sigma_2}}[\adv{X}]}}
\end{equation}
If $\csg$ is determined (see \Cref{def:adv-robust-determinacy-optimality-value}), the dual characterisation also holds. Hence:
\begin{align*}
\valg(s,X) 
&= \inf_{\sigma_2\in \Sigma_2}{\sup_{\sigma_1\in \Sigma_1}{\inf_{P\in \Punc}{\gs{\ev}^{(\sigma_1, \sigma_2), P}}[X]}} 
= \inf_{\adv{\sigma_2}\in \adv{\Sigma_2}}{\sup_{\adv{\sigma_1}\in \adv{\Sigma_1}}{\ev_{\adv{\csg},s}^{\adv{\sigma_1}, \adv{\sigma_2}}}[\adv{X}]} \\
&= \sup_{\adv{\sigma_1} \in \adv{\Sigma_1}}{\inf_{\adv{\sigma_2}\in \adv{\Sigma_2}}{\ev_{\adv{\csg},s}^{\adv{\sigma_1}, \adv{\sigma_2}}[\adv{X}]}} 
\qquad \text{\small(by \AbbrCref{eq:sup-inf-value-preservation})}\\
&= \vag(s,\adv{X}).
\end{align*}

The reverse direction holds by the same reasoning.
\end{proof}

\RNEnNEBijection*
\begin{proof}
\label{prf:infh-adv-rne-to-ne}
We prove both directions.

\startpara{RNE$_\csg$ $\Rightarrow$ NE$_{\adv{\csg}}$}
Consider an RNE $\sigma^* = (\sigma_1^*,\sigma_2^*) \in \Sigma$ in $\csg$, with the corresponding $\adv{\csg}$-profile $\adv{{\sigma^*}} \in \adv{\Sigma}$. 
Let $\adv{\sigma_1} \in \adv{\Sigma_1}$ be any deviation in $\adv{\csg}$ for player 1.
By \Cref{lem:infh-util-preserving-strategy-eq}, there exists $\sigma_1\in \Sigma$ in $\csg$ such that:
\begin{align*}
    \adv{u_1}(\adv{{\sigma^*}}) - \adv{u_1}((\adv{\sigma_1}, \adv{{\sigma_2^*}})) 
    &= \inf_{P\in \Punc}{u_1(\sigma_1^*, P)} 
    - \inf_{P\in \Punc}{u_1((\sigma_1,\sigma_2^*), P)} 
    \\
    &\geq \inf_{P\in \Punc}{[u_1(\sigma^*, P) - u_1((\sigma_1,\sigma_2^*),P)]} \\
    &\geq 0 
    \qquad \text{\small($\sigma^*$ is an RNE)}.
\end{align*}

Analogously, for any $\adv{\sigma_2}\in \adv{\Sigma_2}$:
\begin{align*}
    \adv{u_2}(\adv{{\sigma^*}}) - \adv{u_2}((\adv{{\sigma_1^*}}, \adv{\sigma_2})) 
    &= \sup_{P\in \Punc}{u_2(\sigma^*, P)} 
    - \sup_{P\in \Punc}{u_2((\sigma_1^*,\sigma_2), P)} 
    \\
    &\geq \inf_{P\in \Punc}{[u_2(\sigma^*, P) - u_2((\sigma_1^*,\sigma_2),P)]} \\
    &\geq 0 
    \qquad \text{\small($\sigma^*$ is an RNE)}.
\end{align*}
Hence $(\adv{{\sigma_1^*}}, \adv{{\sigma_2^*}})$ is an NE in $\adv{\csg}$. The value equivalence follows from \Cref{cor:infh-adv-value-preservation}.

\startpara{NE$_{\adv{\csg}}$ $\Rightarrow$ RNE$_\csg$}
Suppose for contradiction that $\adv{\sigma}$ is an $\adv{\csg}$-NE with value $\vag(s,\adv{X})$ equal to $\valg(s,X)$ (by \Cref{cor:infh-adv-value-preservation}), 
but $\sigma$ is not an $\csg$-RNE. 
Since $\sigma$ is not an RNE, by definition there exists a beneficial deviation $\sigma_i'\in \Sigma_i$ such that for:
\begin{itemize}[nosep]
    \item $i=1$: 
    There exists a $P\in \Punc$ such that $u_1(\sigma,P) < u_1((\sigma_1',\sigma_2),P)$, which implies that $\inf_{P\in \Punc}{[u_1(\sigma,P) - u_1((\sigma_1',\sigma_2),P)]} < 0$.
    This contradicts \Cref{def:adv-robust-determinacy-optimality-value} of the value of $\csg$, specifically 
    $\valg(s,X) = 
    \sup_{\sigma_1}{\inf_{\sigma_2}{\inf_{P}{\gs{\ev}^{(\sigma_1, \sigma_2), P}}[X]}}$, 
    as we could improve it by taking $\sigma_1',\sigma_2, P$.
    
    \item $i=2$: 
    There exists a $P\in \Punc$ such that $u_2(\sigma,P) < u_2((\sigma_1,\sigma_2'),P)$, which implies that $\inf_{P\in \Punc}{[u_2(\sigma,P) - u_2((\sigma_1,\sigma_2'),P)]} < 0$. Since $u_2(\sigma,P)=-u_1(\sigma,P)$, we have that $u_1(\sigma,P) > u_1((\sigma_1,\sigma_2'),P)$, which again contradicts the definition of the game's value, $\valg(s,X) = 
    \inf_{\sigma_2}{\sup_{\sigma_1}{\inf_{P}{\gs{\ev}^{(\sigma_1, \sigma_2), P}}[X]}}$, 
    as we could reduce it by taking $\sigma_2',\sigma_1,P$.
\end{itemize}
Hence $\sigma$ must be an RNE with a corresponding value equal to $\vag(s,\adv{X})$.
    
\end{proof}

\subsection{Finite-action Adversarial Expansion}
\label{sec:infh-fin-act-resolution}
The previous construction of $\adv{\csg}$ yields an \emph{infinite-action} CSG, as player 2's choices at the auxiliary $S'$-states may include uncountably many distributions, since $\adv{A_2}(s,a)=\Punc_{sa} = \bigtimes_{s' \in S}{[\Pcheck_{sas'}, \Phat_{sas'}]}$. 
However, since an ICSG is \emph{polytopic}, for each state-action pair $(s,a)$, $\Punc_{sa}$ forms a convex polytope which can be captured by randomising over its finite set of extrema $\verts[\Punc_{sa}]$. 
The following result stipulates that we can consider only the vertex distributions of $\Punc_{sa}$ without loss of optimality:

\begin{lemma}[Optimality of Extrema]
\label{lem:optimality-of-extrema}
    In a zero-sum ICSG, there exists a minimising transition function $P^* \in \Punc$ such that at any state $s$ under joint action $a$, $P^*_{sa}$ lies at a vertex of $\Punc_{sa}$.
\end{lemma}
\begin{proof}
    The Bellman equation \eqref{eq:infh-adv-bellman} shows that nature's minimisation objective is linear in $P$. From linear programming theory \cite{dantzig2002LP}, a linear function minimised over a convex polytope attains its minimum at a vertex. Then by $(s,a)$-rectangularity, there must exist an optimal $P^*\in \Punc$ such that each $P^*_{sa}$ occurs at a vertex of the polytope $\Punc_{sa}$.
\end{proof}

This result allows us to consider the \emph{finite-action} version of $\adv{\csg}$ as defined in \Cref{def:infh-adv-resolution}, which restricts player 2's action set at the $S'$-states, i.e., nature's action set, from $\Punc_{sa}$ to $\verts[\Punc_{sa}]$.

\subsection{Player/Nature-first Value Equivalence}
\label{sec:nature-first-ICSG-supp}

\PlayerNatureEq*
\begin{proof} 
\label{prf:infh-nature-or-player-first-val-equiv}
Since a value exists for $\csg$, by determinacy and value preservation (\Cref{cor:infh-adv-value-preservation}), $\adv{\csg}$ is also determined with the same value. Then:
\begin{align*}
    LHS
    &= \sup_{\sigma_1} \inf_{\sigma_2} \inf_{P}{\gs{\ev}^{(\sigma_1, \sigma_2), P}[X]}
    = \sup_{\adv{\sigma_1}} \inf_{\adv{\sigma_2}} \ev^{\adv{\sigma_1}, \adv{\sigma_2}}_{\adv{\csg}, s}[\adv{X}] 
    &&\text{\small(by \AbbrCref{cor:infh-adv-value-preservation})} \\
    &= \inf_{\adv{\sigma_2}} \sup_{\adv{\sigma_1}} \ev^{\adv{\sigma_1}, \adv{\sigma_2}}_{\adv{\csg}, s}[\adv{X}] 
    && \text{\small(by determinacy of $\adv{\csg}$)} \\
    &= \inf_{(\sigma_2, P)} \sup_{\sigma_1} \gs{\ev}^{(\sigma_1, \sigma_2), P}[X] 
    && \text{\small(by construction of $\adv{\csg}$)} \\
    &= \inf_{P} \inf_{\sigma_2} \sup_{\sigma_1} \gs{\ev}^{(\sigma_1, \sigma_2), P}[X] \\
    &= \inf_{P} \sup_{\sigma_1} \inf_{\sigma_2} \gs{\ev}^{(\sigma_1, \sigma_2), P}[X] 
    = RHS
\end{align*}
The last line follows because, the CSG induced by fixing any $P\in\Punc$ is finite-state and finitely-branching, which is known to be determined for all the objectives we consider \cite{martin1998determinacy}.

\end{proof}


\section{Proofs for Finite-horizon Zero-sum ICSGs}
\label{sec:finh-supp}

We now add a parameter $h$ to the value function $V$, which denotes the number of steps remaining until the horizon of the objective is completed. The robust Bellman equation is then defined as:
\[
\resizebox{\linewidth}{!}{$
\begin{aligned}
V(s,h) 
= \sup_{\sigma_1\in \distr(A_1(s))} \inf_{\sigma_2\in \distr(A_2(s))} {
\left\{ r_{s}^{\sigma} + \sum_{a\in A(s)}{\sigma_{sa} \inf_{P_{sa}\in \Punc_{sa}}{\sum_{s'\in S}{P_{sas'}\cdot V(s', h-1)}}} \right\}
}.
\end{aligned}
$}
\]

Henceforth in the zero-sum setting, we let $\Sigma_i$ denote the set of \emph{memoryless} strategies for each player $i\in \{1,2\}$, and interpret $\Punc$ as the set of transition functions resulting from \emph{memoryless} nature strategies.
\begin{lemma}[Strategy class sufficient for optimality]
\label{lem:finh-memoryfull-optimal-strategies}
     Given a \emph{finite-horizon} objective $X$ for $\csg$, each player has a \emph{time-varying} robust optimal strategy that depends only on 
     \begin{enumerate*}[label=\arabic*)]
        \item the current state $last(\pi)$, and
        \item the current time-step $|\pi|$;
    \end{enumerate*}
    and nature has a \emph{deterministic time-varying} optimal strategy.
\end{lemma}
\begin{proof}
\label{prf:finh-memoryfull-optimal-strategies}
Let $P^*: FPaths_{\csg} \rightharpoonup \distr(S)$ be the witnessing optimal transition function for nature, i.e., one that attains the inner minimum in the previous expression for $\valg(s,k)$. 
Let $\pi_1, \pi_2 \in FPaths_{\csg}$ be two finite paths such that $|\pi_1|=|\pi_2|=h \in [0,k]$ and $last(\pi_1)=last(\pi_2)=s\in S$.

Then for any $\sigma\in \Sigma$ we have:
\begin{align*}
    P^*(\pi_1)
    &= \arg\min_{P\in \Punc}{X(\pi_1) + \valg^{\sigma,P}(s,k-h)} 
    &&\text{\small(by additivity of objectives in \AbbrCref{lst:objs})}\\
    &= \arg\min_{P\in \Punc}{\valg^{\sigma,P}(s,k-h)} 
    &&\text{\small($X(\pi_1)$ is fixed)}\\
    &= \arg\min_{P\in \Punc}{X(\pi_2) + \valg^{\sigma,P}(s,k-h)}\\
    &= P^*(\pi_2)
\end{align*}
This shows that $P^*$ assigns the same next-state distribution to any two (finite) paths ending in the same state and of the same length. Hence, finite memory on the current state and time-step is \emph{sufficient} for $P^*$ to be optimal. 

To show that this time-dependence is also \emph{necessary}, we refer to \Cref{eg:finh-needs-timestep}, which illustrates that if the time-step is not remembered, the resulting transition function may behave suboptimally. There, the optimal strategy at the same state $s$ changes depending on the number of steps remaining of the objective's horizon. Therefore, memory of the time-step cannot be omitted.

The same reasoning applies to each player's strategy, with nature's objective $\min_{P} \valg^{\sigma, P}(s,k)$ replaced by $\max_{\sigma_1} \valg^{\sigma, P}(s,k)$ for player 1 and $\min_{\sigma_2} \valg^{\sigma, P}(s,k)$ for player 2.

The proof that nature admits a \textit{deterministic}, time-varying optimal strategy is analogous to the proof of determinism in~\Cref{prf:infh-memoryless-optimal-strategies}.
\end{proof}

\begin{example}
\label{eg:finh-needs-timestep}
Consider the ICSG in \Cref{fig:finh-eg}, where player 1 aims to reach her target state $s_2$ within $k=2$ steps. We analyse nature's choice of transition probabilities from state $s_0$, under the objective of minimising player 1's probability of reaching $s_2$.
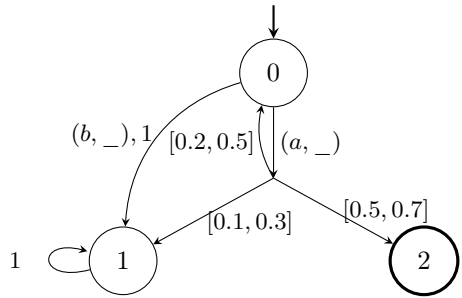
\begin{figure}[h]
    \centering
    \begin{tikzpicture}[->, >=stealth,
        every node/.style={draw, circle, minimum size=9mm, inner sep=1pt, font=\normalsize},
        labelstyle/.style={scale=1.0, fill=none, draw=none, font=\footnotesize}
      ]
    
      \node[draw, circle] (s0) at (0, 0) {0};
      \node[draw, circle] (s1) at (-2, -2.5) {1};
      \node[draw, circle, very thick] (s2) at (2, -2.5) {2};
    
      \draw[->, thick] (0, 0.9) -- (s0);
    
      \coordinate (mid) at (0, -1.4);
      \draw (s0) -- node[labelstyle, right] {$(a,\_)$} (mid);
    
      \draw (mid) to[bend left=20] node[labelstyle, left] {$[0.2, 0.5]$} (s0);
      \draw (mid) -- node[labelstyle, right, pos=0.6, yshift=-2pt] {$[0.1, 0.3]$} (s1);
      \draw (mid) -- node[labelstyle, pos=0.5, xshift=20pt] {$[0.5, 0.7]$} (s2);
    
      \draw (s0) to[bend right=35] node[labelstyle, left] {$(b,\_), 1$} (s1);
    
      \draw (s1) edge[loop left] node[labelstyle, left] {$1$} (s1);
    
    \end{tikzpicture}
    
    \caption{Zero-sum ICSG with player actions $A_1 = A_2 = \{a, b\}$. Player 1's objective is to maximise the probability of reaching a target set $T = \{s_2\}$ within $k=2$ steps.}
    \label{fig:finh-eg}
\end{figure}

At time-step 1 (with 1 step remaining), nature aims to minimise the immediate probability of transitioning to $s_2$. One optimal distribution is then:
\[
P_{s_0(a,\_)}(s) = \begin{cases}
    0.4 & s = s_0, \\
    0.1 & s = s_1, \\
    0.5 & s = s_2.
\end{cases}
\]

However, this distribution is \emph{suboptimal} at step 0 (with 2 steps remaining). In this case, nature must also account for the possibility that player~1 remains in $s_0$ for one step and reaches $s_2$ in the second. To minimise the total probability of reaching $s_2$, nature should increase the probability of transitioning to the sink state $s_1$, which prevents further progress. Thus the unique minimising distribution at time-step $0$ is:
\[
P_{s_0 (a,\_)}(s) = \begin{cases}
    0.2 & s = s_0, \\
    0.3 & s = s_1, \\
    0.5 & s = s_2.
\end{cases}
\]
\end{example}

\subsection{Player-first Adversarial Expansion}
\label{sec:finh-adv-resolution}

We extend the infinite-horizon adversarial expansion from \Cref{def:infh-adv-resolution} by augmenting states with the remaining horizon.

\begin{definition}[Adversarial Expansion]
\label{def:finh-adv-resolution}
    Given a \emph{$k$-step bounded} objective $X$, we define the \emph{adversarial expansion} of $\csg$ as a 2-player CSG $\advk{\csg} = (N, \advk{S}, \advk{\sbar}, \advk{A}, \advk{\Delta}, \advk{P})$ where:
    \begin{itemize}[nosep]
        \item $\advk{S}= \Si \union \Si'$, where $\Si = \bigcup_{h=0}^k{S_h}$ with $S_h = \{(s,h) \mid s\in S\}$, and 
        $\Si' = \bigcup_{h=0}^k{S_h'}$ with $S_h' = \{(s,a,h) \mid s\in S, a \in A(s)\}$.
        \item $\advk{\sbar} = (\sbar,k) \in S_k$;
        \item 
        $\advk{A} = (\advk{A_1} \union \idleset) \times (\advk{A_2} \union \idleset)$ where 
        $\advk{A_1}=A_1$ and $\advk{A_2}=A_2\union \bigcup_{s\in S, a\in A}{\verts[\Punc_{sa}]}$;

        \item $\advk{\Delta}: \advk{S} \rightarrow 2^{\advk{A_1}\union \advk{A_2}}$, such that if $\exists{h\in [0,k]}.\,\advk{s}=(s,h)\in \Si$ then $\advk{\Delta}(\advk{s})=\Delta_1(s)\union \Delta_2(s)$, else if $\exists{h\in [0,k]}.\,{\advk{s}=(s,a,h)\in \Si'}$ then $\advk{\Delta}(\advk{s})=\verts[\Punc_{sa}]$, else $\advk{\Delta}(\advk{s})=\emptyset$;
        \item 
        $\advk{P}: \advk{S} \times \advk{A} \rightarrow \distr(\advk{S})$ where 
        \[
        \resizebox{\linewidth}{!}{$
        \advk{P}(\advk{s},\advk{a},\advk{s'}) = 
        \begin{cases}
            1 & 
            \begin{aligned}
            &\ift \exists{h\in [1,k]}.\left[ 
                \begin{aligned}
                    &\advk{s} = (s,h) \in S_h
                    \land \advk{a} = a \in A(s) \\
                    &\land \advk{s'} = (s,a,h) \in S_h'
                \end{aligned} \right],
            \end{aligned}
            \\
            1 & \elseift {\advk{s}=\advk{s'}=(s,0) \in S_0 }, \\ 
            P_{sas'}
                    &\elseift \exists{h\in [1,k]}.\,
                    \left[
                    \begin{aligned}
                        &\advk{s} = (s,a,h) \in S_h' \\
                        &\land \advk{a} = \left(\idle, P_{sa} \right) \in \advk{A}((s,a,h)) \\
                        &\land \advk{s'} = (s',h-1) \in S_{h-1}
                    \end{aligned} \right], \\
            0 & \otherwiset.
        \end{cases}
        $}
        \]
    \end{itemize}
\end{definition}

With $h\in [1,k]$ steps remaining of the horizon, a $\csg$-transition $s \xrightarrow{a} s'$ under nature's choice of $P\in \Punc$, corresponds to a two-step $\advk{\csg}$-transition $(s,h) \xrightarrow{a} (s,a,h) \xrightarrow{( \idle, P_{sa,k-h})} (s',h-1)$. 
At or after the $k$th step, the game loops at a state $(s,0)\in S_0$. 

Denote by $\gs{FPaths}^k$ the set of finite paths of length $k$ in $\csg$ starting from $s$. Since both players share the same horizon $k$, the analysis can be restricted to finite paths of length at most $k$, i.e., prefixes of $\gs{FPaths}^k$. These correspond to finite paths of even length at most $2k$, i.e., even-length prefixes of $FPaths_{\advk{\csg}, (s,k)}^{2k}$.

As established in \Cref{lem:finh-memoryfull-optimal-strategies}, for optimality it suffices to consider \emph{time-varying} strategies for the players of the form $\sigma_i: S \times A \times H \rightharpoonup \distr(A_i \union \idleset)$.
For nature (i.e., player 2 at the $\Si'$-states), it suffices to consider \emph{time-varying deterministic} strategies of the form $P: S \times A \times H \rightharpoonup \distr(S)$.
Since the time-step is already encoded in the augmented state space $\advk{S}$, any time-varying $\csg$-strategy corresponds to a memoryless (i.e., stationary) strategy in $\advk{\csg}$.

The notion of adversarial expansion then extends to rewards, strategies, paths and objectives in a similar fashion as the infinite-horizon case (see \Cref{sec:infh-adv-resolution}). We only replace $P^*_{sa}$ with $P^*_{sa,h}$, and augment states with the remaining horizon $h$. 
Thus we only present the path and objective mapping here for clarity.

\begin{definition}[Adversarial Expansion cont.]
\label{def:finh-adv-cont}
Consider a finite $\csg$-\emph{path} $\pi = s_0 \xrightarrow{a^0} s_1 \xrightarrow{a^1} s_2 \rightarrow \ldots s_k$ under nature's selection of $P\in \Punc$, and let $P_j:= P_{s_j a^j, j}$ for $0\leq j \leq k-1$, we define: 
\begin{align}
\label{eq:finh-adv-path}
\begin{aligned}
\advkP{\pi} 
&:= (s_0,k) \xrightarrow{a^0} (s_0,a^0,k) \xrightarrow{(\idle, P_0)} (s_1,k-1) \\
&\quad \xrightarrow{a^1} (s_1,a^1,k-1) \xrightarrow{(\idle,P_1)} (s_2,k-2) \rightarrow \ldots \xrightarrow{(\idle,P_{k-1})} (s_k,0).
\end{aligned}
\end{align}

Consider a player's \emph{objective} $X$. We define $\advk{X}$ as:
\begin{itemize}[nosep]
\item Bounded probabilistic reachability:\\
$\advk{X}(\advkP{\pi}) := \ind \left[\exists{j \leq k, j\in \nats}.{\left( \advkP{\pi}(j)=(s,k-j)\in S_{k-j} \land s\in T \right)} \right]$;
\item Bounded cumulative reward:
$\advk{X}(\advkP{\pi}) := \sum_{i=0\in \nats}^{k-1}{ \advk{r}(\advkP{\pi_i}) }$.
\end{itemize}
\end{definition}

Under this definition, we again have the following preservation results:
\begin{proposition}[Path Bijection] 
\label{prop:finh-path-bijection}
For all $0\leq h \leq k$, 
$FPaths_{\advk{\csg},(s,k)}^{2h,P} = \{ \advkP{\pi} \mid \pi \in \gs{FPaths}^{h,P} \}$ and 
$\gs{FPaths}^{h,P} = \{ \pi \mid \advkP{\pi} \in FPaths_{\advk{\csg},(s,k)}^{2h,P} \}$.
\end{proposition}

To explicitly distinguish nature states $S'$ from player states $S$ in an $\adv{\csg}$-path $\advP{\pi}$, we index $S'$-states with primed indices $\nats' := \{0', 1', \ldots\}$, such that index $j':= 2j+1$ for $j\in \nats$. This in turn allows the state and transition indices in $\pi$ to be preserved in $\advP{\pi}$:

\begin{proposition}[Path index preservation]
\label{prop:finh-path-index-preserved}
    $\advk{\csg}$ visits the same sequence of $S$-states as $\csg$:     $\forall{j\in\{0,1,\ldots,k \} }.{
        \left[ \advkP{\pi}(j) = (\pi(j), k-j) \land \advkP{\pi}[j] = (\pi[j], \idle)
        \right]}$.
\end{proposition}
The preservation of path probability, rewards and objectives then follows immediately from \Cref{prop:finh-path-bijection,prop:finh-path-index-preserved}. 
%
%
Subsequently, all results for the infinite-horizon objectives --- namely \Cref{thm:infh-nature-or-player-first-val-equiv}, \Cref{lem:infh-util-preserving-strategy-eq}, \Cref{cor:infh-adv-value-preservation}, \Cref{thm:infh-adv-rne-to-ne} --- also apply to the finite-horizon objectives. 
Therefore, we can also perform robust model checking and strategy synthesis of zero-sum ICSGs for these finite-horizon objectives using established methods for standard zero-sum CSGs in \cite{kwiatkowska2021automatic}.

\section{Proofs for Infinite-horizon Nonzero-sum ICSGs}
\label{sec:nz-ICSG-supp}

\subsection{Extended Player-first Adversarial Expansion}
\label{sec:nz-infh-adv-resolution-supp}

Since players may now have distinct reward structures, we denote player $i$'s reward structure as $r_i = (r_{i,A}, r_{i,S})$ for $i\in \{1,2\}$, and write $r = (r_1, r_2)$.

\begin{definition}[Adversarial Expansion cont.]
\label{def:nz-infh-adv-components}
    Consider an infinite $\csg$-\emph{path} $\pi = s_0 \xrightarrow{a^0} s_1 \xrightarrow{a^1} s_2 \rightarrow \ldots$ under nature's choice of $P \in \Punc$, and let $P_j := P_{s_j a^j}$, we define $\advP{\pi}$ as:
    \[
    \advP{\pi} := 
    s_0 \xrightarrow{(*a^0,\idle)} (s_0,a^0) \xrightarrow{(\idle,\idle, P_0^*)} s_1 \xrightarrow{(*a^1,\idle)} (s_1,a^1) \xrightarrow{(\idle,\idle,P_1^*)} s_2 \rightarrow \ldots
    \]

    Consider a $\csg$-\emph{profile} $\sigma=(\sigma_1,\sigma_2)$ under nature's selection $P\in \Punc$. We define $\advP{\sigma} = (\adv{\sigma_1}, \adv{\sigma_2}, \advP{\sigma_3})$, where for any $s\in S$: 
    \begin{itemize}[nosep]
        \item for player $l\in \{1,2\}$: $\adv{\sigma_l}(\adv{s},a_l)=\sigma_l(s,a_l)$ if $\adv{s}=s\in S \land a_l \in A_l(s)$ else $\adv{\sigma_l}(\adv{s},a_l)=\ind\left[\adv{s}\in S' \land a_l=\idle \right]$;
        \item $\advP{\sigma_3}(\adv{s},a_3) = \ind \left[ a_3 = P_{sa}^* \right]$ if $\adv{s}=(s,a)\in S'$, else $\advP{\sigma_3}(\adv{s},a_3) = \\ \ind\left[ \adv{s}\in S \land a_3=\idle \right]$.
    \end{itemize}
    We continue to write $\adv{\pi} := \advPstar{\pi}$ and $\adv{\sigma} := \advPstar{\sigma}$.\\

    Consider the \emph{reward structure} $r=(r_1,r_2)$ for $\csg$. We define $\adv{r} = (\adv{r_1}, \adv{r_2}, \adv{r_3})$ where each $\adv{r_l} = (\adv{r_{l,A}},\adv{r_{l,S}})$ such that:
    \begin{itemize}[nosep]
        \item for $l\in \{1,2\}$: $\adv{r_{l,A}}(s,(*a,\idle)) = r_{l,A}(s,a) \cdot \ind[s\in S]$, $\adv{r_{l,S}}(s) = r_{l,S}(s) \cdot \ind[s\in S]$;
        \item $\adv{r_{3,A}}(\cdot)=\adv{r_{3,S}}(\cdot) = 0$.
    \end{itemize}

    Consider the \emph{objective} $X_l$ of a player $l\in\{1,2\}$. We define $\adv{X}$ identically to \Cref{def:infh-adv-components}, with $\adv{X_3} := -\adv{\Xsum}$.
\end{definition}

Following \Cref{def:nz-infh-adv-resolution,def:nz-infh-adv-components}, all preservation results previously established for zero-sum ICSGs still hold in the nonzero-sum setting. 

\subsection{Strategy Injection from $\epsRNE$ to $\advepsNE$}
\label{sec:nz-rne-injection}

\EpsRNEtoNEInjection*
\begin{proof}
\label{prf:nz-infh-rne-to-ne}
Suppose $\sigma$ is an RNE in $\csg$. 
We first consider a player~$l \in \{1,2\}$ and a unilateral deviation of that player, $\sigma_l'\in \Sigma_l$. 
By \Cref{lem:nz-infh-util-preserving-strategy-eq}, we~have:
\begin{align*}
    \adv{u_l}(\adv{\sigma}) - \adv{u_l}(\adv{\sigma_{-l}}[\adv{\sigma_l'}]) 
    &= u_l(\sigma, P^*) - u_l(\sigma_{-l}[\sigma_l'], P^*)\\
    &\geq \inf_{P \in \Punc} \left[ u_l(\sigma, P) - u_l(\sigma_{-l}[\sigma_l'], P) \right]
    \geq -\varepsilon
\end{align*}
where the last inequality follows from the definition of an $\varepsilon$-RNE.

Now consider player 3 (nature). By construction, $\adv{\sigma_3}$ deterministically selects $P^* \in \Punc$ which minimises the total utility $\usum(\sigma, P)$. Therefore, any deviation $\adv{\sigma_3'}$ corresponds to selecting a different transition function $P' \in \Punc$ that is necessarily less or equally optimal as $P^*$. Then:
\begin{align*}
    &\quad \adv{u_3}(\adv{\sigma}) - \adv{u_3}(\advPp{\sigma})\\ 
    &= - \left( \adv{u_1}(\adv{\sigma}) + \adv{u_2}(\adv{\sigma}) \right) + 
    \left( \adv{u_1}(\advPp{\sigma}) + \adv{u_2}(\advPp{\sigma}) \right) \\
    &= -\inf_{P\in \Punc}{[u_1(\sigma,P)+u_2(\sigma,P)]} + 
    [u_1(\sigma,P')+u_2(\sigma,P')] 
    \qquad\text{\small(by \AbbrCref{lem:nz-infh-util-preserving-strategy-eq})}\\
    &\geq 0 \geq -\varepsilon
\end{align*}
Hence $\adv{\sigma}$ is an $\varepsilon$-NE in~$\adv{\csg}$ by definition.

\end{proof}

\subsection{Non-bijection between $\advepsNE$ and $\epsRNE$}
\label{sec:nz-rne-no-bijection}

\begin{example}[$\varepsilon$-NE$_{\adv{\csg}}$ $\not\Rightarrow$ $\varepsilon$-RNE$_{\csg}$]
\label{eg:nz-infh-epsNE-neq-epsRNE}
Consider a one-shot ICSG $\csg$ with actions $\{A,B\}$ and payoffs: $u(A,A)=(1,1), \ u(A,B)=(0.2,0.2), \ u(B,A)=(0.2,0.7), \ \\u(B,B)=(2p,1-p)$ where $p\in[0.2,0.4]$. In $\adv{\csg}$, nature would pick $p=0.2$ to minimise $2p+(1-p)=p+1$. Thus for $0 \leq \varepsilon <0.1$, $(B,B)$ is an $\varepsilon$-NE in $\adv{\csg}$ because if player~1 (2) deviates to $A$, she gets $0.2+\varepsilon<0.4$ ($0.7+\varepsilon<0.8$), so neither player deviates. However, $\varepsilon$-RNE requires considering \textit{all} $p$: at $p=0.4$, player~2 profits by deviating ($0.7>0.6+\varepsilon$). Hence $(B,B)$ is an $\varepsilon$-NE in $\adv{\csg}$ but not an $\varepsilon$-RNE in $\csg$.  
    
\end{example}

\subsection{Finite-action Adversarial Expansion}
\label{sec:fin-act-nz-infh}




In the nonzero-sum setting, computing (subgame-perfect) RSWNE may involve solving \emph{nonlinear} programs~\cite{kwiatkowska2020multi}, to which standard duality theory does not apply.
We thus establish the finite-action reduction here using a more general topological argument.

\begin{lemma}[Optimality of Extrema]
\label{lem:nz-optimality-of-extrema}
    In a nonzero-sum ICSG, a profile $\sigma\in \Sigma$ is an $\varepsilon$-RSWNE when nature selects from $\Punc$ iff it is an $\varepsilon$-RSWNE when nature is restricted to select from its set of extrema at each state-action pair, $\union_{s\in S, a\in A}{\verts[\Punc_{sa}]}$.
\end{lemma} 
\begin{proof}
\label{prf:nz-optimality-of-extrema}
First note that $\verts(\Punc)=\bigtimes_{s\in S,a\in A}{\verts[\Punc_{sa}]}$ by $(s,a)$-rectangularity of an ICSG. Thus, selecting a transition function $P\in \Punc$ is equivalent to selecting from $P\in \bigtimes_{s\in S,a\in A}{\verts[\Punc_{sa}]}$, which is in turn equivalent to selecting a distribution $P_{sa}\in \Punc_{sa}$ independently at each state-action pair $(s,a)\in S\times A$.

For a profile $\sigma=(\sigma_1,\sigma_2)$ and a unilateral deviation
$\sigma_i'\in\Sigma_i$, define the deviation gain
\[
u_i^\Delta(\sigma_i',P)
:= u_i(\sigma_{-i}[\sigma_i'],P) - u_i(\sigma,P).
\]

\paragraph{$\boldsymbol{(\Leftarrow)}$}
Since $\verts(\Punc)\subseteq\Punc$, the result follows by monotonicity.
Indeed, assume that $\sigma$ is an $\varepsilon$-RSWNE under $\verts(\Punc)$.
We verify that $\sigma$ satisfies the two $\varepsilon$-RSWNE conditions (see \Cref{def:rsw-rsc}) under the full uncertainty set $\Punc$:
\begin{enumerate}
    \item \textbf{$\varepsilon$-RNE condition}: 
    For all $i\in N=\{1,2\}$ and $\sigma_i'\in\Sigma_i$,
    \[
    \sup_{P\in \Punc}{u_i^\Delta(\sigma_i',P)}
    \geq  
    \sup_{P\in \verts(\Punc)}{u_i^\Delta(\sigma_i',P)}
    \geq -\varepsilon,
    \]
    so $\sigma$ remains an $\varepsilon$-RNE profile under $\Punc$.
    
    \item \textbf{RSW optimality condition}: For any $\sigma\in\RNE$,
    \[
    \inf_{P\in\verts(\Punc)} \usum(\sigma,P)
    \ge \inf_{P\in\Punc} \usum(\sigma,P),
    \]
    which implies
    \begin{align*}
        \max_{\sigma\in \RNE}\inf_{P\in \verts(\Punc)}{\usum(\sigma,P)} 
        &\geq \max_{\sigma\in \RNE}\inf_{P\in \Punc}{\usum(\sigma,P)} \\
        \implies 
        \arg\max_{\sigma\in \RNE}\inf_{P\in \verts(\Punc)}{\usum(\sigma,P)} 
        &\subseteq \arg\max_{\sigma\in \RNE}\inf_{P\in \Punc}{\usum(\sigma,P)} 
    \end{align*}
    Hence $\sigma$ is RSW-optimal under $\Punc$.
\end{enumerate}

\paragraph{$\boldsymbol{(\Rightarrow)}$}
Let $\sigma$ be an $\varepsilon$-RSWNE under $\Punc$. 
Each local uncertainty set $\Punc_{sa}$ is convex and compact, being a product of closed intervals. 
Since $S$ and $A$ are finite, their Cartesian product $\Punc := \bigtimes_{(s,a)\in S\times A} \Punc_{sa}$ is likewise convex and compact, by Tychonoff’s theorem \cite{tychonoff1930topologische} and preservation of convexity under Cartesian products \cite{boyd2004convex}. 
Hence, by Bauer's maximum principle \cite{bauer1958minimalstellen}, any continuous convex function $f$ on $\Punc$ attains its infimum at some extreme point, i.e.,
\begin{equation}
\label{eq:inf-Punc-equals-inf-verts-Punc}
    \inf_{P \in \verts(\Punc)}{f(P)} = \inf_{P\in \Punc}{f(P)}.
\end{equation}

We again verify that $\sigma$ satisfies the $\varepsilon$-RSWNE conditions under $\verts(P)$:
\begin{enumerate}

\item \textbf{$\varepsilon$-RNE condition}: 
Observe that $u_i(\sigma,P \mid s)=\sum_{a}{\sigma_{sa}\sum_{s'}{P_{sas'}u_i(\sigma,P\mid s')}}$ is linear in $P$.
Thus for fixed $i$ and $\sigma_i'$, the function $f(P) := -u_i^\Delta(\sigma_i',P)=u_i(\sigma, P) - u_i(\sigma_{-i}[\sigma_i'], P)$ is also linear, hence convex and continuous in $P$. 
By \Cref{eq:inf-Punc-equals-inf-verts-Punc}, we have that for all $i\in\{1,2\}, \sigma_i'\in \Sigma_i$:
\[
\inf_{P \in \verts(\Punc)}{-u_i^\Delta(\sigma_i',P)} = \inf_{P\in \Punc}{-u_i^\Delta(\sigma_i',P)} \geq -\varepsilon
\]
which confirms that $\sigma$ is an $\varepsilon$-RNE under $\verts(\Punc)$. 
Moreover, combined with the backward direction, we conclude that $\sigma$ is an $\varepsilon$-RNE iff it is an $\varepsilon$-RNE over $\verts(\Punc)$, i.e., $\epsRNE(\Punc) = \epsRNE(\verts(\Punc))$ where $\epsRNE(K)$ is the set of $\varepsilon$-RNE profiles under a set $K$.

\item \textbf{RSW optimality condition}:
The function $f(P):=\usum(\sigma,P)$ is linear in $P$, and therefore substituting in \Cref{eq:inf-Punc-equals-inf-verts-Punc} yields 
\[
\inf_{P \in \verts(\Punc)}{\usum(\sigma,P)}
= \inf_{P\in\Punc}{\usum(\sigma,P)}.
\]
Hence $\sigma$ attains the same robust social welfare under $\Punc$ and $\verts(\Punc)$, and RSW optimality is preserved.
\end{enumerate}

We conclude that $\sigma$ is an $\varepsilon$-RSWNE under $\verts(\Punc)$.
\end{proof}

\subsection{Player/Nature-first Value Equivalence}
\label{sec:nz-nature-first-supp}

A similar invariance result to \Cref{thm:infh-nature-or-player-first-val-equiv} holds in the nonzero-sum setting.
Note that we focus on memoryless strategies of the players and nature, as the same proof for \Cref{lem:infh-memoryless-optimal-strategies} holds in the nonzero-sum case.

\begin{restatable}[Player/nature-first Value Equivalence]{theorem}{PlayerNatureEqNZ}
\label{thm:nz-infh-nature-or-player-first-val-equiv}
From any $s \in S$, $\valg(s)$ is invariant under the player-first or nature-first semantics: 
\[
\sup_{\sigma \in \epsRNE}\inf_{P\in \Punc}{\gs{\ev}^{\sigma, P}[\Xsum]}
=
\inf_{P\in \Punc} \sup_{\sigma \in \epsRNE}{\gs{\ev}^{\sigma, P}[\Xsum]}.
\]
\end{restatable}
\begin{proof}
\label{prf:nz-infh-nature-or-player-first-val-equiv}
Recall that, under our notation, the optimal (robust) social welfare value is $\Vsum(s,X) = \sup_{\sigma\in \epsRNE}\inf_{P\in\Punc}{\gs{\ev}^{\sigma, P}[\Xsum]} = \sup_{\sigma\in \epsRNE}\inf_{P\in\Punc}{\usum(\sigma,P)}$. 
For brevity, we write $\advP{\sigma}$ for $(\adv{\sigma_1},\adv{\sigma_2},\advP{\sigma_3})$ in the following.
\begin{align}
\nonumber
LHS &= 
\sup_{\sigma \in \epsRNE}\inf_{P\in \Punc}{\usum(\sigma,P)}
= \sup_{\adv{\sigma} \in \advepsRNE}{[\adv{u_1}(\adv{\sigma}) + \adv{u_2}(\adv{\sigma}) ]}
\qquad \text{\small(by \AbbrCref{lem:nz-infh-util-preserving-strategy-eq})} \\ \nonumber
&= \sup_{(\adv{\sigma_1}, \adv{\sigma_2}, \_) \in {\advepsRNE}} \inf_{\advP{\sigma_3} \in \adv{\Sigma_3}}{\left[\adv{u_1}\left( \advP{\sigma} \right) + \adv{u_2}\left( \advP{\sigma} \right) \right]} 
\\ \nonumber
&\quad \text{\small($\adv{\sigma_3}$ minimises $\usum$ by definition; see \AbbrCref{def:nz-infh-adv-resolution})}\\ 
&= \inf_{\advP{\sigma_3}\in \adv{\Sigma_3}} \sup_{(\adv{\sigma_1}, \adv{\sigma_2}, \_) \in {\advepsRNE}} {\left[ \adv{u_1}\left( \advP{\sigma} \right) + \adv{u_2}\left( \advP{\sigma} \right) \right]}
\tag{$\ddagger$} \label{eq:nz-use-zs-minimax-thm}\\ \nonumber
&= \inf_{P\in \Punc} \sup_{\sigma \in \epsRNE}{[u_1(\sigma, P) + u_2(\sigma,P)]} 
= RHS \\ \nonumber
&\quad\text{\small(by construction of $\adv{\csg}$ and \AbbrCref{lem:nz-infh-util-preserving-strategy-eq})} 
\end{align}

In the above derivation, step~\eqref{eq:nz-use-zs-minimax-thm} follows from the structure of the adversarial expansion $\adv{\csg}$. Specifically, the zero-sum relationship between the coalition of players $N = \{1, 2\}$ and nature (player 3) allows us to reduce the finitely-branching CSG $\adv{\csg}$ into a 2-coalition CSG $\clt{\csg}$. This coalition game is defined as in \cite[Definition~13]{kwiatkowska2021automatic} for $\adv{\csg}$, and is effectively a \emph{zero-sum} 2-player CSG between a player $\clt{1}$ representing $N$ and another player $\clt{2}$ representing nature.

It follows from the finiteness of $\adv{\csg}$ that the coalition game $\clt{\csg}$ is also finitely-branching, and therefore \emph{determined} (i.e., value exists) for the class of objectives we consider \cite{martin1998determinacy}. Furthermore, as given in \cite[Definition~13]{kwiatkowska2021automatic}, the coalition reward structure $\clt{r}$ aggregates the individual rewards of players 1 and 2 in $\adv{\csg}$, i.e., $\clt{r} = \adv{r_1} + \adv{r_2}$. Consequently, by linearity of the expected value $\Vsum$ in the rewards (see \Cref{eq:nz-infh-bellman}), the value of $\clt{\csg}$, i.e., the expected utility of player $\clt{1}$, is exactly $\usum=\adv{u_1}+\adv{u_2}$. Finally, applying the definition of determinacy (see Definition 9 of \cite{kwiatkowska2021automatic}) justifies the $\sup$-$\inf$ swap in step \eqref{eq:nz-use-zs-minimax-thm}.

\end{proof}

\subsection{Filtering $\advepsNE$ for $\epsRNE$}
\label{sec:nz-infh-filter-ne-for-rne-supp}

\PureDevFiltersRNE*
\begin{proof}
\label{prf:nz-infh-pure-dev-to-filter-rne}
First note that the $\varepsilon$-RNE condition can be reformulated~as:
\begin{alignat*}{3}
    & &\inf_{P\in \Punc}{-u_i^\Delta(\sigma_i',P)} &\geq -\varepsilon 
    \quad \forall{i\in N, \sigma_i' \in \Sigma_i} \\
    &\iff&
    \inf_{P\in \Punc}\inf_{\sigma_i' \in \Sigma_i}{-u_i^\Delta(\sigma_i',P)} &\geq -\varepsilon \  
    \iff
    \sup_{P\in \Punc}\sup_{\sigma_i' \in \Sigma_i}{u_i^\Delta(\sigma_i',P)} &\leq \varepsilon
\end{alignat*}

It thus remains to show that the following holds for any $P\in \Punc$:
\begin{equation}
\label{eq:filter-rne-sup-sup}
    \sup_{\sigma_i' \in \Sigma_i}{u_i^\Delta(\sigma_i',P)} \leq \varepsilon
    \iff 
    \sup_{\eta_i \in \pureSigmai}{u_i^\Delta(\eta_i,P)} \leq \varepsilon
\end{equation}

In what follows, consider an arbitrary $P\in \Punc$.

\paragraph{$\boldsymbol{(\Rightarrow)}$}
This direction is straightforward: since every deterministic strategy is also a mixed strategy, i.e., $\pureSigmai \subseteq \Sigma_i$, we have that:
\[
\sup_{\eta_i \in \pureSigmai}{u_i^\Delta(\eta_i,P)} 
\leq
\sup_{\sigma_i' \in \Sigma_i}{u_i^\Delta(\sigma_i',P)}
\leq \varepsilon.
\]

\paragraph{$\boldsymbol{(\Leftarrow)}$}
For a fixed candidate $\varepsilon$-NE profile $\sigma \in \Sigma$, deviator $i \in \{1,2\}$ and an initial state $s \in S$, the expected utility under $\sigma$ is also fixed as:
\[
u_i(\sigma,P)=V_i(s,X\mid \sigma,P) =: V'.
\]

Assume the RHS of \eqref{eq:filter-rne-sup-sup} holds, or equivalently ${u_i(\sigma_{-i}[\eta_i],P) \leq V'+\varepsilon}$ for all $\eta_i \in \pureSigmai$. 
We aim to show that this implies the following:
\begin{equation}
\label{eq:mixed-dev-rne-cond}
u_i(\sigma_{-i}[\sigma_i'],P) \leq V'+\varepsilon \quad \forall{\sigma_i'\in \Sigma_i}.
\end{equation}

Consider an arbitrary mixed strategy $\sigma_i' \in \Sigma_i$. Define the deviated profile as $\sigma' := \sigma_{-i}[\sigma_i']$, and the shorthand notations $\sigma'_{sa}=\sigma'(s,a)$, $P_{sas'} = P_{sas'}$ and $V_i(s\mid \sigma') = V_i(s,X\mid \sigma, P)$. 
Further denote the other player as $j := i \mod 2 + 1$, and write $a = (a_i, a_j)$ if $i < j$, or $a = (a_j, a_i)$ otherwise.

The expected utility of player~$i$ under $\sigma'$ and $P$ is then:
\begin{align*}
u_i(\sigma',P)
&= V_i(s \mid \sigma')
= r_{s}^{\sigma'} + 
  \sum_{a\in A(s)} \sigma'_{sa}
    \sum_{s'\in S} P_{sas'} \, V_i(s' \mid \sigma') \\
&= \sum_{a\in A(s)} \sigma'_{sa}\, r_{sa}
  + \sum_{a\in A(s)} \sigma'_{sa}
    \sum_{s'\in S} P_{sas'} \, V_i(s' \mid \sigma') \\
&= \sum_{a_i\in A_i(s)} \sum_{a_j\in A_j(s)}
     \sigma_j(s,a_j)\, \sigma_i'(s,a_i)
     \Bigl[r_{sa} + \sum_{s'\in S} P_{sas'} \, V_i(s' \mid \sigma')\Bigr] \\
&= \sum_{a_i\in A_i(s)} \sigma_i'(s,a_i)
     \underbrace{\sum_{a_j\in A_j(s)} \sigma_j(s,a_j)
       \Bigl[r_{sa} + \sum_{s'\in S} P_{sas'} \, V_i(s' \mid \sigma')\Bigr]}_{=: C_{a_i}}.
\end{align*}
Hence
\[
u_i(\sigma',P)
= \sum_{a_i\in A_i(s)} \sigma_i'(s,a_i)\, C_{a_i}
\le \max_{a_i\in A_i(s)} C_{a_i},
\]
since $\sigma_i'$ is a probability distribution. 
Further let $a_i^* \in \arg\max_{a_i\in A_i(s)} C_{a_i}$, and let $\eta_i^*$ be a deterministic strategy s.t. $\eta_i^*(s,a_i^*)=1$. Then
\[
\max_{a_i} C_{a_i} = \sum_{a_i\in A_i(s)} \eta_i^*(s,a_i)\, C_{a_i}
= u_i(\sigma_{-i}[\eta_i^*],P) \le V' + \varepsilon
\]
where the last inequality is by assumption.

Since this applies to any $\sigma_i'\in \Sigma_i'$, \Cref{eq:mixed-dev-rne-cond} holds.
\end{proof}

\subsubsection{The Deviation IMDP.}
\label{sec:dev-IMDP-supp}
Let $j := i \bmod 2 + 1$ denote the other player, who continues to play their strategy $\sigma_j$ from the candidate profile $\sigma$. For notational convenience, we write $a = (a_i, a_j)$ if $i < j$ and $a = (a_j, a_i)$ with otherwise. In either case, let $\sigma' = \sigma_{-i}[\sigma_i']$ denote the $\csg$-profile in which player $i$ deviates unilaterally from $\sigma$.

\begin{definition}[Deviation IMDP]
\label{def:deviation-mdp}
    Consider a starting state $s$, deviator $i \in N$ and a $\csg$-profile $\sigma \in \Sigma$ such that $\adv{\sigma}\in \advNE$. We define the corresponding \emph{deviation IMDP} as an IMDP $\rve{\csg_{i,\sigma}} = (\adv{S}, s, \rve{A}, \rve{\Phat}, \rve{\Pcheck}, \rve{r})$ in which player $i$ acts as the maximising agent, and where:

    \begin{itemize}[nosep]
        \item $\rve{A} = A_i \union \idleset$, and let 
        $\rve{A}(\adv{s}) = A_i(s)$ if $\adv{s}=s \in S$ and $\idleset$ otherwise;
        \item $\rve{\Phat}, \rve{\Pcheck}: \adv{S} \times \rve{A} \times \adv{S} \rightharpoonup [0,1]$, such that 
        if $\adv{s}=s \in S \land \rve{a_i}=a_i\in A_i(s) \land s'=(s,(a_i,a_j))\in S'$ then $\rve{\Phat}(\adv{s},\rve{a_i},s')=\rve{\Pcheck}(\adv{s},\rve{a_i},s') = \sigma_j(s,a_j)$, 
        else if $\adv{s}=(s,a)\in S' \land \rve{a_i}=\idle \land s'\in S$ then $\rve{\Phat}(\adv{s},\rve{a_i},s')= \Phat_{sas'}$ and $\rve{\Pcheck}(\adv{s},\rve{a_i},s')=\Pcheck_{sas'}$, and otherwise $0$.
        \item $\rve{r}: \adv{S} \times \rve{A} \rightarrow \reals$, where $\rve{r}(\adv{s},\rve{a_i}) = \sum_{a_j\in A_j(s)}{\sigma_j(s,a_j) r_{sa}} - r_{s}^{\sigma}$ if $\adv{s}=s\in S \land \rve{a_i} = a_i \in A_i(s)$, and $\rve{r}(\adv{s},\rve{a_i}) = 0$ otherwise.
    \end{itemize}
\end{definition}

Given a strategy $\sigma_i'$ for player~$i$ and a transition function $P \in \Punc$ chosen by nature, the expected value $\rve{V}(s)$ from state $s \in \adv{S}$ satisfies the following recursion:
\begin{align}
\nonumber
\rve{V}(s \mid \sigma_i') 
&:= \rve{V}(s,\rve{X} \mid \sigma_i',P) \\
&= \rve{r}(s,\sigma_i') + \sum_{a_i\in \rve{A}(s)}{\sigma_i'(s,a_i) \sum_{s'\in \adv{S}}{\rve{P_{s a_i s'}}} \cdot \rve{V}(s'\mid \sigma_i')}
\label{eq:nz-infh-rve-val-fn}
\end{align}

It can be shown by induction on $\rve{V}(s\mid \sigma_i')$ over the structure of $\rve{\csg}$, that this definition coincides with player $i$'s deviation gain $u_i^\Delta$ \eqref{eq:rve-ui-delta}:

\begin{lemma}[$\rve{u}$ matches $u_i^\Delta$]
\label{lem:rve-util-correctness}
    Consider the deviation IMDP $\rve{\csg}$ for a given deviator $i\in \{1,2\}$ and $\csg$-profile $\sigma\in \Sigma$. 
    Under nature's chosen transition function $P\in \Punc$ and agent strategy $\sigma_i'$, $\rve{\csg}$ satisfies that $\rve{u}(\sigma_i',P) = u_i^\Delta(\sigma_i',P)$.
\end{lemma}
\begin{proof}
\label{prf:rve-util-correctness}
We show this by induction on the value function of $\rve{\csg}$ over states $\adv{S}$. 
Let $V(s)$ and $V'(s)$ denote the value of state $s$ in $\csg$ under strategies $\sigma$ and $\sigma' = \sigma_{-i}[\sigma_i']$, respectively.
Let $\rve{V'}$ be the value function in $\rve{\csg}$ under $\sigma_i'$ and $P$.
For brevity, we also write $r_{s \sigma}=r(s,\sigma)$, $\rve{r_{s \sigma_i'}}=\rve{r}(s,\sigma_i')$ and $\rve{u_{\sigma_i' P}} = \rve{u}(\sigma_i',P)$ in the following. 
From the starting state $s\in S$:
\begin{align}
\nonumber 
\rve{u_{\sigma_i' P}}
&= \rve{V'}(s)\\ \nonumber 
&
\begin{aligned}
    &= \rve{r_{s \sigma_i'}} + \sum_{\rve{a}\in \rve{A}(s)}{\sigma_i'(s,\rve{a}) \sum_{s'\in \adv{S}}{\rve{P_{s,\rve{a} s'}} \cdot \rve{V'}(s')}} 
    \qquad \text{\small(from \AbbrCref{eq:nz-infh-rve-val-fn})}\\
    &= \rve{r_{s \sigma_i'}} + 
    \sum_{a_i\in A_i(s)}{\sigma_i'(s,a_i) \sum_{(s,a) \in S'}{\sigma_j(s,a_j)} \cdot \rve{V'}((s,a))}\\
    &= \rve{r_{s \sigma_i'}} + 
    \sum_{a_i\in A_i(s)}{\sigma_i'(s,a_i) \sum_{a_j\in A_j(s)}{\sigma_j(s,a_j)} \cdot \rve{V'}((s,a))}
\end{aligned}
\\
&= \rve{r_{s \sigma_i'}} + 
\sum_{a\in A(s)}{\sigma_{s,a}' \cdot \rve{V}\left( (s,a)\right)}.
\label{eq:nz-infh-rve-s-val}
\end{align}

From each auxiliary state $(s,a) \in S'$:
\begin{align}
    \rve{V'}((s,a))
    = \rve{r}_{(s,a),\sigma_i'} + \sum_{s'\in \adv{S}}{\rve{P_{(s,a) \idle s'}} \cdot \rve{V'}(s')}
    = \sum_{s'\in S}{P_{sas'}} \rve{V'}(s')
\label{eq:nz-infh-rve-sa-val}
\end{align}

Combining \eqref{eq:nz-infh-rve-s-val} and \eqref{eq:nz-infh-rve-sa-val} we have:
\begin{align}
    \rve{u_{\sigma_i' P}}
    &= \rve{r_{s \sigma_i'}} + 
    \sum_{a\in A(s)}{\sigma_{s,a}' \cdot \sum_{s'\in S}{P_{sas'} \rve{V'}(s')}} \nonumber \\
    &= \rve{r_{s \sigma_i'}} 
    + \sum_{a}{\sigma_{s,a}' \sum_{s'}{P_{sas'} [V_i'(s') - V_i(s')] }} \nonumber 
    \qquad \text{\small(by induction hypothesis)} \nonumber \\
    &= r_{s \sigma'}-r_{s \sigma} 
    + \sum_{a}{\sigma_{sa}' \sum_{s'}{P_{sas'} V_i'(s') }} 
    - \sum_{a}{\sigma_{sa}' \sum_{s'}{P_{sas'}V_i(s')}} 
    &&\label{eq:rve-rew-simplified} \\ \nonumber
    &= \left( r_{s \sigma'} + 
    \sum_{a}{\sigma_{sa}' \sum_{s'}{P_{sas'} V_i'(s') }} \right) 
    - \left( r_{s \sigma} + 
    \sum_{a}{\sigma_{sa} \sum_{s'}{P_{sas'} V_i(s') }} \right)\\
    &= V_i'(s) - V_i(s) \nonumber 
    = u_i(\sigma',P) - u_i(\sigma,P) 
    = u_i^\Delta(\sigma_i',P) \nonumber
\end{align}
where \eqref{eq:rve-rew-simplified} follows from
\begin{align*}
    \rve{r_{s \sigma_i'}}
    &= \sum_{a_i\in A_i(s)} {\sigma_i'(s,a_i) \cdot \rve{r}(s,a_i)}\\
    &= \sum_{a_i\in A_i(s)} {\sigma_i'(s,a_i) \left( \sum_{a_j\in A_j(s)}{\sigma_j(s,a) r_{sa}} - r_{s \sigma} \right)}\\
    &= \sum_{a\in A(s)}{\sigma_{sa}' r_{sa}} - r_{s \sigma}
    = r_{s \sigma'} - r_{s \sigma}
\end{align*}
\end{proof}

\begin{corollary}
\label{cor:opt-rve-val-equals-max-ui-delta}
    Consider the deviation IMDP $\rve{\csg}$ for a given deviator $i\in \{1,2\}$ and $\csg$-profile $\sigma\in \Sigma$. 
    The \emph{optimistic} value of $\rve{\csg}$ is equal to the \emph{maximal deviation gain} of the deviator, i.e.,
    \[
    \rve{\Vbar_{i,\sigma}} := 
    \sup_{P\in\Punc}\sup_{\sigma_i'\in \Sigma_i}{\rve{u}(\sigma_i',P)}
    = \sup_{P\in \Punc}\sup_{\sigma_i'\in \Sigma_i}{u_i^\Delta(\sigma_i',P)}
    = \Vbar_{i,\sigma}.
    \]
\end{corollary}
\section{Proofs for Finite-horizon Nonzero-sum ICSGs}
\label{sec:nz-finh-supp}

\subsection{Player-first Adversarial Expansion}
\label{sec:nz-finh-adv-resolution}


Suppose player 1 has the shorter horizon $k_1$. Once she reaches her target set $T_1$ or exceeds her horizon, she should no longer receive rewards. Yet, her actions beyond $k_1$ may still affect transitions and, in turn, player 2's outcome. Therefore, the game must continue up to $k = \max(k_1, k_2)$.
To model this correctly:
\begin{itemize}[nosep]
    \item For \emph{bounded probabilistic reachability}, we add atomic propositions to track when states in $T_1$ are still relevant for player 1, i.e., only when they are within her horizon $k_1$ (see \Cref{def:nz-finh-adv-resolution}). In standard fashion, we use a set of atomic propositions $AP$ and labelling function $L$. For clarity, these are omitted from the definition of CSGs in the main paper.
    \item For \emph{bounded cumulative reward}, we assign player 1 a reward of zero in all steps beyond her horizon $k_1$.
\end{itemize}

We define a common horizon $k := \max(k_1, k_2)$ and adjust the reward structures and objective mappings to respect each player’s horizon (see \Cref{def:nz-finh-adv-reward-struct}).
Since the adversarial expansion of paths and strategies extends naturally from the game's construction, here we focus on defining the game composition, reward structure, and objective mappings.

\begin{definition}[Adversarial Expansion]
\label{def:nz-finh-adv-resolution}
    Given a pair of nonzero-sum objectives $X = (X_1, X_2)$, where $X_1$ and $X_2$ are $k_1$- and $k_2$-step bounded respectively. Let $k=\max{(k_1,k_2)}$.    
    We define the \emph{adversarial expansion} of $\csg$ as a 3-player CSG $\advk{\csg} = (\advk{N}, \advk{S}, \advk{\sbar}, \advk{A}, \advk{\Delta}, \advk{P}, \advk{AP}, \advk{L})$ where:
    \begin{itemize}[nosep]
        \item $\advk{N} = \{1,2,3\}$ and $\advk{A}$ is as defined in \Cref{def:nz-infh-adv-resolution};
        \item $\advk{S}= \Si \union \Si'$, where 
        $\Si = \bigcup_{h=0}^k{S_h}$ with $S_h = \{(s,h) \mid s\in S\}$, and 
        $\Si' = \bigcup_{h=1}^k{S_h'}$ with $S_h' = \{(s,a,h) \mid s\in S, a\in A(s) \}$;
        \item $\advk{\sbar} = (\sbar,k) \in S_k$;
        \item $\advk{\Delta}: \advk{S} \rightarrow 2^{\union_{i=1}^3{\advk{A_i}}}$, such that if $\advk{s}=(s,h)\in \Si$ then $\advk{\Delta}(\advk{s})=\Delta(s)$, else if $\advk{s}=(s,a,h)\in \Si'$ then $\advk{\Delta}(\advk{s})= \verts[\Punc_{sa}]$ and $\emptyset$ otherwise.
        \item 
        $\advk{P}: \advk{S} \times \advk{A} \rightarrow \distr(\advk{S})$, such that 
        \[
        \resizebox{\linewidth}{!}{$
        \advk{P}(\advk{s}, \advk{a}, \advk{s'}) = 
        \begin{cases}
            1 & \begin{aligned}
                    \ift & \exists{h\in [0,k]}.\left[
                    \begin{aligned}
                        &\advk{s} = (s,h) \in S_h \\
                        &\land \advk{a} = (*a,\idle) \\
                        & \land \advk{s'} = (s,a,h) \in S_h'
                    \end{aligned} \right],
                \end{aligned} \\
            1 & \elseift {\advk{s}=\advk{s'}=(s,0)\in S_0},\\
            P_{sa}(s')
                    &\elseift \exists{h\in [1,k]}. \
                    \left[
                    \begin{aligned}
                        & \advk{s} = (s,a,h)\in S_{h}' \\
                        &\advk{a} = (\idle, \idle, P_{sa}) \\
                        & \advk{s'} = (s',h-1) \in S_{h-1}'
                    \end{aligned} \right], \\
            0 & \otherwiset;
        \end{cases}
        $}
        \]

        \item $\advk{AP} = \{a_{T_1},a_{T_2}\}$ where each $a_{T_l}$ identifies a time-augmented state in $T_l$;
        \item $\advk{L}: \advk{S} \rightarrow 2^{\advk{AP}}$, where 
        \begin{equation}
        \label{eq:nz-finh-label}
        \resizebox{0.9\linewidth}{!}{$
            a_{T_l} \in \advk{L}(\advk{s})
            \iff \exists{h\in [\max(0,k-k_l), k]}. \ {\advk{s}=(s,h)\in S_{h} \land s\in T_l}.
        $}
        \end{equation}
    \end{itemize}
\end{definition}

%
We assign zero reward to the terminal states $S_0$. For bounded cumulative reward objectives, rewards are only accumulated up to step $k-1$ by definition (see \Cref{sec:prelim-CSG}), so states beyond this point do not contribute. For bounded probabilistic reachability, the objective is defined using state labels $L(s)$ instead of rewards.

\begin{definition}[Adversarial Expansion cont.]
\label{def:nz-finh-adv-reward-struct}
Consider a \emph{reward structure} $r=(r_A,r_S)$ for $\csg$. We define $\advk{r} = (\advk{r_A}, \advk{r_S})$, where for $l\in \{1,2\}$:
\begin{align}
\label{eq:nz-finh-adv-rew}
\begin{aligned}
\advk{r_{l,A}}(\advk{s}, \advk{a}) 
&= r_{l,A}(s,a) \cdot \ind \left[
    \begin{aligned}
        &\exists{h \in [\max(1,k-k_l), k]}.\,\advk{s} = (s,h) \in S_h \\
        &\land \advk{a} = (*a, \idle) \in \advk{A}(\advk{s})
    \end{aligned}
    \right] \\
\advk{r_{l,S}}(\advk{s}) 
&= r_{l,S}(s) \cdot \ind\Big[\exists{h \in [\max(1,k-k_l), k]}.\, \advk{s} = (s,h) \in S_h \Big]
\end{aligned}
\end{align}
and $\advk{r_{3,A}}(\cdot) = \advk{r_{3,S}}(\cdot) = 0$. 

Consider the \emph{objective} $X_l$ of player $l\in \{1,2\}$. We define $\advk{X_l}$ as follows:
\begin{itemize}[nosep]
    \item Bounded Probabilistic Reachability:
    
    $\advk{X_l}(\advk{\pi})
    := \ind \left[ \exists{j \leq k}.{\left\{ \advk{\pi}(j)=(s,k-j)\in S_{k-j} \land a_{T_l} \in \advk{L}((s,k-j)) \right\}} \right]$;
    \item Bounded Cumulative Reward: 
    $\advk{X_l}(\advk{\pi}) 
    := \sum_{i=0}^{k-1}{ \advk{r_l}(\advk{\pi_i}) }$. 
\end{itemize}
Player 3's objective is defined as $\advk{X_3}:=-(\advk{X_1}+\advk{X_2})$. 
\end{definition}
Under the above definition of $\advk{r}$ and $\advk{X}$, the valuation of rewards is preserved:

\begin{proposition}[Reward Preservation]
\label{prop:nz-finh-rew-preserved}
    For all $h\in [0,k]$ where $k=\max(k_1,k_2)$, 
    $\pi \in IPaths_{\csg}$, $j\in [0,h]$ and $l\in \{1,2\}$, we have that $\advk{r_l}(\advk{\pi_j}) = r_l(\pi_j)\cdot \ind[j \leq k_l]$.
\end{proposition}
\begin{proof}
If $j\in [0,h]$ and $h\in [0,k]$ then $j\in [0,k]$, so:
\begin{align*}
    \advk{r_l}(\advk{\pi_j}) 
    &= \advk{r_{l,S}}(\advk{\pi}(j)) + \advk{r_{l,A}}(\advk{\pi}(j),\pi[j]) \\
    &\qquad + \advk{r_{l,S}}(\advk{\pi}(j')) + \advk{r_{l,A}}(\advk{\pi}(j'),\advk{\pi}[j']) \\
    &= \advk{r_{l,S}}(s_j,k-j) + \advk{r_{l,A}}((s_j, k-j),(*a^j,\idle)) 
    + \advk{r_{l,S}}((s_j,a^j,k-j)) \\
    &\qquad + \advk{r_{l,A}}((s_j,a^j,k-j), (\idle,\idle, P_j^*))\\
    &= \advk{r_{l,S}}(s_j,k-j) + \advk{r_{l,A}}((s_j, k-j),(*a^j,\idle))\\
    &= r_{l,S}(s_j) \cdot \ind[E_j] + r_{l,A}(s,a^j) \cdot \ind[F_j]
\end{align*}
where
\begin{align*}
E_j &= \exists{h \in [\max(1,k-k_l), k]}.\, (s_j,k-j) = (s,h) \in S_h
\qquad \text{and}\\
F_j &= \exists{h \in [\max(1,k-k_l), k]}. \left\{
    \begin{aligned} 
        &(s_j,k-j) = (s,h) \in S_h \\
        &\land (*a_j,\idle) \in \advk{A}(\advk{s})
    \end{aligned}
    \right\}.
\end{align*}
Continuing,
\begin{align*}
    r_{l,S}(s_j) \cdot \ind[E_j]
    &= r_{l,S}(s_j) \cdot \ind\big[j\in [0,\min(k_l,k)] \big] \\
    r_{l,A}(s,a^j) \cdot \ind[F_j],
    &= r_{l,A}(s,a^j) \cdot \ind\left[ j\in [0,\min(k_l,k)] \land a^j\in A(s) \right].\\
\end{align*}
Combining and simplifying (since $\min(k_l,k)=k_l$),
\begin{align*}
    \advk{r_l}(\advk{\pi_j})
    &= r_{l,S}(s_j) \cdot \ind[j\in [0,k_l]]
    + r_{l,A}(s,a^j) \cdot \ind[ j\in [0,k_l] ] \\
    &= \left[r_{l,S}(\pi(j)) + r_{l,A}(\pi(j), \pi[j]) \right] \cdot \ind[j \leq k_l]\\
    &= r_l(\pi_j) \cdot \ind[j \leq k_l]
\end{align*}
\end{proof}

It follows from the path-index and reward-preservation results (\Cref{prop:finh-path-index-preserved,prop:nz-finh-rew-preserved}) that the valuation of objectives is likewise preserved. Subsequently, all other results from the nonzero-sum infinite-horizon setting carry over (except that we now consider time-varying strategies, as in the zero-sum finite-horizon case).






\section{Solving the Inner Problem}
\label{sec:infh-rvi-solve-inner}
As established in \Cref{lem:infh-nature-inner-problem}, in a zero-sum 2-player ICSG $\csg$, nature's inner problem to be solved for each state-action pair $(s,a)\in S\times A$ can be formulated as computing:
\[
P^*_{sa} :=\arg\inf_{P_{sa} \in \Punc_{sa}} \sum_{s' \in S} P_{sas'} \cdot V(s', X).
\]
In the nonzero-sum case, $V$ is replaced by $\Vsum$. The resulting optimal distribution $P^*_{sa}$ is then used in the value computation of $\csg$, which we detail in \Cref{sec:zs-mc-supp} (zero-sum) and \Cref{sec:nz-mc-supp} (nonzero-sum).

A similar inner problem is solved efficiently for IMDPs \cite{nilim2005robust}, IDTMCs \cite{katoen2012three} and $L_1$-MDPs \cite{strehl2008analysis} via a bisection algorithm, which is explicitly presented in \cite{suilen2024robust} for IMDPs. 
Although the algorithm has not been previously used for ICSGs, its extension to this setting is straightforward since the inner optimisation problem takes a similar form. 
Our adapted version of this algorithm is presented as \Cref{alg:icsg-inner-problem}, where we let $V(s) := \valg(s,n)$, and fix an $(s,a) \in S \times A$. For nonzero-sum ICSGs, we use $V(s) := \valg^1(s,n) + \valg^2(s,n)$.

\vspace{-0.4cm}
\resizebox{0.9\textwidth}{!}{%
\begin{minipage}{\textwidth}
\begin{algorithm}[H]
\caption{Solving the ICSG inner problem $\inf_{P \in \Punc_{sa}}$ (adapted from \cite{suilen2024robust})}
\label{alg:icsg-inner-problem}
\begin{algorithmic}[1]
    \Require $(s,a)$, current value estimates $[V(s)]_{s\in S}$, transition bounds $\Pcheck$, $\Phat$
    \State Sort $S^* = \{s_1^*, \ldots, s_{|S|}^*\}$ such that $V(s_i^*) \leq V(s_{i+1}^*)$ 
    \label{algline:inner-problem-sort}
    \State $\forall s_i^* \in S^*$: $P^*_{sa}(s_i^*) \gets 0$
    \State $budget \gets 1 - \sum_{s' \in S} \Pcheck_{sas'}$
    \State $i \gets 1$
    \While{$budget \ge \Phat_{s a s_i^*} - \Pcheck_{s a s_i^*}$}
        \State $P^*_{sa}(s_i^*) \gets \Phat_{s a s_i^*}$
        \State $budget \gets budget - (\Phat_{s a s_i^*} - \Pcheck_{s a s_i^*})$
        \State $i \gets i + 1$
    \EndWhile
    \State $P^*_{sa}(s_i^*) \gets \Pcheck_{s a s_i^*} + budget$
    \For{$j \in \{i+1, \ldots, |S|\}$}
        \State $P^*_{sa}(s_j^*) \gets \Pcheck_{s a s_i^*}$
    \EndFor
    \Return $P^*_{sa}$
\end{algorithmic}
\end{algorithm}
\end{minipage}
}
\vspace{0.1cm}

The algorithm runs in $O(|S| \log |S|)$ time due to the sorting step and is highly efficient in practice. 
Its correctness follows from the structure of the optimisation problem: minimising a linear objective under interval bounds and a simplex constraint, for which greedy allocation of the remaining probability mass to the lowest-valued successors yields an optimal solution \cite{bertsimas2005optimal}. 
Further details on the algorithm can be found in~\cite{suilen2024robust,strehl2008analysis}.

Note that \Cref{alg:icsg-inner-problem} assumes an \emph{adversarial} nature and solves $\inf_{P_{sa} \in \Punc_{sa}}$ at each RVI step. If we instead adopt a \emph{controlled} resolution of uncertainty (in favour of player 1), we would solve $\sup_{P_{sa} \in \Punc_{sa}}$ per iteration. This can be achieved by reversing the sort order at line \ref*{algline:inner-problem-sort}, which ensures that probability mass is greedily allocated to successor states with \emph{higher} value estimates.

\section{Value Computation for Zero-sum ICSGs}
\label{sec:zs-mc-supp}

\subsection{Infinite-horizon Properties}
\label{sec:mc-infh}

To ensure convergence, specifically for infinite-horizon reward properties, we retain a standard assumption for model checking standard CSGs \cite{kwiatkowska2021automatic}:
\begin{assumption}
\label{assmp:zs-mc}
    From any state $s$ where $r_S(s)<0$ or $r_A(s,a)<0$ for some action $a$, under all profiles of $\csg$, with probability 1 we reach either 
    a target state in $T$, or 
    a zero-reward state that cannot be left with probability 1 under all profiles.
\end{assumption}
We compute $n$-step reachability values which form a non-decreasing sequence converging to $\vag(s, X)$ as $n \to \infty$, i.e.,
\[
\lim_{n \rightarrow \infty}{\vag(s,\adv{X},n)} = \vag(s,\adv{X}) = \valg(s,X) \quad \text{(by \AbbrCref{cor:infh-adv-value-preservation})}.
\]
Convergence is estimated using the maximum relative change between successive value updates, as is typically done in value iteration for probabilistic model checking. However, this heuristic does not guarantee that the computed values lie within a specified error bound. 
As highlighted in~\cite{haddad2018interval}, this limitation applies even to simple models such as MDPs.

In what follows, we denote by $val(\Zi)$ the value of a matrix game $\Zi \in \reals^{|A_1(s)| \times |A_2(s)|}$, which can be computed by solving an LP problem \cite{v1928theorie,von1944theory}. For each state-action pair $(s,a)\in S\times A$, we compute the minimising next-state distribution $P^*_{sa}$ via \Cref{alg:icsg-inner-problem}.
Note that as in \cite{kwiatkowska2021automatic}, although we present the value computation recursively, they are implemented iteratively in practice.

\startpara{Probabilistic reachability}
\label{sec:rvi-prob-reachability}
To accelerate convergence, we pre-compute the set $U$ of states from which the target set $T$ cannot be reached, which can be done using standard graph algorithms~\cite{de2000graphalgo}. 
The $T$-states are assigned value $1$, and all other states are initialised to $0$. 
Below, we derive the value functions for this objective. 

We first apply the recursive value computations from \cite{kwiatkowska2021automatic} to $\adv{\csg}$, then use \Cref{def:infh-adv-resolution} and other results from \Cref{sec:infh-adv-resolution} to demonstrate how explicit construction of $\adv{\csg}$ can be avoided.

\paragraph{1) Value computation over $\adv{\csg}$.}

Based on \cite{kwiatkowska2021automatic}, the RVI update at step $n$ over $\adv{\csg}$ is:
\begin{equation}
\label{eq:PR-rvi-raw-value-fn}
\vag(\adv{s},n) = \begin{cases}
    1 & \ift \adv{s} \in T, \\
    0 & \elseift \adv{s} \in U \text{ or } n=0,\\
    val(\Zi) & \otherwiset.
\end{cases}
\end{equation}
Here, $\Zi \in \reals^{|\adv{A_1}(\adv{s})| \times |\adv{A_2}(\adv{s})|}$ is a matrix game defined by:
\begin{equation}
\label{eq:PR-rvi-raw-zij}
z_{i,j} = \sum_{s' \in \adv{S}} \adv{P_{\adv{s} (a_i, b_j) s'}} \cdot v_{n-1}^{s'}
\end{equation}
where $(a_i, b_j) \in \adv{A_1}(\adv{s}) \times \adv{A_2}(\adv{s})$ and $v_{n-1}^{s'} = \vag(s', n-1)$. 

\paragraph{2) Simplifying the recursion.}

We now expand on Case 3 in~\eqref{eq:PR-rvi-raw-value-fn} to simplify the value computation using the construction of $\adv{\csg}$ (see \Cref{def:infh-adv-resolution}).
Recall that the state space of $\adv{\csg}$ is given by $\adv{S} = S \cup S'$. In the following, we analyse the value update separately for players' states $S$ and nature's states $S'$.

First consider a players' state $s \in S$. Given a joint action $a = (a_1, a_2)$, $\adv{\csg}$ transitions deterministically to the nature's state $(s,a) \in S'$. Hence, at $s \in S$, the matrix game entries \eqref{eq:PR-rvi-raw-zij} may be written as:
\begin{equation}
\label{eq:PR-rvi-s-zij}
z_{i,j} = v_{n-1}^{(s,(a_i,b_j))} = \vag\big((s,(a_i,b_j)),n-1 \big).
\end{equation}

Next, consider a nature's state $(s,a) \in S'$. At these states, player 1 has a fixed dummy action $a_1 = \idle$, while player 2 (nature) deterministically selects a next-state distribution $P^*_{sa}$ to minimise the expected next-step value:
\[
P^*_{sa} = \arg\min_{P_{sa} \in \verts[\Punc_{sa}]} \sum_{s'\in S} P_{sas'} \cdot v_{n-1}^{s'}.
\]
This minimisation problem can be solved using \Cref{alg:icsg-inner-problem}, which is later detailed in \Cref{sec:infh-rvi-solve-inner}.

By definition, $P^*_{sa}$ is an optimal choice for player 2, so the action $(\idle, P^*_{sa})$ must achieve an NE in the matrix game $\Zi$. Since all successors of the state $(s,a)$ are in $S$, the value update simplifies to:
\begin{equation}
\label{eq:PR-rvi-sa-zij}
\vag \big((s,a),n \big) = \sum_{s'\in S} P^*_{sa}(s') \cdot v_{n-1}^{s'}.
\end{equation}

Substituting \eqref{eq:PR-rvi-sa-zij} into \eqref{eq:PR-rvi-s-zij}, we get the following \emph{two-step} recursion at~${s\in S}$:
\begin{equation}
\label{eq:PR-rvi-s-zij-final}
z_{i,j} = \sum_{s' \in S} P^*_{s(a_i,b_j)}(s') \cdot v_{n-2}^{s'}.
\end{equation}
This enables us to skip the iterations at states in $S'$, as the starting state is always an $S$ state. 

\paragraph{3) Final value function.} 
Combining the above gives the final value update~for~${s \in S}$:
\begin{equation}
\label{eq:PR-rvi-prob-final}
\valg(s, n) = \begin{cases}
    1 & \ift s \in T, \\
    0 & \elseift s \in U \text{ or } n=0,\\
    val(\Zi) & \otherwiset
\end{cases}
\end{equation}
where $\Zi$ has entries $z_{i,j} = \sum_{s' \in S} P^*_{s(a_i,b_j) s'} \cdot v_{n-1}^{s'}$.

\startpara{Expected reward reachability}
\label{sec:rvi-expected-reachability}
Following \cite{kwiatkowska2021automatic}, we first preprocess the game graph of $\csg$: Target states $T$ are made absorbing, while states from which $T$ is not reached almost surely (identifiable following \cite{de2000graphalgo}) are assigned infinite value. 
To ensure convergence in the presence of zero-reward cycles, we adopt the over-approximation method of \cite{chen2013automatic}: All zero rewards are replaced with a small constant $\gamma>0$, giving $r_\gamma=(r_A^\gamma,r_S^\gamma)$ so that every cycle accumulates positive reward. Note that larger $\gamma$ accelerates convergence but may reduce accuracy when refining values under the original reward structure $r$. 
We then apply RVI to compute the over-approximation, with the $n$th-step update at $s\in S$ being:
\begin{equation}
\label{eq:rvi-crr-final-value-fn}
\valg(s, n \mid r_\gamma) =
\begin{cases}
0 & \ift s \in T, \\
\infty & \elseift s \in U, \\
val(\Zi) & \otherwiset,
\end{cases}
\end{equation}
where $\Zi$ has entries
$z_{i,j} = r_\gamma(s,(a_i,b_j)) + \sum_{s'\in S} P^*_{s (a_i,b_j) s'} \cdot v_{n-1}^{s'}$.

Finally, the resulting upper bounds under $r_\gamma$ are used to initialise a second RVI under the original reward structure $r$, yielding $\valg(s,X\mid r)$.

\subsection{Finite-horizon Properties}
\label{sec:mc-finh}

By \Cref{cor:infh-adv-value-preservation}, we aim to compute $\valg(s,X) = \vag((s,k), \advk{X})$ for a $k$-step objective $X$.

\startpara{Bounded probabilistic reachability}
Derivation of the value function for this objective is similar to the infinite-horizon case, but requires an additional step: unifying the RBI iteration index $n$ with the time-step index $h$.

\paragraph{Time indexing.}
Recall from \Cref{sec:finh-adv-resolution} that a path in the original ICSG $\csg$ of length $k$ is indexed by discrete time-steps $h\in \{0, \dots, k\}$ (encoded in a $\advk{\csg}$-state). Meanwhile, a path in the time-augmented expanded game $\advk{\csg}$ interleaves players' states $S$ and nature's states $S'$ using the sequence $\{0, 0', 1, 1', \dots, (k-1)', k\}$, where the primed indices correspond to the $S'$-states.

To unify value updates across the two indexing systems, we define a ``flattened'' iteration index $n \in \{0, \dots, 2k\}$ for the value function $\vag(\advk{s}, \advk{X}, n)$, which relates to the time-step $h$ via 
\begin{equation}
\label{eq:rvi-h2n-index}
h(\advk{s},n) = \ceil{n/2}, 
\quad \text{or conversely} \quad 
n(h,\advk{s}) = 2h - \ind[\advk{s} \in S_h'].
\end{equation}
It follows that 
RBI can be executed directly over $\csg$, using the value function below, which is defined recursively over $n \in \{0, \dots, k\}$ for $s\in S$:
\label{sec:rbi-prob-reachability}
\begin{equation}
\label{eq:BR-value-fn-final}
\valg(s,n) =
\begin{cases}
    1 & \ift s \in T, \\
    0 & \elseift n = 0, \\
    val(\Zi) & \otherwiset
\end{cases}
\end{equation}
where $\Zi$ is defined with $z_{i,j} = \sum_{s' \in S} P^*_{s,(a_i,b_j)s',n} \cdot \valg(s', n-1)$.

\startpara{Bounded cumulative rewards}
\label{sec:rbi-cumulative-rewards}
For $n \in \{0, \dots, k\}$, $s\in S$:
\begin{equation}
\label{eq:CR-rbi-value-fn}
\valg(s,n\mid r, X) =
\begin{cases}
    0 & \ift n = 0, \\
    val(\Zi) & \otherwiset
\end{cases}
\end{equation}
with $z_{i,j} = r(s, (a_i, b_j)) + \sum_{s' \in S} P^*_{s (a_i, b_j) s', n} \cdot \valg(s', n-1)$.

\section{Value Computation for Nonzero-sum ICSGs}
\label{sec:nz-mc-supp}

\subsection{Infinite-horizon Properties}
\label{sec:nz-mc-infh}


Convergence is determined by comparing the value of $\Vsum$ across successive iterations. This is necessary because the individual values $V_1$ and $V_2$ may vary across different RSWNE profiles, but the maximal social welfare is uniquely defined, thus providing a stable metric for assessing convergence.

To derive the value function for nonzero-sum objectives, we proceed in two stages:
\begin{enumerate*}[label=\arabic*)]
\item we first apply the multi-player CSG model-checking algorithm from~\cite{kwiatkowska2020multi} to the 3-player adversarial expansion $\adv{\csg}$;
\item we then introduce two simplifications that reduce this to a variant of the nonzero-sum \emph{2-player} algorithm from~\cite{kwiatkowska2021automatic}, executed directly on the original ICSG~$\csg$.
\end{enumerate*}
\begin{enumerate}[nosep]
    \item \emph{Reducing trimatrix to bimatrix games:} 
    The 3-player game $\adv{\csg}$ exhibits a zero-sum structure between the coalition of players $\{1,2\}$ and player $\{3\}$ (nature). This enables us to effectively eliminate player 3 and reduce the \emph{trimatrix} game at each step to a general-sum \emph{bimatrix} game between player 1 and 2, which can be solved using the Lemke-Howson algorithm based on labelled polytopes \cite{lemke1964equilibrium} and a reduction to Satisfiability Modulo Theories (SMT) problems, as done in \cite{kwiatkowska2021automatic}.
    \item \emph{Avoiding construction of $\adv{\csg}$:} 
    Similar to the zero-sum setting (\Cref{sec:mc-infh}), we precompute nature’s optimal choice $P^*_{sa} \in \Punc_{sa}$ for each state–action pair $(s,a)\in S\times A$ using \Cref{alg:icsg-inner-problem}. This eliminates the need to compute values at the auxiliary states $S'$ in $\adv{\csg}$, allowing model checking to be performed directly on $\csg$.
\end{enumerate}

In what follows, we denote by $\Zi=(\Zi_1,\Zi_2) \in \reals^{|A_1(s)|\times |A_2(s)|}$ a bimatrix game with RSWNE values $\VRSWN(\Zi)$. 
For any set of states $S^*$, we define $\eta_{S^*}(s) := \ind[s \in S^*]$. Let $U_l$ be the set of states from which $T_l$ is unreachable. These sets can be precomputed using standard graph algorithms in e.g., \cite{de2000graphalgo}. 
Following \cite{kwiatkowska2020multi}, we maintain two sets of players during RVI: 
\begin{enumerate*}[label=\arabic*)]
    \item $D$, those who have reached their goals, and 
    \item $E$, those who can no longer reach their goals.
\end{enumerate*}
We emphasise that we adopt this set-based notation from \cite{kwiatkowska2020multi} for simplicity of representation; however, in the 2-player case, the resulting value functions are equivalent to that in~\cite{kwiatkowska2021automatic}.

\startpara{Probabilistic reachability}
\label{sec:nz-rvi-prob-reachability}
We compute $\valg(s,n) = \valg(s, \emptyset, \emptyset, n)$, defined over $s\in S$ and $D,E\subseteq N$ as follows:
\begin{align}
\label{eq:PR-nz-rvi-value-fn-final}
\valg(s,D,E,n) = \begin{cases}
    (\eta_D(1), \eta_D(2)) & \ift D \union E = N,\\
    (\eta_{T_1}(s), \eta_{T_2}(s)) & \elseift n=0,\\
    \valg(s, D\union D', E, n) & \elseift D' \neq \emptyset,\\
    \valg(s, D, E \union E', n) & \elseift E' \neq \emptyset,\\
    \VRSWN(\Zi) & \otherwiset
\end{cases}
\end{align}
where $D' = \{l\in N \exclude (D\union E) \mid s\in T_l \}$ and 
$E' = \{l\in N \exclude (D\union E) \mid s\in U_l \}$. 
When expressed in bimatrix form, 
$\Zi_l$ ($l\in\{1,2\}$) has entries $z_{ij}^l$ such that: if $l\in D$ then $z_{ij}^l =1$, else if $l \in E$ then $z_{ij}^l =0$, else $z_{ij}^l = \sum_{s'\in S}{P^*_{s(a_i,b_j) s'}\cdot \valg^l(s',D,E,n-1) }$.
Here, $P^*_{s(a_i,b_j)}$ is computed via \Cref{alg:icsg-inner-problem}.

\startpara{Expected reward reachability}
\label{sec:nz-rvi-exp-reachability}
We no longer track $E$, as their value would just be infinite for this objective. Let $\valg(s,n) = \valg(s, \emptyset, n)$ for any $s\in S$, with:
\begin{equation}
\label{eq:ER-nz-rvi-value-fn-final}
\valg(s,D,n) = \begin{cases}
    \nullb & \ift D = N\lor n=0,\\
    \valg(s, D\union D', n) & \elseift D' \neq \emptyset,\\
    \VRSWN(\Zi) & \otherwiset
\end{cases}
\end{equation}
where $D' = \{l\in N \exclude D \mid s\in T_l \}$, and each $\Zi_l$ ($l\in \{1,2\}$) has entries $z_{ij}^l = r_l(s,(a_i,b_j)) + \sum_{s'\in S}{P^*_{s(a_i,b_j) s'}\cdot \valg^l(s',D,n-1)}$ if $l\not\in D$, and $0$ otherwise.

\subsection{Finite-horizon Properties}
\label{sec:nz-mc-finh}

Let $k = \max(k_1, k_2)$ be the maximum of the players' horizons for objectives~$X=(X_1,X_2)$. 
In the following we compute $\valg(s)$ over $n\in [0,k]$ for states $s\in S$.

\startpara{Bounded probabilistic reachability}
\label{sec:nz-mc-finh-BR}
$\valg(s) = \valg(s, \emptyset, \emptyset, 0)$, where:
\begin{align}
\label{eq:BR-nz-rbi-value-fn-final}
\valg(s,D,E,n) &= \begin{cases}
    (\eta_D(1), \eta_D(2)) & \ift D \union E = N,\\
    \valg(s, D\union D', E, n) & \elseift D' \neq \emptyset,\\
    \valg(s, D, E \union E', n) & \elseift E' \neq \emptyset,\\
    \VRSWN(\Zi) & \otherwiset
\end{cases}
\end{align} 
with entries defined as follows:
\[
z_{ij}^l = 
\begin{cases}
    1 &\ift l\in D,\\
    0 &\elseift l \in E \lor k_l- n \leq 0,\\
    \sum_{s'\in S}{P^*_{s (a_i,b_j) s', n}\cdot \valg^l(s',D,E,n+1)} &\otherwiset.
\end{cases}
\]
Note that in the above, $n$ denotes the number of steps already taken, so $k_l - n$ represents the remaining horizon of player $l$'s objective $X_l$.

\startpara{Bounded cumulative rewards}
\label{sec:nz-mc-finh-CR}
$\valg(s) = \valg(s,0)$ with $\valg(s,n) = \VRSWN(\Zi)$ where 
\[
z_{ij}^l = 
\begin{cases}
    r_l(s,(a_i,b_j)) + \sum_{s'\in S}{P^*_{s (a_i,b_j) s', n}\cdot \valg^l(s',n+1)} &\ift k_l-n > 0,\\
    0 &\otherwiset.
\end{cases}
\]

\vspace{-0.2cm}
\section{Supplementary Experimental Results}
\label{sec:experiment-results-supp}

\Cref{tab:model-stats} reports model statistics for the benchmarks used, while \Cref{tab:mc-stats-full,tab:nz-mc-stats-full} present extended experimental results for zero-sum and nonzero-sum verification, respectively, covering a broader set of benchmarks than in the main paper.

The setup of \Cref{tab:mc-stats-full,tab:nz-mc-stats-full} is as follows. For each case study, we recorded the maximum and average number of actions per coalition in the matrix games solved during RVI or RBI, along with the total number of iterations performed. For infinite-horizon reachability properties of the form \texttt{F $\phi$} (``eventually $\phi$ is satisfied'')\footnote{See \cite{kwiatkowska2021automatic} for the interpretation of the properties expressed in rPATL.} where RVI is applied in two phases (see \Cref{sec:mc-infh}), iteration counts are reported for both phases. 
Further, following \cite{kwiatkowska2021automatic}, we report the verification timing divided into two components: 
\begin{enumerate*}[label=\arabic*)]
    \item \emph{qualitative} analysis (``Qual.''), which involves pre-computing the set of states from which the target is reached with probability 0 or 1, and is inapplicable to finite-horizon properties; and 
    \item \emph{quantitative} verification (``Quant.''), which includes the time spent solving NFGs.
\end{enumerate*}


\begin{table}[!t]
\renewcommand{\arraystretch}{0.95}
\centering
\caption{Model statistics for the (interval) CSG case studies. Model construction time is mostly the same for the CSGs and ICSGs.}
\label{tab:model-stats}

\begin{adjustbox}{width=0.9\textwidth}
\begin{tabular}{|c|r||c|r|r|r|r|}
\hline
\multicolumn{1}{|c|}{\textbf{Case study}} & 
\multicolumn{1}{c||}{\textbf{Param.}} & 
\multicolumn{1}{c|}{\textbf{Players}} & 
\multicolumn{1}{c|}{\textbf{States}} & 
\multicolumn{1}{c|}{\textbf{Trans.}} & 
\multicolumn{2}{c|}{\textbf{Constr. time (s)}} \\
\cline{6-7}
\textbf{[params]} & \textbf{values} & & & & \multicolumn{1}{c|}{\textbf{CSG}} & \multicolumn{1}{c|}{\textbf{ICSG}} \\
\hline

\multirow{5}{*}{\shortstack{\textbf{Robot coordination} \\ $[l]$}} 
 & 4  & \multirow{5}{*}{2} & 226     & 6,610       & 0.10 & 0.11 \\
 & 8  & & 3,970    & 201,650     & 0.47 & 0.55 \\
 & 12 & & 20,450   & 1,221,074  & 1.15 & 1.6 \\
 & 16 & & 65,026   & 4,198,450  & 3.65 & 3.86 \\
 & 24 & & 330,626  & 23,049,650  & 18.01 & 24.73 \\
\hline

\multirow{5}{*}{\shortstack{\textbf{Future markets} \\ $[months]$}} 
 & 6  & \multirow{5}{*}{3} & 33,338   & 143,487     & 1.148 & 0.839 \\
 & 12 & & 254,793  & 1,257,112 & 3.20 & 3.16 \\
 & 24 & & 826,617  & 4,315,864 & 11.66 & 12.58 \\
 & 36 & & 1,398,441 & 7,374,616    & 23.45 & 20.80 \\
 & 48 & & 1,970,265 & 10,433,368   & 27.80 & 29.60 \\
\hline

\multirow{4}{*}{\shortstack{\textbf{User centric network} \\ $[td,K]$}} 
 & 3  & \multirow{4}{*}{7} & 32,214   & 121,659     & 1.215 & 1.554 \\
 & 4  & & 104,897  & 433,764     & 2.982 & 2.957 \\
 & 5  & & 294,625  & 1,325,100    & 6.977 & 24.357 \\
 & 6  & & 714,849  & 3,465,558    & 17.891 & 24.132 \\
\hline

\multirow{4}{*}{\shortstack{\textbf{Aloha (deadline)} \\ $[b_{\max}, D]$}} 
 & 2  & \multirow{4}{*}{3} & 14,230   & 28,895      & 0.48 & 0.61 \\
 & 3  & & 72,566   & 181,438     & 1.25 & 1.31 \\
 & 4  & & 413,035  & 1,389,128    & 5.35 & 5.92 \\
 & 5  & & 2,237,981 & 9,561,201    & 26.51 & 29.13 \\
\hline

\multirow{4}{*}{\shortstack{\textbf{Aloha} \\ $[b_{\max}]$}} 
 & 2  & \multirow{4}{*}{3} & 14,230   & 28,895      & 0.46 & 0.48 \\
 & 3  & & 72,566   & 181,438     & 1.40 & 1.14 \\
 & 4  & & 413,035  & 1,389,128    & 4.89 & 5.70 \\
 & 5  & & 2,237,981 & 9,561,201    & 29.37 & 29.91 \\
\hline

\multirow{4}{*}{\shortstack{\textbf{Intrusion detection} \\ $[K,rounds]$}} 
 & 25  & \multirow{4}{*}{2} & 75 & 483 & 0.05 & 0.07 \\
 & 50  & & 150 & 983 & 0.09 & 0.10 \\
 & 100 & & 300 & 1,983 & 0.09 & 0.12 \\
 & 200 & & 600    & 3,983       & 0.10 & 0.12 \\
\hline

\multirow{4}{*}{\shortstack{\textbf{Jamming radio systems} \\ $[chans,slots]$}} 
 & 4,6 & \multirow{4}{*}{2} & 531 & 45,004 & 0.39 & 0.44 \\
 & 4,12 & & 1,623 & 174,796 & 0.77 & 0.90 \\
 & 6,6  & & 531 & 45,004 & 0.46 & 0.41 \\
 & 6,12 & & 1,623  & 174,796 & 0.76 & 0.87 \\
\hline

\multirow{4}{*}{\shortstack{\textbf{Medium access} \\ $[e_{\max}]$}} 
 & 10  & \multirow{4}{*}{3} & 10,591 & 135,915 & 1.17 & 0.80 \\
 & 15  & & 33,886 & 457,680 & 2.44 & 1.91 \\
 & 20  & & 78,181 & 1,083,645 & 4.53 & 4.91 \\
 & 25  & & 150,226 & 2,115,060 & 8.77 & 8.51 \\
\hline

\multirow{4}{*}{\shortstack{\textbf{Medium access} \\ $[e_{\max},s_{\max}]$}} 
 & 4,2 & \multirow{4}{*}{3} & 14,723 & 129,097 & 0.74 & 1.17 \\
 & 4,4  & & 18,751 & 147,441 & 0.87 & 1.22 \\
 & 6,4  & & 122,948 & 1,233,976 & 6.31 & 5.81 \\
 & 6,6  & & 138,916 & 1,315,860 & 5.36 & 7.21 \\
\hline

\multirow{4}{*}{\shortstack{\textbf{Power control} \\ $[e_{\max},pow_{\max}]$}} 
 & 40,8 & \multirow{4}{*}{2} & 32,812 & 260,924 & 0.57 & 0.88 \\
 & 40,16 & & 34,590 & 291,766 & 0.59 & 0.96 \\
 & 80,8 & & 193,396 & 1,469,896 & 2.24 & 2.38 \\
 & 80,16 & & 301,250 & 2,627,278 & 3.26 & 3.81 \\
\hline

\end{tabular}

\end{adjustbox}

\end{table}

\begin{table}[!t]
\centering
\caption{\centering Full verification statistics for \emph{zero-sum} CSGs and ICSGs. 
Values are shown to 2 decimal places, so smaller changes may be present but not visible at this precision.}
\label{tab:mc-stats-full}

\begin{adjustbox}{width=\textwidth}
\begin{tabular}{|c|r||c|r|r|r|r|r|r|r|r|}
\hline
\multicolumn{1}{|c|}{\textbf{Case study}} & 
\multicolumn{1}{c||}{\textbf{Param.}} & 
\multicolumn{1}{c|}{\textbf{Actions}} & 
\multicolumn{2}{c|}{\textbf{Val. Iters}} & 
\multicolumn{4}{c|}{\textbf{Verif. time (s)}} & 
\multicolumn{2}{c|}{\textbf{Value}} \\
\cline{6-9}
\multicolumn{1}{|c|}{\textbf{Property}} &
\multicolumn{1}{c||}{\textbf{values}} & 
\multicolumn{1}{c|}{\textbf{max/avg}} & 
\multicolumn{1}{c}{} & 
\multicolumn{1}{c|}{} & 
\multicolumn{2}{c|}{\textbf{Qual.}} & 
\multicolumn{2}{c|}{\textbf{Quant.}} & 
\multicolumn{1}{c}{} & 
\multicolumn{1}{c|}{} \\
\cline{4-11}
\multicolumn{1}{|c|}{\textbf{[params], $\epsilon$}} &
\multicolumn{1}{c||}{} & 
\multicolumn{1}{c|}{} & 
\multicolumn{1}{c|}{\textbf{CSG}} & 
\multicolumn{1}{c|}{\textbf{ICSG}} & 
\multicolumn{1}{c|}{\textbf{CSG}} & 
\multicolumn{1}{c|}{\textbf{ICSG}} & 
\multicolumn{1}{c|}{\textbf{CSG}} & 
\multicolumn{1}{c|}{\textbf{ICSG}} & 
\multicolumn{1}{c|}{\textbf{CSG}} & 
\multicolumn{1}{c|}{\textbf{ICSG}} \\
\hline \hline

\multirow{3}{*}{\shortstack{\textbf{Robot coordination} \\ $\llangle{rbt_1}\rrangle \Pt_{\max=?} [ \lnot{c} \until^{\leq k} g_1 ]$ \\ $[l,k], \ 0.01$}} 
 & 4,4 & 3,3/2.07,2.07 & 4 & 4 & n/a & n/a & 0.05 & 0.09 & 0.81 & 0.77 \\
 & 8,8 & 3,3/2.52,2.52 & 8 & 8 & n/a & n/a & 0.59 & 0.97 & 0.66 & 0.58 \\
 & 12,12 & 3,3/2.68,2.68 & 12 & 12 & n/a & n/a & 3.90 & 7.21 & 0.52 & 0.42 \\
\hline

\multirow{3}{*}{\shortstack{\textbf{Robot coordination} \\ $\llangle{rbt_1}\rrangle \Rt_{\min=?} [\eventually g_1]$ \\ $[l], \ 0.01$}} 
 & 4 & 3,3/2.07,2.07 & 10;9 & 9;9 & 0.03 & 0.04 & 0.14 & 0.22 & 4.55 & 4.39 \\
 & 8 & 3,3/2.52,2.52 & 15;14 & 15;14 & 0.43 & 0.61 & 2.44 & 3.54 & 8.89 & 8.63 \\
 & 12 & 3,3/2.68,2.68 & 20;19 & 19;18 & 1.64 & 3.27 & 16.11 & 31.59 & 13.15 & 12.84 \\
\hline

\multirow{5}{*}{\shortstack{\textbf{Future markets} \\ $\llangle{i_1}\rrangle \Rt_{\max=?}[\eventually c_1]$ \\ $[months], \ 0.01$}} 
 & 6 & 8,2/1.55,1.22 & 14 & 14 & 1.76 & 1.63 & 2.02 & 2.49 & 4.58 & 4.46 \\
 & 12 & 8,2/1.68,1.27 & 26 & 26 & 14.78 & 20.63 & 31.77 & 46.19 & 4.80 & 4.66 \\
 & 24 & 8,2/1.73,1.29 & 50 & 50 & 71.04 & 101.39 & 261.39 & 437.02 & 4.85 & 4.70 \\
 & 36 & 8,2/1.74,1.29 & 72 & 72 & 133.60 & 180.46 & 721.55 & 831.31 & 4.85 & 4.70 \\
 & 48 & 8,2/1.75,1.29 & 80 & 80 & 208.65 & 288.32 & 2181.12 & 2263.58 & 4.85 & 4.70 \\
\hline

\multirow{4}{*}{\shortstack{\textbf{User centric network} \\ $\llangle{usr}\rrangle  \Rt_{\min=?} [\eventually f]$ \\ $[K], \ 0.01$}} 
 & 3 & 16,8/2.11,1.91 & 45;15 & 44;15 & 0.83 & 1.40 & 788.73 & 833.36 & 0.04 & 0.03 \\
 & 4 & 16,8/2.31,1.92 & 60;21 & 60;21 & 3.52 & 3.82 & 3521.83 & 3725.29 & 4.00 & 4.00 \\
 & 5 & 16,8/2.46,1.94 & 79;25 & 79;23 & 15.64 & 21.51 & 12960.40 & 12789.04 & 4.04 & 4.03 \\
 & 6 & 16,8/2.60,1.96 & 98;31 & 98;29 & 64.67 & 97.75 & 39082.23 & 39676.62 & 7.00 & 7.00 \\
\hline

\multirow{4}{*}{\shortstack{\textbf{Aloha (deadline)} \\ $\llangle{u_2},{u_3}\rrangle \Pt_{\max=?}[\eventually (s_{2,3} \land t\leq D)]$ \\ $[b_{\max},D], \ 1/257$}} 
 & 2,8 & 2,4/1.00,1.00 & 24 & 24 & 0.83 & 1.06 & 1.22 & 1.30 & 0.94 & 0.94 \\
 & 3,8 & 2,4/1.00,1.00 & 23 & 23 & 4.58 & 4.12 & 1.84 & 2.65 & 0.92 & 0.91 \\
 & 4,8 & 2,4/1.00,1.00 & 23 & 23 & 32.22 & 49.47 & 5.31 & 9.27 & 0.92 & 0.91 \\
 & 5,8 & 2,4/1.00,1.00 & 23 & 23 & 294.96 & 440.13 & 21.28 & 37.11 & 0.92 & 0.91 \\
\hline

\multirow{4}{*}{\shortstack{\textbf{Aloha} \\ $\llangle{u_2},{u_3}\rrangle \Rt_{\min=?} [\eventually s_{2,3}]$ \\ $[b_{\max}], \ 1/257$}} 
 & 2 & 4,2/1.00,1.00 & 58;47 & 57;46 & 0.37 & 0.47 & 5.27 & 5.57 & 4.34 & 4.28 \\
 & 3 & 4,2/1.00,1.00 & 71;57 & 69;56 & 1.19 & 1.71 & 18.43 & 26.28 & 4.54 & 4.46 \\
 & 4 & 4,2/1.00,1.00 & 109;86 & 106;84 & 13.35 & 20.69 & 225.42 & 291.32 & 4.62 & 4.53 \\
 & 5 & 4,2/1.00,1.00 & 193;150 & 184;143 & 350.42 & 386.75 & 4667.43 & 4257.41 & 4.65 & 4.54 \\
\hline

\multirow{4}{*}{\shortstack{\textbf{Intrusion detection} \\ $\llangle{policy}\rrangle  \Rt_{\min=?} [\eventually (r={rounds})]$ \\ $[rounds], \ 0.01$}} 
 & 25 & 2,2/1.96,1.96 & 26;26 & 26;26 & 0.03 & 0.08 & 0.18 & 0.25 & 14.58 & 14.23 \\
 & 50 & 2,2/1.98,1.98 & 51;51 & 51;51 & 0.10 & 0.13 & 0.50 & 0.64 & 29.99 & 29.24 \\
 & 100 & 2,2/1.99,1.99 & 101;101 & 101;101 & 0.26 & 0.31 & 1.41 & 1.57 & 60.80 & 59.25 \\
 & 200 & 2,2/2.00,2.00 & 201;201 & 201;201 & 0.49 & 0.61 & 4.19 & 5.03 & 122.43 & 119.28 \\
\hline

\multirow{4}{*}{\shortstack{\textbf{Jamming radio systems} \\ $\llangle{u}\rrangle  \Pt_{\max=?} [ \eventually ({sent}\geq {slots}/2) ]$ \\ $[chans,slots], \ 0.01$}} 
 & 4,6 & 3,3/2.17,2.17 & 7 & 7 & 0.21 & 0.31 & 0.11 & 0.14 & 0.84 & 0.80 \\
 & 4,12 & 3,3/2.49,2.49 & 13 & 13 & 1.14 & 2.30 & 0.25 & 0.64 & 0.77 & 0.71 \\
 & 6,6 & 3,3/2.17,2.17 & 7 & 7 & 0.18 & 0.31 & 0.08 & 0.15 & 0.84 & 0.80 \\
 & 6,12 & 3,3/2.49,2.49 & 13 & 13 & 1.13 & 1.62 & 0.34 & 0.75 & 0.77 & 0.71 \\
\hline

\multirow{4}{*}{\shortstack{\textbf{Jamming radio systems} \\ $\llangle{u}\rrangle  \Rt_{\max=?} [ \eventually (t={slots}+1) ]$ \\ $[chans,slots], \ 0.01$}} 
 & 4,6 & 3,3/2.17,2.17 & 8;8 & 8;8 & 0.10 & 0.20 & 0.39 & 0.65 & 3.23 & 3.12 \\
 & 4,12 & 3,3/2.49,2.49 & 14;14 & 14;14 & 0.74 & 1.04 & 2.06 & 3.47 & 6.33 & 6.11 \\
 & 6,6 & 3,3/2.17,2.17 & 8;8 & 8;8 & 0.11 & 0.15 & 0.56 & 0.69 & 3.23 & 3.12 \\
 & 6,12 & 3,3/2.49,2.49 & 14;14 & 14;14 & 0.51 & 1.17 & 2.12 & 3.50 & 6.33 & 6.11 \\
\hline
\end{tabular}
\end{adjustbox}
\end{table}

\begin{table}[!t]
\renewcommand{\arraystretch}{1.2}
\centering
\caption{\centering Full verification statistics for \emph{nonzero-sum} CSGs and ICSGs verified within 6 hours, following the same setup as \Cref{tab:mc-stats-full}.}
\label{tab:nz-mc-stats-full}

\begin{adjustbox}{width=\textwidth}
\begin{tabular}{|c|r||c|r|r|r|r|r|r|r|r|}
\hline
\multicolumn{1}{|c|}{\textbf{Case study: [params], $\epsilon$}} & 
\multicolumn{1}{c||}{\textbf{Param.}} & 
\multicolumn{1}{c|}{\textbf{Actions}} & 
\multicolumn{2}{c|}{\textbf{Val. Iters}} & 
\multicolumn{4}{c|}{\textbf{Verif. time (s)}} & 
\multicolumn{2}{c|}{\textbf{Value}} \\
\cline{6-9}
\multicolumn{1}{|c|}{\textbf{Property}} &
\multicolumn{1}{c||}{\textbf{values}} & 
\multicolumn{1}{c|}{\textbf{max/avg}} & 
\multicolumn{1}{c}{} & 
\multicolumn{1}{c|}{} & 
\multicolumn{2}{c|}{\textbf{Qual.}} & 
\multicolumn{2}{c|}{\textbf{Quant.}} & 
\multicolumn{1}{c}{} & 
\multicolumn{1}{c|}{} \\
\cline{4-11}
\multicolumn{1}{|c|}{} &
\multicolumn{1}{c||}{} & 
\multicolumn{1}{c|}{} & 
\multicolumn{1}{c|}{\textbf{CSG}} & 
\multicolumn{1}{c|}{\textbf{ICSG}} & 
\multicolumn{1}{c|}{\textbf{CSG}} & 
\multicolumn{1}{c|}{\textbf{ICSG}} & 
\multicolumn{1}{c|}{\textbf{CSG}} & 
\multicolumn{1}{c|}{\textbf{ICSG}} & 
\multicolumn{1}{c|}{\textbf{CSG}} & 
\multicolumn{1}{c|}{\textbf{ICSG}} \\
\hline \hline

\multirow{3}{*}{\shortstack{\textbf{Robot coordination}: $[l,k], \ 0.01$ \\ $\llangle{r_1:r_2}\rrangle_{\max=?} \left( \Pt [ \lnot{c} \until^{\leq k} g_1 ] + \Pt [ \lnot{c} \until^{\leq k} g_2 ] \right)$}}
& 4,4 & 3,3/2.07,2.07 & 4 & 4 & 0.04 & 0.03 & 0.22 & 0.58 & 1.55 & 1.50 \\
& 8,8 & 3,3/2.52,2.52 & 8 & 8 & 0.26 & 0.71 & 0.77 & 54.03 & 0.92 & 0.84 \\
& 12,12 & 3,3/2.68,2.68 & 12 & 12 & 3.46 & 13.51 & 5.08 & 2882.09 & 0.49 & 0.40 \\
\hline

\multirow{3}{*}{\shortstack{\textbf{Robot coordination}: $[l,k_1,k_2], \ 0.01$ \\ $\llangle{r_1:r_2}\rrangle_{\max=?} \left( \Pt [ \lnot{c} \until^{\leq k_1} g_1 ] + \Pt [ \lnot{c} \until^{\leq k_2} g_2 ] \right)$}} 
& 4,4,6 & 3,3/2.07,2.07 & 6 & 6 & 0.02 & 0.03 & 0.09 & 0.42 & 1.96 & 1.95 \\
& 8,8,10 & 3,3/2.52,2.52 & 10 & 10 & 0.26 & 0.95 & 0.97 & 65.45 & 1.87 & 1.82 \\
& 12,12,14 & 3,3/2.68,2.68 & 14 & 14 & 3.61 & 17.33 & 4.42 & 4308.64 & 1.75 & 1.66 \\
\hline

\multirow{2}{*}{\shortstack{\textbf{Robot coordination}: $[l,k], \ 0.01$ \\ $\llangle{r_1:r_2}\rrangle_{\max=?} \left( \Pt [ \lnot{c} \until^{\leq k} g_1 ] + \Pt [ \lnot{c} \until\, g_2 ] \right)$}}
& 4,8 & 3,3/2.10,2.04 & 14 & 11 & 0.68 & 0.18 & 4.36 & 17.36 & 2.00 & 2.00 \\
& 4,16 & 3,3/2.12,2.05 & 14 & 11 & 0.23 & 0.34 & 1.86 & 74.89 & 2.00 & 2.00 \\
& 8,8 & 3,3/2.53,2.51 & 23 & 24 & 2.40 & 3.06 & 15.11 & 18432.63 & 1.52 & 1.45 \\
\hline

\multirow{2}{*}{\shortstack{
    \textbf{Robot coordination}: $[l], \ 0.01$ \\
    $\llangle{r_1:r_2}\rrangle_{\min=?} \left( \Rt [ \eventually g_1 ] + \Rt [ \eventually g_2 ] \right)$}} 
& \multicolumn{10}{c|}{\multirow{2}{*}{no RNE}} \\
& \multicolumn{10}{c|}{} \\
\hline

\multirow{2}{*}{\shortstack{\textbf{Future markets}: $[months], \ 0.01$ \\ $\llangle{i_1:i_2}\rrangle_{\max=?} \left( \Rt[\eventually c_1] + \Rt[\eventually c_2] \right)$}}
 & 3 & 2,2/1.12,1.12 & 6 & 6 & 0.33 & 0.04 & 0.89 & 1.74 & 10.00 & 9.86 \\
 & 6 & 2,2/1.20,1.20 & 12 & 12 & 0.28 & 0.37 & 5.11 & 224.33 & 10.25 & 10.19 \\
\hline

\multirow{3}{*}{\shortstack{\textbf{Aloha (deadline)}: $[b_{\max},D], \ 1/257$\\ $\llangle u_1:u_2,u_3 \rrangle_{\max=?} \left( \Pt[\eventually s_1] + \Pt[\eventually s_{2,3}] \right)$}} 
 & 1,8 & 2,4/1.00,1.01 & 23 & 23 & 0.12 & 0.13 & 0.33 & 5.00 & 1.99 & 1.99 \\
 & 2,8 & 2,4/1.00,1.00 & 23 & 23 & 0.32 & 0.30 & 0.85 & 107.33 & 1.98 & 1.97 \\
 & 3,8 & 2,4/1.00,1.00 & 22 & 22 & 0.70 & 1.15 & 3.25 & 3305.08 & 1.97 & 1.97 \\
\hline

\multirow{2}{*}{\shortstack{\textbf{Aloha}: $[b_{\max}], \ 1/257$ \\ $\llangle u_1:u_2,u_3 \rrangle_{\min=?} \left(\Rt[\eventually s_1] + \Rt[\eventually s_{2,3}] \right)$}} 
 & 2 & 2,4/1.00,1.00 & 55 & 57 & 0.04 & 0.07 & 1.16 & 46.98 & 1.00 & 1.00 \\
 & 3 & 2,4/1.00,1.00 & 66 & 68 & 0.21 & 0.26 & 3.59 & 1201.46 & 1.00 & 1.00  \\
\hline

\multirow{2}{*}{\shortstack{\textbf{Medium access}: $[e_{\max},k], \ 0.01$ \\ $\llangle p_1:p_2,p_3 \rrangle_{\max=?} \left(\Rt[\cumreward{k}] + \Rt[\cumreward{k}] \right)$}} 
 & 10,25 & 2,4/1.91,3.63 & 25 & 25 & n/a & n/a & 579.92 & 813.10 & 26.10 & 25.80 \\
 & 15,25 & 2,4/1.94,3.75 & 25 & 25 & n/a & n/a & 1493.55 & 6878.30 & 34.35 & 33.90 \\
\hline

\multirow{2}{*}{\shortstack{\textbf{Medium access}: $[e_{\max},k_1,k_2], \ 0.01$ \\ $\llangle p_1:p_2,p_3 \rrangle_{\max=?} \left(\Rt[\cumreward{k_1}] + \Rt[\cumreward{k_2}] \right)$}} 
 & 10,20,25 & 2,4/1.91,3.63 & 20 & 20 & 0.27 & 0.16 & 577.33 & 614.30 & 26.10 & 25.88 \\
 & 15,20,25 & 2,4/1.94,3.75 & 20 & 20 & 0.26 & 0.34 & 1147.76 & 6108.81 & 34.35 & 34.06 \\
\hline

\multirow{2}{*}{\shortstack{\textbf{Medium access}: $[e_{\max},s_{\max}], \ 0.01$ \\ $\llangle p_1:p_2,p_3 \rrangle_{\max=?} \left(\Pt[\eventually m_1] + \Pt[\eventually m_{2,3}] \right)$}} 
 & 4,2 & 2,4/1.70,2.88 & 10 & 10 & 0.50 & 0.38 & 37.23 & 199.41 & 2.00 & 2.00 \\
 & 4,4 & 2,4/1.64,2.70 & 12 & 12 & 0.93 & 0.47 & 26.40 & 392.65 & 1.48 & 1.48 \\
\hline

\multirow{2}{*}{\shortstack{\textbf{Power control}: $[e_{\max},pow_{\max}], \ 0.01$ \\ $\llangle p_1:p_2 \rrangle_{\max=?} \left(\Rt[\eventually e_1=0] + \Rt[\eventually e_2=0] \right)$}} 
 & 40,8 & 2,2/1.91,1.91 & 40 & 40 & 2.10 & 0.90 & 13.69 & 2613.72 & 13357.80 & 13354.41 \\
 & 40,16 & 2,2/1.95,1.95 & 40 & 40 & 1.99 & 0.75 & 11.68 & 3260.81 & 13357.80 & 13354.41 \\
\hline


\multirow{2}{*}{\shortstack{\textbf{Power control}: $[e_{\max},pow_{\max},k_1,k_2], \ 0.01$ \\ $\llangle p_1:p_2 \rrangle_{\max=?} \left(\Rt[\instreward{k_1}] + \Rt[\instreward{k_2}] \right)$}} 
 & 80,4,15,20 & 2,2/1.69,1.69 & 24 & 24 & 0.80 & 1.00 & 7.50 & 3400.22 & 2.00 & 2.00 \\
 & 80,4,20,20 & 2,2/1.69,1.69 & 24 & 24 & 0.87 & 1.70 & 7.61 & 3304.92 & 2.00 & 2.00 \\
\hline
\end{tabular}
\end{adjustbox}
\end{table}

\fi

\end{document}